%% file: main.tex
\documentclass[11pt]{article}
\input{macros}
\usepackage{enumitem}

\date{}
\makeatletter
\newcommand\notsotiny{\@setfontsize\notsotiny{9.5}{10.5}}
\makeatother

\Crefname{fact}{Fact}{Fact}
\crefname{fact}{fact}{fact}

\usepackage{adjustbox}
\allowdisplaybreaks
\title{Sharper Bounds for Chebyshev Moment Matching, with Applications}

\usepackage{makecell}
\author{
\begin{tabular}{cc}
\makecell{Cameron Musco \\ UMass Amherst \\ \texttt{cmusco@cs.umass.edu}} \hspace{4em}&\hspace{4em}  
\makecell{Christopher Musco \\ New York University \\\texttt{cmusco@nyu.edu}} \\\\
\makecell{Lucas Rosenblatt \\ New York University \\ \texttt{lucas.rosenblatt@nyu.edu}} \hspace{4em}&\hspace{4em}  
\makecell{Apoorv Vikram Singh \\ New York University \\ \texttt{apoorv.singh@nyu.edu}}
\end{tabular}
}

\input{bib_macros}

\begin{document}

\maketitle

\begin{abstract}
We study the problem of approximately recovering a probability distribution given noisy measurements of its Chebyshev polynomial moments. This problem arises broadly across algorithms, statistics, and machine learning. 
By leveraging a \emph{global decay bound} on the coefficients in the Chebyshev expansion of any Lipschitz function, we sharpen prior work, proving that accurate recovery in the Wasserstein distance is possible with more noise than previously known. Our result immediately yields a number of applications:

\vspace{.3em}

\begin{enumerate}[wide, labelwidth=0pt, labelindent=0pt,itemsep=.4em]
\item We give a simple ``linear query'' algorithm for constructing a differentially private synthetic data distribution with Wasserstein-$1$ error $\tilde{O}(1/n)$ based on a dataset of $n$ points in $[-1,1]$. This bound is optimal up to log factors, and matches a recent result of Boedihardjo, Strohmer, and Vershynin [Probab. Theory. Rel., 2024], which uses a more complex ``superregular random walk'' method. 
\item We give an $\tilde{O}(n^2/\epsilon)$ time algorithm for the linear algebraic problem of estimating the spectral density of an $n\times n$ symmetric matrix up to $\epsilon$ error in the Wasserstein distance. Our result accelerates prior methods from  Chen et al. [ICML 2021] and Braverman et al. [STOC 2022].

\item We tighten an analysis of Vinayak, Kong, Valiant, and Kakade [ICML 2019] on the maximum likelihood estimator for the statistical problem of ``Learning Populations of Parameters'', extending the parameter regime in which sample optimal results can be obtained. 
\end{enumerate}
\vspace{.5em}

Beyond these main results, we provide an extension of our bound to estimating distributions in $d > 1$ dimensions. We hope that these bounds will find applications more broadly to problems involving distribution recovery from noisy moment information.
\end{abstract}
\thispagestyle{empty}
\newpage
\setcounter{page}{1}

\section{Introduction}
\label{sec:intro}
The problem of recovering a probability distribution (or its parameters) by ``matching'' noisy estimates of the distribution's moments goes back over 100 years to the work of Chebyshev and Pearson \cite{Pearson:1894,Pearson:1936,Fischer:2011}. Moment matching continues to find a wide variety of applications, both in traditional statistical problems \cite{KalaiMoitraValiant:2010,MoitraValiant:2010,RabaniSchulmanSwamy:2014,WuYang:2019,WuYang:2020,FanLi:2023} and beyond. For example, moment matching is now widely used for solving eigenvalue estimation problems in numerical linear algebra and computational chemistry \cite{WeisseWelleinAlvermann:2006,Cohen-SteinerKongSohler:2018,ChenTrogdonUbaru:2021,Chen:2022}. 

One powerful and general result on moment matching for distributions with \emph{bounded support} is that the method directly leads to approximations with small error in the Wasserstein-$1$ distance (a.k.a. earth mover's distance). Concretely,
% a consequence of Jackson's theorem on polynomial approximation of Lipschitz functions is that, 
given a distribution $p$ supported on $[-1,1]$,\footnote{The result easily extends to $p$ supported on any finite interval by shifting and scaling the distribution to $[-1,1]$. For a general interval $[a,b]$, matching $k$ moments yields error $O(|a - b|/k)$ in the Wasserstein-$1$ distance.} any distribution $q$ for which $\E_{x\sim p}[x^i] = \E_{x\sim q}[x^i]$ for $i = 1, \ldots, k$ satisfies $W_1(p,q) = O(1/k)$, where $W_1$ denotes the Wasserstein-$1$ distance \cite{KongValiant:2017,ChenTrogdonUbaru:2021}. I.e., to compute an $\e$-accurate approximation to $p$, it suffices to compute $p$'s first $O(1/\e)$ moments and to return any distribution $q$ with the same moments. 

Unfortunately, the above result is extremely sensitive to {noise}, so is difficult to apply in the typical setting where, instead of $p$'s exact moments, we only have access to \emph{estimates} of the moments (e.g., computed from a sample). In particular, it can be shown that the moments need to be estimated to accuracy $O(1/2^k)$ if we want to approximate $p$ up to Wasserstein error $O(1/k)$ \cite{JinMuscoSidford:2023}. In other words, distribution approximation is \emph{poorly conditioned} with respect to the standard moments. 

\subsection{Chebyshev moment matching}
\label{sec:cheb_moment_match}
One way of avoiding the poor conditioning of moment matching is to move from the standard moments, $\E_{x\sim p}[x^i]$, to a better conditioned set of ``generalized'' moments. Specifically, significant prior work \cite{WeisseWelleinAlvermann:2006,wang2016differentially,BravermanKrishnanMusco:2022} leverages \emph{Chebyshev moments} of the form $\E_{x\sim p}[ T_i(x)]$, where $T_i$ is the $i^\text{th}$ Chebyshev polynomial of the first kind, defined as:
\begin{align*}
T_0(x) &= 1 & T_1(x) &= x &  T_i(x) &= 2xT_{i-1}(x) - T_{i-2}(x), \text{  for $i \geq 2$}. 
\end{align*}
The Chebyshev moments are known to be less noise sensitive than the standard moments: instead of exponentially small error, it has been shown that $\tilde{O}(1/k)$ error\footnote{Throughout, we let $\tilde{O}(z)$ denote $O(z\log^c(z))$ for constant $c$.} in computing $p$'s first $k$ Chebyshev moments suffices to find a distribution that is $O(1/k)$ close to $p$ in Wasserstein distance (see, e.g., Lemma 3.1 in \cite{BravermanKrishnanMusco:2022}). This fact has been leveraged to obtain efficient algorithms for distribution estimation in a variety of settings. %where accurate estimation of Chebyshev moments is possible. %\footnote{In statistical settings, moments are typically estimated via empirical moments, i.e. by $\frac{1}{n}\sum_{j=1}^n x_j^i$ for a sample $x_1, \ldots, x_n \sim p$. In this case, there is no advantage in moving to the Chebyshev moments since the empirical Chebyshev moments are just linear combinations of the empirical standard moments. They contain no more information about $p$. Our work instead focuses on applications where moment estimates \emph{are not obtained from a sample.}} 
For example, Chebyshev moment matching leads to $O(n^2/\poly(\e))$ time algorithms for estimating the eigenvalue distribution (i.e., the spectral density) of an $n\times n$ symmetric matrix $A$  to error $\e \|A\|_2$ in the Wasserstein distance \cite{BravermanKrishnanMusco:2022}. 
% The spectral density places mass $\frac{1}{n}$ at each of $A$'s $n$ eigenvalues. Its $i^\text{th}$ Chebyshev moment  is equal to the matrix trace $\tr(\frac{1}{n}T_i(A))$, which can estimated efficiently using stochastic trace estimation  \cite{Hutchinson:1990,AvronToledo:2011,MeyerMuscoMusco:2021}. 

Chebyshev moment matching has also been used for \emph{differentially private synthetic data generation}. In this setting, $p$ is uniform over a dataset $x_1, \ldots, x_n$. The goal is to find some $q$ that approximates $p$, but in a differentially private way, which informally means that $q$ cannot reveal too much information about any one data point, $x_j$ (see \Cref{sec:intro_applications} for more details) \cite{DworkNaorReingold:2009}.
A differentially private $q$ can be used to generate private synthetic data that is representative of the original data.
One approach to solving this problem is to compute $p$'s Chebyshev moments, and then add noise, which is known to ensure privacy \cite{DworkRoth:2014}. Then, one can find a distribution $q$ that matches the noised moments. It has been proven that, for a dataset of size $n$, this approach yields a differentially private distribution $q$ that is $\tilde{O}(1/n^{1/3})$ close to $p$ in Wasserstein distance \cite{wang2016differentially}.

\subsection{Our contributions} 
Despite the success of Chebyshev moment matching, including for the applications discussed above, there is room for improvement. For example, for private distribution estimation, alternative methods can achieve nearly-optimal error $\tilde{O}(1/n)$ in Wasserstein distance for a dataset of size $n$ \cite{boedihardjo2022private}, improving on the $\tilde{O}(1/n^{1/3})$ bound known for moment matching. For eigenvalue estimation, existing moment matching methods obtain an optimal quadratic dependence on the matrix dimension $n$, but a suboptimal polynomial dependence on the accuracy parameter, $\e$  \cite{BravermanKrishnanMusco:2022}.

The main contribution of this work is to resolve these gaps by proving a sharper bound on the accuracy to which the Chebyshev moments need to be approximated to recover a distribution to high accuracy in the Wasserstein distance.
% , and we hope will find applications more broadly.
Formally, we prove the following:
\begin{restatable}{theorem}{masterthm}
\label{thm:master_thm}
Let $p$ and $q$ be distributions supported on $[-1,1]$. For any positive integer $k$, if the distributions' first $k$ Chebyshev moments satisfy
\begin{align}
\label{eq:master_require_general}
    \sum_{j=1}^k \frac{1}{j^2}\left(\E_{x\sim p}T_j(x) - \E_{x\sim q}T_j(x)\right)^2 \leq \Gamma^2, 
\end{align}
then, for an absolute constant $c$\footnote{Concretely, we prove a bound of $\frac{36}{k} + \Gamma$, although we believe the constants can be improved, at least to $\frac{2\pi}{k} + \Gamma$, and possibly further. See \Cref{sec:gen} for more discussion.},\vspace{-.5em}
\begin{align}
\label{eq:master_ensure}
    W_1(p,q) \leq \frac{c}{k} + \Gamma.
\end{align}
As a special case, \eqref{eq:master_require_general} holds if for all $j \in \{1, \ldots, k\}$,
\begin{align}
\label{eq:master_require}
    \left | \E_{x\sim p}T_j(x) - \E_{x\sim q}T_j(x)\right| \leq \Gamma \cdot \sqrt{\frac{{j}}{{1 + \log k}}}.\footnotemark
\end{align}
\end{restatable}
\footnotetext{Throughout, we let $\log k$ denote the natural logarithm of $k$, i.e., the logarithm with base $e$.}
\Cref{thm:master_thm} characterizes the Chebyshev moment error required for a distribution $q$ to approximate $p$ in Wasserstein distance. The main requirement, \eqref{eq:master_require_general}, involves a weighted $\ell_2$ norm with weights $1/j^2$, which reflects the diminishing importance of higher moments on the Wasserstein distance. Referring to \eqref{eq:master_require}, we obtain a bound of $W_1(p,q) \leq O(1/ k)$ as long as $q$'s $j^\text{th}$ moment differs from $p$'s by $\tilde{O}(\sqrt{j}/k)$. In contrast, prior work requires error $\tilde{O}(1/k)$ \emph{for all of the first $k$ moments} to ensure the same Wasserstein distance bound (Lemma 3.1, \cite{BravermanKrishnanMusco:2022}). 

As a corollary of \Cref{thm:master_thm}, we obtain the following algorithmic result:
\begin{corollary}\label{corr:recovery}
Let $p$ be a distribution supported on $[-1,1]$. 
Given estimates $\hat{m}_1, \ldots, \hat{m}_k$ satisfying
$
    \sum_{j=1}^k \frac{1}{j^2}\left(\E_{x\sim p}T_j(x) - \hat{m}_j\right)^2 \leq \Gamma^2, 
$
\Cref{alg:weighted_regression} returns a distribution $q$ with $W_1(p,q) \leq c'\cdot \left(\frac{1}{k} + \Gamma\right)$ for a fixed constant $c'$, in $\poly(k)$ time.
\end{corollary}
\Cref{alg:weighted_regression} simply solves a linearly-constrained least-squares regression problem to find a distribution $q$ supported on a sufficiently fine grid whose moments match those of $p$ nearly as well as $\hat{m}_1, \ldots, \hat{m}_k$. \Cref{corr:recovery} follows by applying \Cref{thm:master_thm} to bound $W_1(p,q)$.
The linear constraints ensure that $q$ is positive and sums to one (i.e, that it is a distribution). This problem is easily solved using off-the-shelf software: in \Cref{sec:experiments} we implement our method using a solver from MOSEK \cite{mosek} and report some initial experimental results. 

Like prior work, our proof of \Cref{thm:master_thm} (given in \Cref{sec:gen}) relies on tools from polynomial approximation theory. In particular, we leverage a constructive version of Jackson's theorem on polynomial approximation of Lipschitz functions via ``damped Chebyshev expansions'' \cite{Jackson:1912}. 
Lipschitz functions are closely related to approximation in Wasserstein distance through the Kantorovich-Rubinstein duality: $W_1(p,q) = \max_{1-\text{Lip}\, f} \int_{-1}^1 f(x)(p(x) - q(x))dx$. 

In contrast to prior work, we couple Jackson's theorem with a tight ``global'' characterization of the coefficient decay in the Chebyshev expansion of a Lipschitz function. In particular, in \Cref{claim:sum_c_j_sq},  we prove that any $1$-Lipschitz function $f$ with Chebyshev expansion $f =  \sum_{j=0}^{\infty} c_j T_j$ has coefficients that satisfy $\sum_{j=1}^{\infty} j^2 c_j^2 = O(1)$. Prior work only leveraged the well-known ``local'' decay property, that the $j^\text{th}$ coefficient has magnitude bounded by $O(1/j)$ \cite{ThreftenBook}. This property is implied by our bound, but is much weaker. We believe that our new decay bound may be of independent interest given the ubiquitous use of Chebyshev expansions across computational science, statistics, and beyond.

\subsection{Applications}\label{sec:intro_applications}

We highlight three concrete applications of our main bounds,
\Cref{thm:master_thm} and \Cref{claim:sum_c_j_sq}, 
%can be immediately applied to a number of algorithms.
% Recall that to prove the \Cref{thm:master_thm}, we leveraged a constructive version of Jackson's theorem on polynomial approximation of Lipschitz functions, which we couple with a tight \emph{global} characterization of coefficient decay in the Chebyshev expansion of a Lipschitz function. In particular, we prove in \Cref{claim:sum_c_j_sq} that an $\ell$-Lipschitz function $f$ with Chebyshev expansion $f = \sum_{j=0}^{\infty} c_j T_j$ has coefficients that satisfy $\sum_{j=1}^{\infty} j^2 c_j^2 =\bigo{\ell^2}$. We believe that this is a key insight that can be leveraged in many applications in statistics, differential equations, approximation theory, where one approximates a function of interest by truncating the Chebyshev series or using the \emph{damped} truncated Chebyshev series. Then, in many cases, the aim is to bound the coefficients of the Chebyshev series of the function of interest. Previous works relied on the fact that either $c_j = \bigo{\ell/j}$, or $\sum_{j=1}^\infty c_j^2 = \bigo{\ell^2}$. In contrast, we proved something much stronger, i.e., $\sum_{j=1}^{\infty} j^2 c_j^2 =\bigo{\ell^2}$, which helps us achieve stronger bounds. 
%We highlight three concrete applications 
to algorithms for private synthetic data generation, spectral density estimation, and estimating populations of parameters. We suspect further applications exist.

\paragraph{Application 1: Differentially Private Synthetic Data.}
Privacy-enhancing technologies seek to protect individuals' data {without} preventing learning from the data. For theoretical guarantees of privacy, the industry standard is \textit{differential privacy}~\cite{DworkRoth:2014}, which is used in massive data products like the US Census, and is a core tenet of the recent Executive Order on the Safe, Secure, and Trustworthy Development and Use of Artificial Intelligence~\cite{biden2023executive, abowd2018us, abowd2019census}. 

Concretely, we are interested in the predominant notion of \emph{approximate differential privacy}:
\begin{definition}[Approximate Differential Privacy] \label{def:dp}
A randomized algorithm ${\mathcal {A}}$ is $(\e,\delta)$-differentially private if, for all pairs of neighboring datasets $X,X'$, and all subsets $\mathcal{B}$ of possible outputs:
\begin{align*}
    \Prob[{\mathcal {A}}(X)\in \mathcal{B}]\leq e^{\e} \cdot \Prob[{\mathcal {A}}(X')\in \mathcal{B}] + \delta. 
\end{align*}
\end{definition}
In our setting, a dataset $X$ is a collection of $n$ points in a bounded interval (without loss of generality, $[-1,1]$). Two datasets of size $n$ are considered ``neighboring'' if all of their data points are equal except for one. Intuitively, \Cref{def:dp} ensures that the output of $\mathcal {A}$ is statistically indistinguishable from the would-be output had any one individual's data been replaced with something arbitrary.

There exist differentially private algorithms for many statistical tasks \cite{ji2014differential, LiLyuSu:2017, mireshghallah2020privacy}. One task of primary importance is \textit{differentially private data synthesis}. Here, the goal is to generate \emph{synthetic data} that matches the original dataset along a set of relevant statistics or distributional properties. The appeal of private data synthesis is that, once generated, the synthetic data can be used for a wide variety of downstream tasks: a separate differentially private algorithm is not required for each potential use case. 
Many methods for private data synthesis have been proposed~\cite{hardt2010simple, zhang2017privbayes,RosenblattLiuPouyanfar:2020, liu2021iterative, abowd2019census, aydore2021differentially, rosenblatt2023epistemic,domingo2021limits}. Such methods offer strong empirical performance and a variety of theoretical guarantees, e.g., that the synthetic data can effectively answer a fixed set of data analysis queries with high accuracy \cite{hardt2010simple, McKennaMullinsSheldon:2022}. 
Recently, there has been interest in algorithms with more general distributional guarantees -- e.g., statistical distance guarantees between the synthetic data and the original data \cite{wang2016differentially,boedihardjo2022private,HeVershyninZhu:2023}. By leveraging \Cref{thm:master_thm}, we contribute the following result to this line of work:

\begin{restatable}{theorem}{synthdata}\label{thm:synth_data}
    Let $X = \{x_1, \ldots, x_n\}$ be a dataset with each $x_j\in [-1,1]$. Let $p$ be the uniform distribution on $X$. For any $\epsilon,\delta \in (0,1)$, there is an $(\e, \delta)$-differentially private algorithm based on Chebyshev moment matching that, in $O(n) + \poly(\e n)$ time, returns a distribution $q$ satisfying, for a fixed constant $c_1$,
    \vspace{-1.5em}
    \begin{align*}
        \E[W_1(p,q)] \leq c_1\frac{\log(\e n)\sqrt{\log(1/\delta)}}{\e n}.
    \end{align*}
    Moreover, for any $\beta \in (0,1/2)$, $W_1(p,q) \leq c_1\frac{\sqrt{\log(1/\beta) + \log(\e n)}\sqrt{\log(\e n)\log(1/\delta)}}{\e n}$ with probability $\geq 1-\beta$.
\end{restatable}
 
Theorem \ref{thm:synth_data} is proven in \Cref{{sec:synth_data}}. The returned  distribution $q$ is represented as a discrete distribution on $O(\e n)$ points in $[-1,1]$, so can be sampled from efficiently to produce a synthetic dataset of arbitrary size.
Typically, $\delta$ is chosen to be $1/\poly(n)$, in which case \Cref{thm:synth_data} essentially matches\footnote{Our result is for \textit{approximate} $(\epsilon,\delta)$-DP instead of exact $(\epsilon,0)$-DP. However, we obtain a very good $\sqrt{\log(1/\delta)}$ dependence on the approximation parameter $\delta$. Thus, we can set $\delta = 1/\text{poly}(n)$ and match the accuracy of \cite{boedihardjo2022private} up to constant factors. In our experience, approximate DP results where $\delta$ can be chosen to be a vanishingly small polynomial in $n$ are considered alongside exact DP results.} a recent result of Boedihardjo, Strohmer, and Vershynin \citep{boedihardjo2022private}, who give an $(\e,0)$-differentially private method with expected Wasserstein-$1$ error $O({\log^{3/2}(n)}/{(\e n)})$, which is optimal up to logarithmic factors.\footnote{An $\Omega(1/(\epsilon n))$ lower bound on the expected Wasserstein error holds via standard ``packing lower bounds'' which imply that even the easier problem of privately reporting the mean of a dataset supported on $[-1,1]$ requires error $\Omega(1/(\epsilon n))$. See e.g., \citep{gautam}, Theorem 3.} 
Like that method, we improve on a natural barrier of $\tilde{O}(1/(\e\sqrt{n}))$ error that is inherent to naive ``private histogram'' methods for approximation in the Wasserstein-$1$ distance \citep{xiao2010differential,qardaji2013understanding,xu2013differentially,DworkRoth:2014,zhang2016privtree,LiLyuSu:2017}. ``Private \textit{hierarchical} histogram'' methods can also be shown to match the Wasserstein-$1$ error of $\tilde{O}({1}/{(\e n)})$, albeit with worse polylog factors in $n$  \citep{hay2009boosting, ghazi2023differentially, FeldmanMcMillanSivakumar:2024}.

The result of \cite{boedihardjo2022private} introduces a ``superregular random walk'' that directly adds noise to $x_1,\ldots,x_n$ using a correlated distribution based on a Haar basis. Our method is simpler, more computationally efficient, and falls into the empirically popular \textit{Select, Measure, Project} framework for differentially private synthetic data generation \cite{vietri2022private, liu2021iterative, McKennaMullinsSheldon:2022}. In particular, as detailed in \Cref{alg:dp_algorithm}, we compute the Chebyshev moments of $p$, add independent noise to each moment using the standard Gaussian mechanism \cite{dwork2006our,McSherryMironov:2009}, and then recover $q$ matching these noisy moments. We verify the strong empirical performance of the method in \Cref{sec:experiments}.
A similar method  was analyzed in \cite{wang2016differentially}, although that work obtains a much weaker Wasserstein error bound of $\tilde{O}(1/(\e n^{1/3}))$. \Cref{thm:master_thm}'s tighter connection between Chebyshev moment estimation and distribution approximation allows us to obtain a significantly better dependence on $n$.

%\Cref{thm:synth_data} is proven in \Cref{sec:synth_data}. 

%We also generalize the result to higher dimensions in \Cref{sec:high_dim}, proving that we can obtain expected Wasserstein error $\tilde{O}(1/(\e n)^{1/d})$ \chris{should the $\epsilon$ actualy be in parenthesis?}, which matches prior work up to a logarithmic factor \cite{boedihardjo2022private}. Notably, our reduction is completely black-box: we show that such a result can be obtained given access to \emph{any} differentially private algorithm that achieves error $\tilde{O}(1/\e n)$ for one dimensional data. This reduction might be of independent interest.

%Finally, we
We note that \cite{HeVershyninZhu:2023} also claims a faster and simpler alternative to \cite{boedihardjo2022private}. While their simplest method achieves error $\tilde{O}(1/\sqrt{n})$, they describe a more complex method that matches our $\tilde{O}(1/n)$ result up to a $\log(n)$ factor. While we are not aware of an implementation of that algorithm, empirically comparing alternative synthetic data generators with Wasserstein distance guarantees would be a productive line of future work.  Additionally, we note that \citet{FeldmanMcMillanSivakumar:2024} recently study a stronger notion of \emph{instance optimal} private distribution estimation in the Wasserstein distance. It would be interesting to explore if moment matching has  applications in this setting. 

\paragraph{Application 2: Matrix Spectral Density Estimation.} Spectral density estimation (SDE) is a central problem in numerical linear algebra. In the standard version of the problem, we are given a symmetric $n\times n$ matrix $A$ with eigenvalues $\lambda_1 \geq \ldots \geq \lambda_n \in \R$. The goal is to output a distribution $q$ that is close in Wasserstein distance to the uniform distribution over these eigenvalues, $p$. %, ideally in runtime signficantly faster than is required to directly compute $p$ via a full eigendecomposition. %I.e., we want to find approximate eigenvalues that are close \emph{on average} to $A$'s true eigenvalues. 
%The SDE problem differs from other eigenvalue problems in that we are not interested in any information about $A$'s \emph{eigenvectors}.
An approximate spectral density can be useful in determining properties of $A$'s spectrum -- e.g., if its eigenvalues are decaying rapidly or if they follow a distribution characteristic of random matrices. 
Efficient SDE algorithms were originally studied in computational physics and chemistry, where they are used to compute the ``density of states'' of quantum systems \cite{Skilling:1989,SilverRoder:1994,MoldovanAndelkovicPeeters:2020}. More recently, the problem has found applications in network science \cite{DongBensonBindel:2019,Cohen-SteinerKongSohler:2018,JinKarmarkarMusco:2024}, deep learning \cite{CunKanterSolla:1991,PenningtonSchoenholzGanguli:2018,MahoneyMartin:2019,YaoGholamiKeutzer:2020}, optimization \cite{GhorbaniKrishnanXiao:2019}, and beyond \cite{LiXiErlandson:2019, ChenTrogdonUbaru:2022}. %The problem is intimately related to approximation in Wasserstien distance. 

Many popular SDE algorithms are based on Chebyshev  moment matching \cite{WeisseWelleinAlvermann:2006,BhattacharjeeJayaramMusco:2025,Chen:2024}. The $i^\text{th}$ Chebyshev moment of the spectral density is equal to $\E_{x\sim p}T_i(x) = \frac{1}{n} \sum_{j=1}^n  T_i(\lambda_j) = \tr(\frac{1}{n}T_i(A))$. Stochastic trace estimation methods such as Hutchinson's method can estimate this trace using a small number of matrix-vector products with $T_i(A)$ \cite{Hutchinson:1990,MeyerMuscoMusco:2021}. Since $T_i$ is a degree-$i$ polynomial, each matrix-vector product with $T_i(A)$ requires just $i$ products with $A$. Thus, with a small number of products with $A$, we can obtain approximate moments for use in estimating $p$. Importantly, this approach can be applied even in the common \emph{implicit} setting, where we do not have direct access to the entries of $A$, but can efficiently multiply the matrix by vectors \cite{AvronToledo:2011}. 

Recently, Braverman, Krishnan and Musco \cite{BravermanKrishnanMusco:2022} gave a theoretical analysis of Chebyshev moment-matching for SDE, along with the related Kernel Polynomial Method \cite{WeisseWelleinAlvermann:2006}. They show that, when $n$ is sufficiently large, specifically, $n = \tilde \Omega(1/\e^2)$, then $\tilde O(1/\e)$ matrix-vector products with $A$ (and $\poly(1/\e)$ additional runtime) suffice to output $q$ with $W_1(p,q) \le \e \norm{A}_2$, where $\norm{A}_2 = \max_i|\lambda_i|$ is $A$'s spectral norm. 
% This improves on the naive ``implicit method'' of simply recovering $A$ from $n$ matrix-vector products and running an explicit eigenvalue method, like the QR algorithm, to compute $A$'s spectral density. 

While the result of \cite{BravermanKrishnanMusco:2022} also holds for smaller values of $n$, it suffers from a polynomially worse $1/\e$ dependence in the number of matrix-vector products required. 
By leveraging \Cref{thm:master_thm}, we resolve this issue, showing that $\tilde O(1/\e)$ matrix-vector products suffice for \emph{any} $n$. 
Roughly, by weakening the requirements on how well we approximate $A$'s spectral moments, \Cref{thm:master_thm} allows us to decrease the accuracy with which moments are estimated, and thus reduce the number of matrix-vector products used by Hutchinson's estimator. Formally, in \Cref{sec:sde}, we prove:  
% \Cam{I think preceeding sentence is confusing. Can we just drop it? Chris I'll talk to you in person.} Moreover, if $A$ is a dense matrix, the algorithm runs in $\tilde O(n^2/\e + \poly(1/\e))$ time, which can be much faster than the $O(n^\omega)$ time required to compute $p$ directly via a full eigendecomposition. 
% \Cam{possible rewording of the last few sentences of the preceeding paragraph: The result of \cite{BravermanKrishnanMusco:2022} is nearly-optimal in terms of matrix-vector product complexity, an important notion of computational complexity in numerical linear algebra. In terms of runtime complexity, even when $A$ is a dense and unstructured matrix, their algorithm runs in $\tilde O(n^2/\e + \poly(1/\e))$ time, which can be much faster than the $O(n^\omega)$ time required to compute the spectral density $p$ directly via a full eigendecomposition.}
%
%Despite significant work in the area, theoretical bounds showing that popular methods achieve small error in the Wasserstein distance have only recent been established \cite{ChenTrogdonUbaru:2021,BravermanKrishnanMusco:2022}.% 
\begin{restatable}{theorem}{corsde}\label{cor:sde}
% For any symmetric matrix $A\in \R^{n\times n}$ with eigenvalues $\lambda_1, \ldots, \lambda_n \in \R$, there is an algorithm that computes $O\left(1/\e \cdot (1+\log^3(1/\e)/n\e)\right) = \tilde{O}\left(1/\e\right)$ matrix-vector products with $A$ and uses $O(n + \poly(1/\e))$ additional runtime to compute $\tilde{\lambda}_1, \ldots, \tilde{\lambda}_n$ satisfying \eqref{eq:sde_goal}.
There is an algorithm that, given $\e \in (0,1)$, symmetric $A\in \R^{n\times n}$ with spectral density $p$, and upper bound\footnote{The power method can compute $S$ satisfying $\|A\|_2 \leq S \leq 2\|A\|_2$ using $O(\log n)$ matrix-vector products with $A$ and $O(n)$ additional runtime \cite{KuczynskiWozniakowski:1992}. In some settings, an upper bound on $\|A\|_2$ may be known apriori \cite{JinKarmarkarMusco:2024}.} $S \geq \|A\|_2$,  uses $\tilde{O}\left(\frac{1}{\e}\right)$ matrix-vector products\footnote{Formally, we prove a bound of $\min\left\{n, O\left(\frac{1}{\e}\right) \cdot \left (1+\frac{\log^2(1/\e) \log^2(1/(\e \delta))}{n\e} \right )\right\}$ matrix-vector products to succeed with probability $1-\delta$. For constant $\delta$, this is at worst $O\left({\log^4(1/\e)}/{\e}\right)$, but actually $O({1}/{\e})$ for all $\e = \Omega(\log^4 n/n)$.} with $A$ and $\tilde{O}(n/\epsilon + 1/\e^{3})$ additional running time to output a distribution $q$ such that, with high probability, $W_1(p,q) \leq \e S$.
\end{restatable}
When $A$ is dense, \Cref{cor:sde} yields an  algorithm that runs in $\tilde O(n^2/\e + 1/\e^3)$ time, much faster than the $O(n^\omega)$ time required to compute $p$ directly via eigendecomposition.
In terms of matrix-vector products, the result cannot be improved by more than logarithmic factors. In particular, existing lower bounds for estimating the trace of a positive definite matrix \cite{MeyerMuscoMusco:2021,woodruff2022optimal} imply that $\Omega(1/\e)$ matrix-vector products with $A$ are necessary to approximate the spectral density $p$ up to error $\epsilon \|A\|_2$ (see \Cref{app:sde_lower_bound}). Thus, \Cref{cor:sde} essentially resolves the complexity of the SDE problem in the ``matrix-vector query model'' of computation, where cost is measured via matrix-vector products with $A$. This model has become central to theoretical work on numerical linear algebra, as it generalizes other important models like the matrix sketching and Krylov subspace models \cite{MeyerMuscoMusco:2021,SunWoodruffYangZhang:2021,woodruff2022optimal}. Our work contributes to recent progress on establishing tight upper and lower bounds for problems such as linear system solving \cite{BravermanHazanSimchowitz:2020}, eigenvector approximation \cite{MuscoMusco:2015,SimchowitzElAlaouiRecht:2018}, trace estimation \cite{JiangPhamWoodruffZhang:2024}, and more \cite{ChewiDeDiosPontLiLuNarayanan:2023,BakshiNarayanan:2023,AmselChenDumanKelesHalikiasMuscoMusco:2024,ChenDumanKelesHalikiasMuscoMusco:2024}. 

%\Cam{Would be nice to add sentence about application here.}
\paragraph{Application 3: Estimating Populations of Parameters.}
Our final application is to a classical statistical problem that has been studied since at least the 1960s \cite{Lord:1965,Lord:1969,Wood:1999}:
% In many cases, we have vast datasets that capture the characteristics of a large number of individuals. Although the number of individuals is large, the number of observations per individual can often be limited, making it impossible to estimate the parameter of interest for each individual. In such cases, the goal is to accurately estimate the distribution of the parameters of interest. The problem is formalized in \citet{vinayak19a} as follows:
\begin{restatable}[Population of Parameters Estimation]{problem}{popprob}\label{prob:popprob}
        Let $p$ be an unknown distribution over $[0,1]$. Consider a set of $N$ independent coins, each with unknown bias $p_i$ drawn from the distribution $p$. For each coin $i$, we observe the outcome of $t$ independent coin tosses $X_i \sim \text{Binomial}(t,p_i)$. The goal is find a distribution $q$ that is close to $p$ in Wasserstein-$1$ distance. 
\end{restatable}
%\Cam{I added this to give some motivation. Is it fine?}
Problem \ref{prob:popprob} is motivated by settings (medicine, sports, etc.) where we want to estimate the distribution of a parameter over a large population of $N$ individuals, but we only have noisy measurements of that parameter through a potentially much smaller number of observations, $t$, per individual.
A simple approach is to compute empirical estimates for $p_1, \ldots, p_N$ based on $X_1, \ldots, X_N$ and to return the resulting distribution of biases. Doing so achieves error $O(1/\sqrt{t} + 1/\sqrt{N})$ in Wasserstein distance. Interestingly, Tian, Kong, and Valiant \cite{TianKongValiant:2017} show that in the ``small sample'' regime when $N$ is large compared to $t$, it is possible to do much better. In particular, when $t = O(\log N)$, they introduce a moment-matching method with error $O(1/t)$.
% This is a classical statistical problem studied since at least 1960s \cite{Lord:1965,Lord:1969,Wood:1999} with significant additional work in the statistics community \cite{Wood:1999}. We refer the reader to \citet{vinayak19a,TianKongValiant:2017} for more information. 
% Recent work on this problem includes \citet{TianKongValiant:2017}, where they show that in the case when $N$ is sufficiently large but $t = \bigo{\log N}$ is small, i.e., the small sample regime,  moment matching approach yield an error of $\bigo{1/t}$ in Wasserstein-$1$ distance. In contrast, the empirical estimator would only yield an error of $\bigo{1/\sqrt{t}}$. 

More recently, \citet{vinayak19a} analyze the maximum likelihood estimator (MLE) for $p$. The MLE, which we denote by $\hat{p}_{\text{mle}}$, has a relatively simple form and can be computed efficiently. They prove that it matches the error of \cite{TianKongValiant:2017} in the small sample regime. Moreover, in the \emph{medium sample regime}, where $t = O(N^{2/9-\e})$ for any $\e > 0$, the MLE achieves error $\bigo{1/\sqrt{t \log N}}$, which is still an improvement on the empirical estimator. Formally, they prove the following theorem:

% \citet{vinayak19a} show that the maximum likelihood estimator (MLE) of $p$ is formulated as 
% \begin{equation*}
%     \hat{p}_{\text{mle}} \in \argmax_{Q \in \cD} \sum_{i=1}^N \log \int_0^1 \binom{t}{X_i} y^{X_i} (1-y)^{t-X_i} \rd Q(y) \mcom
% \end{equation*}
\begin{restatable}[{\citep[Theorem 3.2]{vinayak19a}}]{theorem}{thmvinayak} \label{thm:vinayak}
For any fixed constant $\e > 0$ and $t \in \Brac{\Omega(\log N), \bigo{N^{2/9-\e}}}$, with probability $99/100$, 
    \begin{equation}\label{eq:vinayak_bound}
        W_1(p, \hat{p}_{\text{mle}}) \leq O\paren{\frac{1}{\sqrt{t \log N}}}.
    \end{equation}
\end{restatable}
We are able to tighten this result by directly applying our new global bound on the Chebyshev coefficients of Lipschitz functions (\Cref{claim:sum_c_j_sq}).
% \citet{vinayak19a} use the dual of the Wasserstein-$1$ distance to bound the Wasserstein-$1$ between $p$ and $\hat{p}_{\text{mle}}$. 
% Let $f^*$ be the $1$-Lipschitz function, which witnesses the maximum Wasserstein-$1$ distance between the distributions. On a high level, they use Jackson's theorem to approximate $f^*$ by the ``damped Chebyshev expansion'' and then use the dual of the Wasserstein-$1$ distance to bound the Wasserstein-$1$ distance between $p$ and $\hat{p}_{\text{mle}}$. 
% Using our tight ``global'' characterization of the coefficient decay in the Chebyshev expansion of a Lipschitz function, 
In particular, in \Cref{sec:population_of_param}, we show how to increase the range of $t$ in \Cref{thm:vinayak} to $t \in \Brac{\Omega(\log N), \bigo{N^{1/4-\e}}}$. Moreover, \citet{vinayak19a} propose a simple conjecture that would improve their bound to $t \in \Brac{\Omega(\log N), \bigo{N^{2/3-\e}}}$. Combining our improvement with their conjecture would allow for $t \in \Brac{\Omega(\log N), \bigo{N^{1-\e}}}$, which is essentially optimal, as even if $t = \infty$ (i.e., we have access to the true parameter $p_i$), one cannot achieve an error better than $\bigo{1/\sqrt{N}}$ in Wasserstein distance (see \Cref{sec:population_of_param} for more details).

%\paragraph{Estimating Populations of Parameters}
%\chris{fill out with minimal explanation of this application}

%History: 
%\begin{itemize}
%    \item Early work on the problem: Studied at least since Lord in the 1960s \cite{Lord:1965,Lord:1969} with significant additional work in the statistics community: \cite{Wood:1999}. See \cite{vinayak19a} for more information.
%    \item Recent work before the paper we improve: \cite{TianKongValiant:2017}. This paper does use moment matching explicitly. Not sure it really matters though.
%\end{itemize}

\subsection{Extension to higher dimensions}
Finally, we note that we extend our main theorem (\Cref{thm:master_thm}), to arbitrary dimension $d > 1$ in \Cref{sec:multivariate_master}. Doing so requires two ingredients: 1) a high-dimensional generalization of our global Chebyshev coefficient decay bound, and 2) a constructive proof of Jackson's theorem in $d> 1$ dimensions, which shows that a \emph{damped} truncated multivariate Chebyshev series well-approximates any Lipschitz function.
% We also extend all our theorems to $d$ dimensions in  \Cref{sec:multivariate_master}. We give a constructive proof of the Jackson's theorem in \Cref{thm:jackson_high_dim} which shows that the \emph{damped} truncated multivariate Chebyshev series of a Lipschitz function is a good multivariate polynomial approximation to the Lipschitz function. Then, in \Cref{lem:cheb_coeff_high}, we generalize our global Chebyshev coefficient decay bound in $d$-dimensions. This sets the basis for proving the $d$-dimensional analogue of \Cref{thm:master_thm}, which is proved in \Cref{thm:master_thm_high}. 
As an application, we give an algorithm for differentially private synthetic data generation in $d>1$ dimensions in \Cref{sec:dp_high},
proving that we can obtain expected Wasserstein error $\tbigo{1/(\e n)^{1/d}}$, which matches prior work up to logarithmic factors \citep{boedihardjo2022private}.

%\cite{BravermanKrishnanMusco:2022} prove a weaker result that only achieves $\tilde{O}(1/\e)$ matrix-vector products for $n = \tilde{\Omega}(1/\e^2)$ and a worse dependence on $\e$ for comparatively smaller values of $n$. %$\tilde{O}(1/\e^2)$ matrix vector products for smaller values of $n$. On the other hand, \Cref{cor:sde} implies $\tilde{O}(1/\e)$ for all $n$. 

%\vspace{-.5em}
%\subsection{Roadmap}

%In \Cref{sec:prelim} we define notation and give technical preliminaries. In  \Cref{sec:gen} we prove our main result (\Cref{thm:master_thm}) on distribution recovery via Chebyshev moment approximation. In \Cref{sec:synth_data} we detail our application to private synthetic data generation. In \Cref{sec:sde} we detail our application to matrix spectral density estimation. Finally, in \Cref{sec:experiments} we give an empirical evaluation of our private synthetic data generation algorithm.

\section{Preliminaries}\label{sec:prelim}
Before our main analysis, we introduce notation and technical preliminaries. 

\paragraph{Notation.} We let $\N$ denote the natural numbers and $\Np$ denote the positive integers. For a vector $x\in \R^k$, we let $\|x\|_2 = \sqrt{\sum_{i=1}^k x_i^2}$ denote the Euclidean norm.
We often work with functions from $[-1,1] \rightarrow \R$. For two such functions, $f,g$, we use the convenient inner product notation:
\begin{align*}
\inprod{f, g} \defeq \int_{-1}^{1} f(x) g(x) \, dx.
\end{align*}
We will often work with products, quotients, sums, and differences of two functions $f,g$, which are denoted by $f\cdot g$, $f/g$, $f+g$, and $f-g$, respectively. E.g., $[f\cdot g](x) = f(x)g(x)$. For a function $f: [-1,1] \to \R$, we let $\|f\|_{\infty}$ denote $\|f\|_{\infty} = \max_{x\in [-1,1]} |f(x)|$ and $\|f\|_1 = \int_{-1}^1 |f(x)|\,dx$.

\begin{comment}
\paragraph{Wasserstein Distance.} This paper concerns the approximation of probability distributions in the Wasserstein-$1$ distance. We consider the standard version  involving the Euclidean metric:
\begin{definition}[Wasserstein-$1$ Distance, Euclidean Metric]\label{def:w1_transport}
 Let $p$ and $q$ be two distributions on $\R^{d}$. Let $Z(p,q)$ be the set of all couplings between $p$ and $q$, i.e., the set of distributions on $\R^{d} \times \R^{d}$ whose marginals equal $p$ and $q$. Then the Wasserstein-$1$ distance between $p$ and $q$ is: 
 \begin{align*}
 W_1(p,q) = \inf_{z\in Z(p,q)}\left[ \E_{(x,y) \sim z} \|x - y\|_2\right]. 
 \end{align*}
 % where $\|x-y\|_2$ denotes the Euclidean distance. For $d = 1$, we would simply have $|x - y|$.
\end{definition}
The Wasserstein-$1$ distance measures the total cost (in terms of distance per unit mass) required to ``transport'' the distribution $p$ to $q$. Alternatively, it has a well-known dual formulation:
\begin{fact}[Kantorovich-Rubinstein Duality]\label{fact:w1_dual}
Let $p,q$ be as in \Cref{def:w1_transport}. Then:
\begin{align*}
W_1(p,q) = \sup_{1\text{-Lipschitz } f} \E_{x\sim p} f(x) - \E_{y\sim q} f(y),
\end{align*}
where $f:\R^{d} \rightarrow \R$ is a Lipschitz function under the Euclidean metric.
\end{fact}
We use \Cref{fact:w1_dual} in the one dimensional setting. In this case, the supremum is over all functions $f:\R\rightarrow\R$ satisfying $|f(x) - f(y)| \leq |x-y|$ for all $x,y$.
\end{comment}

\paragraph{Wasserstein Distance.} This paper concerns the approximation of probability distributions in the Wasserstein-$1$ distance, which is defined below. Note that we only consider distributions supported on $[-1,1]$, but the definition generalizes to any distribution on $\R$ or $\R^d$.
\begin{definition}[Wasserstein-$1$ Distance]\label{def:w1_transport}
 Let $p$ and $q$ be two distributions on $[-1,1]$. Let $Z(p,q)$ be the set of all couplings between $p$ and $q$, i.e., the set of distributions on $[-1,1] \times [-1,1]$ whose marginals equal $p$ and $q$. Then the Wasserstein-$1$ distance between $p$ and $q$ is: 
 \begin{align*}
 W_1(p,q) = \inf_{z\in Z(p,q)}\left[ \E_{(x,y) \sim z} |x - y|\right]. 
 \end{align*}
 % where $\|x-y\|_2$ denotes the Euclidean distance. For $d = 1$, we would simply have $|x - y|$.
\end{definition}
The Wasserstein-$1$ distance measures the total cost (in terms of distance per unit mass) required to ``transport'' the distribution $p$ to $q$. Alternatively, it has a well-known dual formulation:
\begin{fact}[Kantorovich-Rubinstein Duality]\label{fact:w1_dual}
Let $p,q$ be as in \Cref{def:w1_transport}. Then $
W_1(p,q) =  \sup_{1\text{-Lipschitz } f} \langle f, p-q\rangle$,
where $f: \R \rightarrow \R$ is $1$-Lipschitz if $|f(x) - f(y)| \leq |x-y|$ for all $x,y \in \R$.
%where $f:\R^{d} \rightarrow \R$ is a Lipschitz function under the Euclidean metric.
\end{fact}
Above we slightly abuse notation and use $p$ and $q$ to denote (generalized) probability density functions\footnote{$p$ and $q$ might correspond to discrete distributions, in which case they will be sums of Dirac delta functions.} instead of the distributions themselves. We will do so throughout the paper.

%We use \Cref{fact:w1_dual} in the one dimensional setting. In this case, the supremum is over all functions $f:\R\rightarrow\R$ satisfying $|f(x) - f(y)| \leq |x-y|$ for all $x,y$. 
In our analysis, it will be convenient to work with functions that are smooth, i.e., that are infinitely differentiable. Since any Lipschitz function can be arbitrarily well approximated by a smooth function, we can do so when working with \Cref{fact:w1_dual}. In particular, for distributions on $[-1,1]$\footnote{Since $\|p - q\|_1 \leq 2$, if $\|f - \tilde{f}\|_{\infty} \leq \epsilon$ for some approximation $\tilde{f}$, then $\left |\langle f, p-q\rangle - \langle \tilde{f}, p-q\rangle\right| \leq 2\epsilon$. Since any Lipschitz function can be arbitrarily well-approximated by a smooth function in the $\ell_{\infty}$ norm, taking a $\sup$ over Lipschitz functions or smooth Lipschitz functions is therefore equivalent.} we have:
\begin{align}
\label{eq:our_w1}
    W_1(p,q) = \sup_{1\text{-Lipschitz, smooth } f} \langle f, p-q\rangle.
\end{align}

\paragraph{Chebyshev Polynomials and Chebyshev Series.}
Our main result analyzes the accuracy of (noisy) Chebyshev polynomial moment matching for distribution approximation. The Chebyshev polynomials are defined in \Cref{sec:cheb_moment_match}, and can alternatively be defined on $[-1,1]$ via the trigonometric definition, $T_j(\cos \theta) = \cos(j \theta)$. 
We use a few basic properties about these polynomials.

\begin{fact}[Boundedness and Orthogonality, see e.g. \citep{Hale2015}]\label{fact:cheb_prop}
The Chebyshev polynomials satisfy:
    \begin{enumerate}
        \item {\bf Boundedness:} $\forall x \in [-1,1]$ and $j \in \N$, $\abs{T_j(x)} \leq 1$.
        \item {\bf Orthogonality:} The Chebyshev polynomials are orthogonal with respect to the weight function $w(x) = \frac{1}{\sqrt{1 - x^2}}$. In particular, for $i, j \in \N$, $i\neq j$, $\inprod{T_i\cdot w, T_j} = 0$.
    \end{enumerate}
\end{fact}
To obtain an orthonormal basis we also define the \emph{normalized} Chebyshev polynomials as follows:

\begin{definition}[Normalized Chebyshev Polynomials] \label{def:norm_cheb}
    The $j^\text{th}$ \emph{normalized} Chebyshev polynomial, $\T_j$, is  defined as
$\T_j \defeq T_j / \sqrt{\inprod{T_j \cdot w, T_j}}$. Note that $\inprod{T_j \cdot w , T_j}$ equals $\pi$ for $j = 0$ and $\pi/2$ for $j \geq 1$.
\end{definition}
We define the \emph{Chebyshev series} of a function $f:[-1,1] \to \R$  as $\sum_{j=0}^{\infty} \inprod{f\cdot w, \T_j} \T_j.$
If $f$ is Lipschitz continuous then the Chebyshev series of $f$ converges absolutely and uniformly to $f$ {\citep[Theorem 3.1]{ThreftenBook}}. Throughout this paper, we will also write the Chebyshev series of generalized probability density functions, which could involve Dirac delta functions. This is standard in Fourier analysis, even though the Chebyshev series does not converge pointwise \cite{lighthill1958introduction}. Formally, any density $p$ can be replaced with a Lipschitz continuous density (which has a convergent Chebyshev series) that is arbitrarily close in Wasserstein distance and the same analysis goes through. 

% \begin{remark}[\avs{added this remark}]
%     In this paper, we use the convergence of Chebyshev series for derivatives of Lipschitz functions; its worth noting that derivatives of Lipschitz functions do not always have convergent Chebyshev series. However, we rely on the fact that any Lipschitz function can be arbitrarily well approximated by an infinitely differentiable Lipschitz function, and hence all the proofs in the paper are rigorous. 
% \end{remark}

%  \begin{remark}[Chebyshev Series of Discrete Distribution]
%     Throughout the paper, we write the Chebyshev series of a discrete distribution's probability density function, which does not necessarily converge point-wise. Formally, for a discrete distribution with probability density function $p$, there is a smooth probability density function $p_s$, and $W_1(p,p_s)$ is arbitrarily small. This can be proved by convolving $p$ with appropriate smooth \emph{test functions} or \emph{mollifiers} \citep{mollifier}.  
% \end{remark}

\section{Main Analysis}
\label{sec:gen}
In this section, we prove our main result, \Cref{thm:master_thm}, as well as \Cref{corr:recovery}. To do so, we require two main ingredients. The first is a constructive version of Jackson's theorem on polynomial approximation of Lipschitz functions \cite{Jackson:1930}. A modern proof can be found in \cite[Fact 3.2]{BravermanKrishnanMusco:2022}.
\begin{fact}[Jackson's Theorem \citep{Jackson:1930}] \label{fact:jackson}
  Let $f : [-1,1] \to \R$ be an $\ell$-Lipschitz function. Then, for any $k \in \Np$, there are $k+1$ constants $1 = b_k^0 > \ldots > b_k^k \geq 0$ such that the polynomial 
  $f_k = \sum_{j=0}^k b_k^j \cdot \inprod{ f \cdot w, \T_j }\cdot \T_j$ satisfies $\norm{f-f_k}_{\infty} \leq 18 \ell/k$.
\end{fact}
It is well-known that truncating the Chebyshev series of an $\ell$-Lipschitz function $f$ to $k$ terms leads to error $O(\log k\cdot \frac{\ell}{k})$ in the $\ell_\infty$ distance \cite{ThreftenBook}. The above version of Jackson's theorem improves this bound by a $\log k$ factor by instead using a \emph{damped} truncated Chebyshev series: each term in the series is multiplied by a positive scaling factor between $0$ and $1$. We will not need to compute these factors explicitly, but $b_k^i$ has a simple closed form (see \cite[Equation 12]{BravermanKrishnanMusco:2022}).

To bound the  Wasserstein distance between distributions $p,q$, we need to upper bound $\langle f, p-q\rangle$ for every $1$-Lipschitz $f$.
The value of \Cref{fact:jackson} is that this inner product is closely approximated by $\langle f_k, p-q\rangle$. Since $f_k$ is a damped Chebyshev series, this inner product can be decomposed as a difference between $p$ and $q$'s Chebyshev moments. Details will be shown in the proof of \Cref{thm:master_thm}.

The second ingredient we require is a stronger bound on the decay of the Chebyshev coefficients, $\inprod{ f \cdot w, \T_j }$, which appear in \Cref{fact:jackson}. In particular, we prove the following result:
\begin{lemma}[Global Chebyshev Coefficient Decay]\label{claim:sum_c_j_sq}
    Let $f : [-1,1] \to \R$ be an $\ell$-Lipschitz, smooth function, and let $c_j \defeq \inprod{f \cdot w,\T_j}$ for $j \in \N$. Then, $\sum_{j=1}^{\infty} (j c_j)^2  \leq \frac{\pi}{2} \ell^2$.
\end{lemma}
\Cref{claim:sum_c_j_sq} implies the well known fact that $c_j = O(\ell/j)$ for $j \geq 1$ \cite{Trefethen:2008}. However, it is a much stronger bound: if all we knew was that the Chebyshev coefficients are bounded by  $O(\ell/j)$, then $\sum_{j=1}^{\infty} (j c_j)^2$ could be unbounded, whereas we give a bound of $O(\ell^2)$. Informally, the implication is that not all coefficients can saturate the ``local'' $O(\ell/j)$ constraint at the same time, but rather obey a stronger global constraint, captured by a weighted $\ell_2$ norm of the coefficients.

\subsection{Proof of Theorem \ref{thm:master_thm}}
We prove \Cref{claim:sum_c_j_sq} in \Cref{sec:decay_proof}. Before doing so, we show how it implies \Cref{thm:master_thm}.

% \masterthm*
\begin{proof}[Proof of \Cref{thm:master_thm}]
By \eqref{eq:our_w1}, to bound $W_1(p,q)$, it suffices to bound $\langle{f, p-q\rangle}$ for any $1$-Lipschitz, smooth $f$. Let $f_k$ be the approximation to any such $f$ guaranteed by \Cref{fact:jackson}. We have:
\begin{align}
\label{eq:ip_split}
    \inprod{f, p-q} = \inprod{f_k, p-q} + \inprod{f-f_k, p-q} &\leq \inprod{f_k, p-q} + \|f - f_k\|_\infty\|p-q\|_1 \nonumber\\ &\leq \inprod{f_k, p-q} + \frac{36}{k}.
\end{align}
In the last step, we use that $\|f - f_k\|_\infty\leq 18/k$ by \Cref{fact:jackson}, and that $\|p-q\|_1 \leq \|p\|_1 + \|q\|_1 = 2$.
So, to bound $\inprod{f, p-q}$ we  turn our attention to bounding $\inprod{f_k, p-q}$. 

For technical reasons, we will assume from here on that $p$ and $q$ are supported on the interval $[-1+\delta,1-\delta]$ for arbitrarily small $\delta \rightarrow 0$. This is to avoid an issue with the Chebyshev weight function $w(x) = {1}/{\sqrt{1-x^2}}$ going to infinity at $x =-1,1$. The assumption is without loss of generality, since we can rescale the support of $p$ and $q$ by a $(1-\delta)$ factor, and the distributions' moments and Wasserstein distance change by an arbitrarily small factor as $\delta \rightarrow 0$.
% \begin{align*}
% f_k = \sum_{j=0}^k c_j \T_j,
% \end{align*}
% where each $c_j$ satisfies $|c_j| \leq \abs{\inprod{f \cdot w, \T_j}}$.

We proceed by writing the Chebyshev series of the function $(p-q)/w$:
\begin{align}
\label{eq:pq_cheby_series}
    \frac{p-q}{w} = \sum_{j=0}^\infty \inprod{\frac{p-q}{w} \cdot w, \T_j } \T_j = \sum_{j=0}^\infty \langle p-q,\T_j\rangle \cdot  \T_j = \sum_{j=1}^\infty \langle p-q,\T_j\rangle\cdot  \T_j. 
\end{align}
In the last step we use that both $p$ and $q$ are distributions so $\inprod{ p-q, \T_0} = 1/\pi-1/\pi = 0$.

Next, recall from \Cref{fact:jackson} that $f_k = \sum_{j=0}^k c_j' \T_j$, where each $c_j'$ satisfies $|c_j'| \leq |c_j|$ for $c_j \defeq \langle f \cdot w, \T_j\rangle$.
Using \eqref{eq:pq_cheby_series}, the fact that $\langle { \T_i\cdot w, \T_j}\rangle  = 0$ whenever $i\neq j$, and that $\langle { \T_j\cdot w, \T_j}\rangle = 1$
for all $j$, we have:
\begin{align*}
\inprod{ f_k, p-q} = \inprod{ f_k\cdot w, \frac{p-q}{w}} 
= \inprod{ \sum_{j=0}^k c_j' \T_j \cdot w, \sum_{j=1}^\infty \langle p-q,\T_j\rangle\T_j}
&= \sum_{j=1}^k c_j' \cdot \langle p-q,\T_j\rangle. 
\end{align*}
Via Cauchy-Schwarz inequality and our global decay bound from \Cref{claim:sum_c_j_sq}, we then have: 
\begin{align}
\label{eq:cauchy_schwarz}  
\inprod{ f_k, p-q}  = \sum_{j=1}^k j c_j' \cdot \frac{\langle p-q,\T_j\rangle}{j}
&\leq \left(\sum_{j=1}^k (jc_j')^2\right)^{1/2}\cdot \left( \sum_{j=1}^k \frac{1}{j^2} \langle p-q,\T_j\rangle^2\right)^{1/2} \nonumber\\
&\leq \left(\sum_{j=1}^k (j c_j)^2\right)^{1/2}\cdot \left( \sum_{j=1}^k \frac{1}{j^2} \langle p-q,\T_j\rangle^2\right)^{1/2} \nonumber\\
&\leq \sqrt{\pi/2}\left( \sum_{j=1}^k \frac{1}{j^2} \langle p-q,\T_j\rangle^2\right)^{1/2}.
\end{align}
Observing from \Cref{def:norm_cheb} that $\langle p-q,\T_j\rangle/\sqrt{\pi/2}$ is exactly the difference between the $j^\text{th}$ Chebyshev moments of $p$ and  $q$, we can apply the assumption of the theorem, \Cref{eq:master_require_general}, to upper bound \Cref{eq:cauchy_schwarz} by $\Gamma$. 

Plugging this bound into \Cref{eq:ip_split}, we conclude the main bound of \Cref{thm:master_thm}: 
\begin{align*}
W_1(p,q) = \sup_{1\text{-Lipschitz, smooth } f} \langle f, p-q\rangle \leq \Gamma + \frac{36}{k}.
\end{align*}
We note that the constants in the above bound can likely be improved. Notably, the 36 comes from multiplying the factor of 18 in \Cref{fact:jackson} by 2. As discussed in \cite[Appendix C.2]{BravermanKrishnanMusco:2022}, strong numerical evidence suggests that this 18 can be improved to $\pi$, leading to a bound of $\Gamma + \frac{2\pi}{k}$.

Finally, we comment on the special case in \Cref{eq:master_require}. If $\left | \E_{x\sim p}T_j(x) - \E_{x\sim q}T_j(x)\right|  = |\langle p-q,\T_j\rangle| / \sqrt{\pi/2} \leq \Gamma \cdot \sqrt{\frac{{j}}{{1 + \log k}}}$ for all $j$ then we have that
$
\sum_{j=1}^k \frac{1}{j^2} \langle p-q,T_j\rangle^2 \leq \frac{\Gamma^2}{{1 + \log k}} \sum_{j=1}^k \frac{1}{j} \leq \Gamma^2. 
$
\end{proof}

\subsection{Efficient recovery}
\label{sec:eff_recover}
The primary value of \Cref{thm:master_thm} for our applications is that, given sufficiently accurate estimates,  $\hat{m}_1, \ldots, \hat{m}_k$, of $p$'s Chebyshev moments, we can recover a distribution $q$ that is close in Wasserstein-$1$ distance to $p$, even if there is no distribution whose moments exactly equal $\hat{m}_1, \ldots, \hat{m}_k$. 

This claim is formalized in \Cref{corr:recovery}, whose proof is straightforward. We outline the main idea here.
Recall the condition of the corollary, that $\sum_{j=1}^k \frac{1}{j^2}\left(\hat{m}_j - \langle p, \T_j\rangle \right)^2 \leq \Gamma^2$. Now, suppose we could solve the optimization problem:
\begin{align*}
q^* = \argmin_{\text{distributions $q$ on } [-1,1]} \,\,\sum_{j=1}^k \frac{1}{j^2}\left(\hat{m}_j - \langle q, \T_j\rangle \right)^2.
\end{align*}
Then by triangle inequality we would have:
\begin{align}
\label{eq:triangle_inequal}
\left(\sum_{j=1}^k \frac{1}{j^2}\left(\langle p, \T_j\rangle - \langle q^*, \T_j\rangle \right)^2\right)^{1/2} &\leq \left(\sum_{j=1}^k \frac{1}{j^2}\left(\hat{m}_j - \langle q^*, \T_j\rangle \right)^2\right)^{1/2} + \left(\sum_{j=1}^k \frac{1}{j^2}\left(\hat{m}_j - \langle p, \T_j\rangle \right)^2\right)^{1/2} \nonumber \\ &\leq 2\left(\sum_{j=1}^k \frac{1}{j^2}\left(\hat{m}_j - \langle p, \T_j\rangle \right)^2\right)^{1/2} \leq 2 \Gamma.
\end{align}
It then follows immediately from \Cref{thm:master_thm} that $W_1(p,q^*) \leq O\left(\frac{1}{k} + \Gamma\right)$, as desired. 

The only catch with the argument above is that we cannot efficiently optimize over the entire set of distributions on $[-1,1]$. Instead, we have to optimize over a sufficiently fine discretization. Specifically, we consider discrete distributions on a finite grid, choosing the Chebyshev nodes (of the first kind) instead of a uniform grid because doing so yields a better approximation, and thus allows for a coarser grid. Concretely, \Cref{corr:recovery} is proven by analyzing \Cref{alg:weighted_regression}. 
The full analysis is given in \Cref{app:corr_proof}.

\begin{algorithm}[tb]\caption{Chebyshev Moment Regression} \label{alg:weighted_regression}
\begin{algorithmic}[1]
		\Require Estimates $\hat{m}_1, \ldots, \hat{m}_k$ for the first $k$ Chebyshev polynomial moments of a distribution $p$. 
		\Ensure A probability distribution ${q}$ approximating $p$.
		\State For $g = \lceil k^{1.5} \rceil$, let $\cC = \set{x_1,\dots,x_{g}}$ be the degree $g$ Chebyshev nodes. I.e., $x_i = \cos\paren{\frac{2i-1}{2g} \pi}$.
        \State Let ${q}_1,\dots,{q}_{g}$ solve the following optimization problem: 
        \begin{equation*}
            \begin{array}{ll@{}ll}
            \min_{ z_1,\dots, z_{g}} &  \displaystyle\sum_{j=1}^k \frac{1}{j^2} \paren{\hat m_j - \sum_{i=1}^g {z}_i T_j(x_i)}^2 \\
            \text{subject to} &\displaystyle \sum_{i=1}^{g}  z_i = 1 \text{ and } z_i \geq 0, \,\, \forall i \in \set{1,\ldots,g}.
            \end{array}
        \end{equation*}
        \State Return ${q} = \sum_{i=1}^{m} {q}_i \delta(x-x_i)$, where $\delta$ is the Dirac delta function.
\end{algorithmic}
\end{algorithm}

We note that the optimization problem solved by \Cref{alg:weighted_regression} is a simple linearly constrained quadratic program with $g = O(k^{1.5})$ variables and $O(k^{1.5})$ constraints, so can be solved to high accuracy in $\poly(k)$ time using a variety of methods \cite{YeTse:1989,KapoorVaidya:1986,AndersenRoosTerlaky:2003}. In practice, the problem can also be solved efficiently using first-order methods like projected gradient descent \cite{WrightRecht:2022}.

\subsection{Proof of Lemma \ref{claim:sum_c_j_sq}}\label{sec:decay_proof}

We conclude this section by proving \Cref{claim:sum_c_j_sq}, our global decay bound on the Chebyshev coefficients of a smooth, Lipschitz function, which was key in the proof of \Cref{thm:master_thm}. To do so we will leverage an expression for the derivatives of the Chebyshev polynomials of the first kind in terms of the Chebyshev polynomials \emph{of the second kind}, which can be defined by the recurrence
\begin{align*}
U_0(x) &= 1 & U_1(x) &= 2x &  U_i(x) &= 2xU_{i-1}(x) - U_{i-2}(x), \text{  for $i \geq 2$}. 
\end{align*}

We have the following standard facts (see e.g., \cite{Rivlin69}).
\begin{fact}[Chebyshev Polynomial Derivatives] \label{fact:tj_uj}
Let $T_j$ be the $j^{th}$ Chebyshev polynomial of the first kind, and $U_j$ be the $j^{th}$ Chebyshev polynomial of the second kind. Then, for $j \geq 1$, $T_j'(x) = j U_{j-1}(x)$.
\end{fact}

\begin{fact}[Orthogonality of Chebyshev polynomials of the second kind]\label{fact:2orth}
The Chebyshev polynomials of the second kind are orthogonal with respect to the weight function $u(x) = \sqrt{1-x^2}$. In particular,
    \[ \int_{-1}^{1} U_i(x) U_j(x) u(x)  \, dx = 
    \begin{cases}
       0, & \text{ for } i \neq j \\
        \frac{\pi}{2} , & \text{ for } i = j \mper
    \end{cases} \]
\end{fact}

With the above facts we can now prove \Cref{claim:sum_c_j_sq}.

\begin{proof}[Proof of \Cref{claim:sum_c_j_sq}]
Let $f$ be a smooth, $\ell$-Lipschitz function, with Chebyshev expansion $f(x) = \sum_{j=0}^{\infty} c_j \T_j = \frac{1}{\sqrt{\pi}} c_0 T_0 + \sum_{j=1}^{\infty} \sqrt{\frac{2}{\pi}} c_j T_j$. Using \Cref{fact:tj_uj}, we can write $f$'s derivative as:
\begin{align*}
    f'(x) = \sum_{j=1}^{\infty} \sqrt{\frac{2}{\pi}} c_j T_j'(x) = \sqrt{\frac{2}{\pi}} \sum_{j=1}^{\infty} j c_j U_{j-1}(x) \mper
\end{align*}
By the orthogonality property of \Cref{fact:2orth}, we then have that 
\begin{align*}
    \int_{-1}^{1} f'(x) f'(x) u(x)  \, dx = \frac{2}{\pi} \sum_{j=1}^\infty j^2 c_j^2 \frac{\pi}{2} = \sum_{j=1}^\infty j^2 c_j^2.
\end{align*}
Further, using that $f$ is $\ell$-Lipschitz and so $|f'(x)| \le \ell$, and that the weight function $u(x) = \sqrt{1-x^2}$ is non-negative, we can upper bound this sum by
\begin{align*}
    \sum_{j=1}^\infty j^2 c_j^2 = \int_{-1}^{1} f'(x) f'(x) u(x)  \, dx \leq \ell^2 \int_{-1}^{1} u(x)  \, dx = \frac{\pi \ell^2}{2}.
\end{align*}
This completes the proof of the lemma. We remark that this bound cannot be improved, as it holds with equality for the function $f(x) = x$.
%where in the above inequality, we use the fact that since $f$ is an $\ell$-Lipschitz function, $0 \leq (f'(x))^2 \leq \ell^2$, and that $u(x)$ is a non-negative function.
\end{proof}

\section{Private Synthetic Data Generation}
\label{sec:synth_data}
In this section, we present an application of our main result to differentially private synthetic data generation. We recall the setting from \Cref{sec:intro_applications}: we are given a dataset $X = \{x_1, \ldots, x_n\}$, where each $x_i \in [-1,1]$, and consider the distribution $p$ that is uniform on $X$. The goal is to design an $(\e,\delta)$-differentially private algorithm that returns a distribution $q$ that is close to $p$ in Wasserstein distance.
For the purpose of defining differential privacy (see Def. \ref{def:dp}), we consider the ``bounded'' notation of neighboring datasets, which applies to datasets of the same size \cite{kifer2011no}. Concretely, $X = \{x_1, \ldots, x_n\}$ and $X' = \{x_1', \ldots, x_n'\}$ are \emph{neighboring} if $x_i \neq x_i'$ for \emph{exactly one} value of $i$.\footnote{Although a bit tedious, our results can be extended to the ``unbounded'' notation of neighboring datasets, where $X$ and $X'$ might differ in size by one, i.e., because $X'$ is created by adding or removing a single data point from $X$.}

To solve this problem, we will compute the first $n$ Chebyshev moments of $p$, then add noise to those moments using the standard \emph{Gaussian mechanism}. Doing so ensures that the noised moments are $(\e, \delta)$-differentially private. We then post-process the noised moments (which does not impact privacy) by finding a distribution $q$ that matches the moments. The analysis of our approach follows directly from \Cref{thm:master_thm}, although we use a slightly different method for recovering $q$ than suggested in our general \Cref{alg:weighted_regression}: in the differential privacy setting, we are able to obtain a moderately faster algorithm that solves a regression problem involving $O(n)$ variables instead of $O(n^{1.5})$.

Before analyzing this approach, we introduce preliminaries necessary to apply the Gaussian mechanism. In particular, applying the mechanism  requires bounding the \emph{$\ell_2$ sensitivity} of the function mapping a distribution $p$ to its Chebyshev moments. This sensitivity is defined as follows:

\begin{definition}[$\ell_2$ Sensitivity]\label{def:sens}
Let $\mathcal{X}$ be some data domain (in our setting, $\mathcal{X} = [-1,1]^n$) and let $f: \mathcal{X} \to \mathbb{R}^k$ be a vector valued function. The $\ell_2$-sensitivity of $f$, $\Delta_{2,f}$,  is defined as:
\begin{align*}
    \Delta_{2,f} \defeq \underset{X,X'\in \mathcal{X}}{\max_{\text{neighboring datasets}}} \| f(X) - f(X') \|_2.
\end{align*}
\end{definition}
The Gaussian mechanism provides a way of privately evaluating any function $f$ with bounded $\ell_2$ sensitivity by adding a random Gaussian vector with appropriate variance. Let $\mathcal{N}({0}, \sigma^2 {I}_k)$ denote a vector of $k$ i.i.d. mean zero Gaussians with variance $\sigma^2$.  We have the following well-known result:

\begin{fact}[Gaussian Mechanism \cite{dwork2006our,DworkRoth:2014}]\label{def:gaussian_mechanism}
Let $f : \mathcal{X} \to \mathbb{R}^k$ be a function with $\ell_2$-sensitivity $\Delta_{2,f}$ and let $\sigma^2 = \Delta_{2,f}^2 \cdot 2\ln(1.25 / \delta) / \e^2$, where $\e,\delta \in (0,1)$ are privacy parameters. Then the mechanism $\mathcal{M} = f(X) + \eta$, where $\eta \sim \mathcal{N}({0}, \sigma^2 {I}_k)$ is $(\e, \delta)$-differentially private.
\end{fact}

\begin{algorithm}[tb]\caption{Private Chebyshev Moment Matching} \label{alg:dp_algorithm}
\begin{algorithmic}[1]
		\Require Dataset $x_1, \ldots, x_n \in [-1,1]$, privacy parameters $\epsilon, \delta > 0$. 
		\Ensure A probability distribution $q$ approximating the uniform distribution, $p$, on $x_1, \ldots, x_n$. 
        \State Let $\mathcal{G} = \{-1,-1 + \frac{1}{\lceil \e n\rceil},-1 + \frac{2}{\lceil \e n\rceil}, \ldots, 1\}$. Let $r \defeq |\mathcal{G}| = 2\ceil{\e n}+1$ and let ${g}_i = -1 + \frac{i-1}{\ceil{\e n}}$ denote the $i^\text{th}$ element of $\mathcal{G}$.
		\State For $i = 1, \ldots, n$, let $\tilde{x}_i = \argmin_{y\in \mathcal{G}} |x_i - y|$. I.e., round $x_i$ to the nearest multiple of $1/\lceil \e n\rceil$.
        \State Set $\sigma^2 = \frac{\frac{16}{\pi}(1 + \log {k})\ln(1.25 / \delta)}{\e^2 n^2}$.
        \State Set $k = \ceil{2\e n}$.\footnotemark\, For $j=1, \ldots, k$, let $\hat{m}_j = \eta_j +  \frac{1}{n}\sum_{i=1}^{n} \T_j(\tilde{x}_i)$, where $\eta_j \sim \mathcal{N}(0,j\sigma^2)$. 
        \State Let ${q}_0,\dots,{q}_{r}$ be the solution to the following optimization problem: 
        \begin{equation*}
            \begin{array}{ll@{}ll}
            \min_{ z_1,\dots, z_{r}} &  \displaystyle\sum_{j=1}^{k} \frac{1}{j^2} \paren{\hat m_j - \sum_{i=1}^{r} {z}_i T_j(g_i)}^2 \\
            \text{subject to} &\displaystyle \sum_{i=1}^{r}  z_i = 1 \text{ and } z_i \geq 0, \,\, \forall i \in \set{1,\ldots,r}.
            \end{array}
        \end{equation*}
        \State Return ${q} = \sum_{i=1}^{r} {q}_i \delta(x-g_i)$, where $\delta$ is the Dirac delta function.
\end{algorithmic}
\end{algorithm}

We are now ready to prove the main result of this section, \Cref{thm:synth_data}, which follows by analyzing \Cref{alg:dp_algorithm}. Note that \Cref{alg:dp_algorithm} is very similar to \Cref{alg:weighted_regression}, but we first round our distribution to be supported on a uniform grid, $\mathcal{G}$. Doing so will allow us to solve our moment regression problem over the same grid, which is smaller than the set of Chebyshev nodes used in \Cref{alg:weighted_regression}.
\begin{proof}[Proof of \Cref{thm:synth_data}]
    We analyze both the privacy and accuracy of \Cref{alg:dp_algorithm}. 
    \paragraph{Privacy.} For a dataset $X = \{x_1, \ldots, x_n\} \in [-1,1]^n$, let $f(X)$ be a vector-valued function mapping to the first $k = \ceil{2 \e n}$ ({as set in \Cref{alg:dp_algorithm}}) \emph{scaled} Chebyshev moments of the uniform distribution over $X$. I.e.,
    \begin{align*}
        f(X) = \begin{bmatrix}
        1\cdot \frac{1}{n}\sum_{i=1}^n \T_1(x_i) \\
        \frac{1}{\sqrt{2}}\cdot \frac{1}{n}\sum_{i=1}^n \T_2(x_i)\\
        \vdots\\
        \frac{1}{\sqrt{k}}\cdot \frac{1}{n}\sum_{i=1}^n \T_{k}(x_i)
        \end{bmatrix}
    \end{align*}
    By \Cref{fact:cheb_prop}, $\max_{x_i \in [-1,1]} |\T_j(x_i)| \leq \sqrt{2/\pi}$ for $j \in \Np$, so we have: 
    \begin{align}
    \Delta_{2,f}^2 = \underset{X,X'\in \mathcal{X}}{\max_{\text{neighboring datasets}}} \| f(X) - f(X') \|_2^2 \leq \sum_{j=1}^{k} \frac{1}{jn^2}\cdot \frac{8}{\pi} \leq \frac{8}{\pi n^2}(1 + \log k).\label{eq:sensitivity}
    \end{align}
    For two neighboring datasets $X,X'$, let $\tilde X$ and $\tilde X'$ be the rounded datasets computed in line 2 of \Cref{alg:dp_algorithm} -- i.e., $\tilde{X} = \{\tilde{x}_1, \ldots, \tilde{x}_{n}\}$. Observe that $\tilde X$ and $\tilde X'$ are also neighboring. Thus, it follows from \Cref{def:gaussian_mechanism} and the sensitivity bound of \cref{eq:sensitivity} that $\tilde{m} = f(\tilde{X}) + \eta$ is $(\e,\delta)$-differentially private for $\eta \sim \mathcal{N}(0,\sigma^2 I_k)$ as long as $\sigma^2 = \frac{16}{\pi}(1 + \log k)\ln(1.25 / \delta) / (n^2\e^2)$. Finally, observe that $\hat{m}_j$ computed by \Cref{alg:dp_algorithm} is exactly equal to $\sqrt{j}$ times the $j^\text{th}$ entry of such an $\tilde{m}$. So $\hat{m}_1, \ldots, \hat{m}_{k}$ are $(\e,\delta)$-differentially private. Since the remainder of \Cref{alg:dp_algorithm} simply post-processes $\hat{m}_1, \ldots, \hat{m}_{k}$ without returning to the original data $X$, the output of the algorithm is also $(\e,\delta)$-differentially private, as desired.

    \footnotetext{While we choose $k = \ceil{2\e n}$ by  default, any choice of $k = \ceil{c\e n}$ for constant $c$ suffices to obtain the bound of \Cref{thm:synth_data}. Similarly, the grid spacing in $\mathcal{G}$ can made finer or coarse by a multiplicative constant. A larger $k$ or a finer grid will lead to a slightly more accurate result at the cost of a slower algorithm. We chose defaults so that any error introduced from the grid and choice of $k$ is swamped by error incurred  from the noise added in Line 4. I.e., the error cannot be improved by more than a factor of two with difference choices. See the proof of \Cref{thm:synth_data} for more details.}

    \paragraph{Accuracy.}
    \Cref{alg:dp_algorithm} begins by rounding the dataset $X$ so that every data point is a multiple of $1/\ceil{\e n}$. Let $\tilde{p}$ be the uniform distribution over the rounded dataset $\tilde{X}$. Using the transportation definition of the Wasserstein-$1$ distance, we obtain the bound:
    \begin{align}
    \label{eq:rounding_error}
    W_1(p,\tilde{p}) \leq \frac{1}{2\ceil{\e n}}. 
    \end{align}
    In particular, we can transport $p$ to $\tilde{p}$ by moving every unit of $1/n$ probability mass a distance of at most $1/2\ceil{\e n}$. Given \eqref{eq:rounding_error}, it will suffice to show that \Cref{alg:dp_algorithm} returns a distribution $q$ that is close in Wasserstein distance to $\tilde{p}$. We will then apply triangle inequality to bound $W_1(p,q)$.
    
    To show that \Cref{alg:dp_algorithm} returns a distribution $q$ that is close to $\tilde{p}$ in Wasserstein distance, we begin by bounding the moment estimation error: 
    \begin{align*}
        E \defeq \sum_{j=1}^{k} \frac{1}{j^2} \left(\hat{m}_j(p) - \langle \tilde{p}, T_j \rangle\right)^2, 
    \end{align*}
    where $k$ is as chosen in \Cref{alg:dp_algorithm} and $\langle \tilde{p}, T_j \rangle = \frac{1}{n}\sum_{i=1}^{n} T_j(\tilde{x}_i)$. Let $\sigma^2$ and $\eta_1, \ldots, \eta_{k}$ be as in \Cref{alg:dp_algorithm}. Applying linearity of expectation, we have that:
    \begin{align}
    \label{eq:basic_exp_bound}
    \E[E] = \E\left[\sum_{j=1}^{k} \frac{1}{j^2} \eta_j^2\right] = \sum_{j=1}^{k} \frac{1}{j^2} \E\left[\eta_j^2\right] = \sum_{j=1}^{k} \frac{1}{j^2}\cdot j\sigma^2 \leq (1+\log k)\sigma^2.
    \end{align}
    % Let $\gamma_j =(\eta_j/j) $. Then, $E = \sum_{j=1}^{n} \gamma_j^2$. From \Cref{lem:sub_exp_of_gau_sq} \avs{wrote this lemma in appendix, but can also cite \cite{wainwright_2019}}, for $j \in [n]$, $\gamma_j^2$ is a sub-exponential random variable with parameter $\paren{2\sigma^2/j, 4\sigma^2/j}$. Applying \Cref{fact:sub_exp_conc}, we get that $E - \E[E]$ is a sub-exponential random variable with parameter $(\nu_{*}, \alpha_{*})$, where, $\nu_* = \sqrt{\sum_{j=1}^{n}{4 \sigma^4}/{j^2}} \leq {2 \pi \sigma^2}/{\sqrt{6}}$, $\alpha_{*} = 4 \sigma^2$.
    % Therefore, from \Cref{eq:basic_exp_bound} and choosing $t = 8 \sigma^2 \log n$ in \Cref{fact:sub_exp_conc}, we get that
    % \begin{align}
    % \label{eq:whp_exp_bound}
    % \Pr{ E  \geq 10 (\log n) \sigma^2 } \leq  \frac{1}{n} \mper
    % \end{align}
    
    Now, let $q$ be as in \Cref{alg:dp_algorithm}. Using a triangle inequality argument as in \Cref{sec:eff_recover}, we have:
    \begin{align*}
        \Gamma^2  = \sum_{j=1}^{k} \frac{1}{j^2} \left(\langle {q}, T_j \rangle - \langle \tilde{p}, T_j \rangle\right)^2 \leq \sum_{j=1}^{k} \frac{1}{j^2} \left(\langle {q}, T_j \rangle - \hat{m}_j\right)^2 + \sum_{j=1}^{k} \frac{1}{j^2} \left(\langle \tilde{p}, T_j \rangle - \hat{m}_j\right)^2 \leq 2E.
    \end{align*}
    Above we use that $\tilde{p}$ is a feasible solution to the optimization problem solved in \Cref{alg:dp_algorithm} and, since $q$ is the optimum,  $\sum_{j=1}^{k} \frac{1}{j^2} \left(\langle {q}, T_j \rangle - \hat{m}_j\right)^2 \leq \sum_{j=1}^{k} \frac{1}{j^2} \left(\langle \tilde{p}, T_j \rangle - \hat{m}_j\right)^2$. 
    It follows that $\E[\Gamma^2] \leq 2\E[E]$, and, via Jensen's inequality, that $\E[\Gamma] \leq \sqrt{2\E[E]}$. Plugging into \Cref{thm:master_thm}, we have for constant $c$: 
    \begin{align}
    \label{eq:final_triangle}
    \E[W_1(\tilde{p},q)] \leq \E[\Gamma] + \frac{c}{k} & \leq \sqrt{2 (1+\log {k}) \sigma^2} + \frac{c}{k} = \bigo{\frac{\log (\e n) \sqrt{\log({1/\delta})}}{\e n}} \mper
    \end{align}
    By triangle inequality and \eqref{eq:rounding_error}, $W_1(p,q) \leq W_1(\tilde{p},q) + W_1(\tilde{p},p) \leq W_1(\tilde{p},q) + \frac{1}{2\ceil{\e n}}$. Combined with the bound above, this proves the accuracy claim of the theorem. 

    Recall from \Cref{sec:gen} that the constant $c$ in \Cref{thm:master_thm} is bounded by $36$, but can likely be replaced by $2\pi$, in which case it can be checked that the $\frac{c}{k}$ term in \eqref{eq:final_triangle} will be dominated by the $\sqrt{2 (1+\log {k}) \sigma^2}$ term for our default of $k = \ceil{2\e n}$ in \Cref{alg:dp_algorithm}. However, any choice $k = \Theta(\epsilon n)$ suffices to prove the theorem. We also remark that our bound on the expected value of $W_1(\tilde{p},q)$ can also be shown to hold with high probability. See \Cref{app:concentration} for details.

    We conclude by noting that, as in our analysis of \Cref{alg:weighted_regression} (see \Cref{sec:eff_recover}), \Cref{alg:dp_algorithm} requires solving a linearly constrained quadratic program with $r = 2\ceil{\e n} + 1$ variables and ${r}+1$  constraints, which can be done to high accuracy in $\poly(\e n)$ time. 
    \end{proof}

\subsection{Empirical Evaluation for Private Synthetic Data}\label{sec:experiments}

In this section, we empirically evaluate the application of our main result to differentially private synthetic data generation, as presented in \Cref{sec:synth_data}. Specifically, we implement the procedure given in \Cref{alg:dp_algorithm}, which produces an $(\e, \delta)$-differentially private distribution $q$ that approximates the uniform distribution, $p$, over a given dataset $X = x_1, \ldots, x_n \in [-1,1]$. We solve the linearly constrained least squares problem from \Cref{alg:dp_algorithm} using an interior-point method from MOSEK \cite{cvxpy,mosek,AndersenRoosTerlaky:2003}. 
We evaluate the error $W_1(p,q)$ achieved by the procedure on
both real world data and data generated from known probability density functions (PDFs), with a focus on how the error scales with the number of data points, $n$.

\begin{figure}[!hb]
\centering
\includegraphics[width=\textwidth]{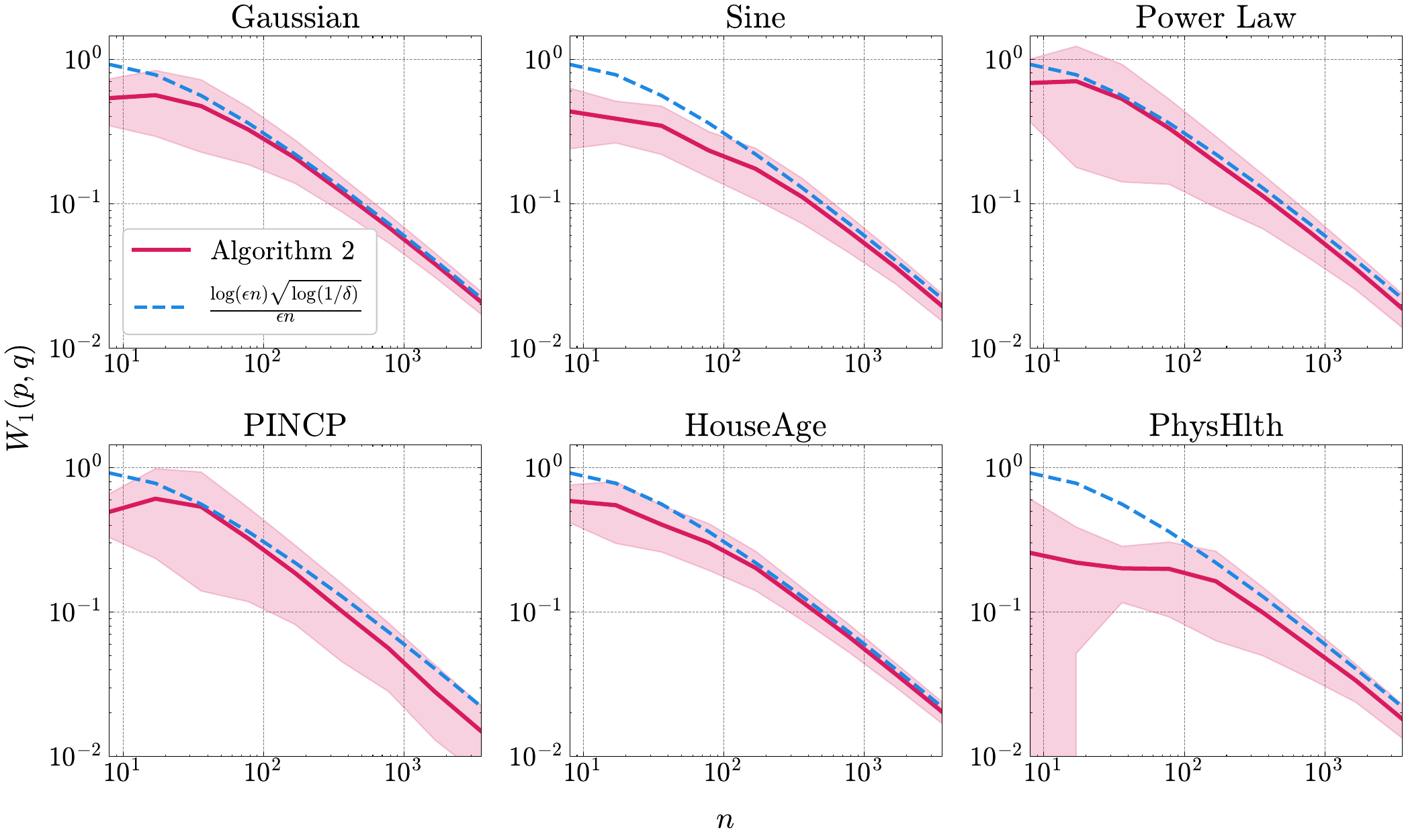}
\vspace{-1.5em}
\caption{
Experimental validation of \Cref{alg:dp_algorithm} for private synthetic data. For each dataset, we collect subsamples of size $n$ for different $n$. We plot the $W_1$ distance between the uniform distribution, $p$, over the subsample and a differentially private approximation, $q$, constructed by \Cref{alg:dp_algorithm} with privacy parameters $\epsilon = 0.5$ and $\delta = 1/n^2$. As predicted by \Cref{thm:synth_data}, the Wasserstein-$1$ error scales as $\tilde{O}(1/n)$. The solid red line shows the mean of $W_1(p,q)$ over 10 trials, while the shaded region plots one standard deviation around the mean (the empirical variance across trials). The blue dotted line plots the theoretical bound of \Cref{thm:synth_data}, without any leading constant.
% \chris{I think we should plot $\log(\e n)\sqrt{\log(1/\delta)}/\e n$ for our chosen values of $\epsilon, \delta$. And put that in the legend.}
}
% \chris{did we specify what the shaded region means?}}
\label{fig:experiments}
\end{figure}

For real world data, we first consider the American Community Survey (ACS) data from the Folktables repository \cite{ding2021retiring}. We use the 2018 ACS 1-Year data for the state of New York; we give results for the \texttt{PINCP} (personal income) column from this data. We also consider the California Housing dataset \cite{pace1997sparse}; we give results for the \texttt{HouseAge} (median house age in district) column, from this data. Finally, we consider the CDC Diabetes Health Indicators dataset \cite{kaggleDiabetesHealth, uciMachineLearning}; we give results for the \texttt{PhysHlth} (number of physically unhealthy days) from this data. For each of these data sets, we collect uniform subsamples of size $n$ for varying values of $n$.

In addition to the real world data, we generate datasets of varying size from three fixed probability distributions over $[-1,1]$. We set the probability mass for $x\in [-1,1]$ proportional to a chosen function $f(x)$, and equal to $0$ for $x\notin [-1,1]$. We consider the following choices for $f$: \texttt{Gaussian}, $f(x) = e^{-0.5 x^2}$; \texttt{Sine}, $f(x) = \sin(\pi x) + 1$; and \texttt{Power Law}, $f(x) = (x + 1.1)^{-2}$.

For all datasets, we run \Cref{alg:dp_algorithm} with privacy parameters $\e=0.5$ and $\delta=1/n^2$; this is a standard setting for private synthetic data \cite{McKennaMullinsSheldon:2022,rosenblatt2023epistemic}. We use the default choice of $k = \ceil{2\e n}$. In \Cref{fig:experiments}, we plot the average Wasserstein error achieved across $10$ trials of the method as a function of $n$. Error varies across trials due to the randomness in \Cref{alg:dp_algorithm} (given its use of the Gaussian mechanism) and due to the random choice of a subsample of size $n$. 

As we can see, our experimental results strongly confirm our theoretical guarantees: the average $W_1$ error closely tracks our theoretical accuracy bound of $O\left(\log(\epsilon n)\sqrt{\log(1/\delta)}/\epsilon n\right)$ from \Cref{thm:synth_data}, which is shown as a blue dotted line in \Cref{fig:experiments}.

\section{Spectral Density Estimation}\label{sec:sde}
In this section, we present a second application of our main result to the linear algebraic problem of Spectral Density Estimation (SDE). We recall the setting from \Cref{sec:intro_applications}: letting $p$ be the uniform distribution over the eigenvalues given $\lambda_1 \geq \dots \geq\lambda_n$ of a symmetric matrix $A \in \R^{n \times n}$, the goal is to find some distribution $q$ that satisfies
\begin{equation}
\label{eq:sde_guar_sec}
W_1(p,q) \leq \e \|A\|_2.
\end{equation}
In many settings of interest, $A$ is implicit and can only be accessed via matrix-vector multiplications. So, we want to understand 1) how many matrix-vector multiplications with $A$ are required to achieve \eqref{eq:sde_guar_sec}, and 2) how efficiently can we achieve \eqref{eq:sde_guar_sec} in terms of standard computational complexity.

We show how to obtain improved answers to these questions by using our main result, \Cref{thm:master_thm}, to give a tighter analysis of an approach from \cite{BravermanKrishnanMusco:2022}. Like other SDE methods, that approach uses \emph{stochastic trace estimation} to estimate the Chebyshev moments of $p$. In particular, let $m_1, \ldots, m_k$ denote the first $k$ Chebyshev moments. I.e., $m_j = \frac{1}{n}\sum_{i=1}^n T_j(\lambda_i)$. Then we have for each $j$,
\begin{align*}
m_j = \frac{1}{n}\sum_{i=1}^n T_j(\lambda_i) = \frac{1}{n}\tr(T_j(A)), 
\end{align*}
where $\tr$ is the matrix trace. Stochastic trace estimation methods like Hutchinson's method can approximate $\tr(T_j(A))$ efficiently via multiplication of $T_j(A)$ with random vectors \cite{Girard:1987,Hutchinson:1990}. In particular, for any vector $g\in \R^{n}$ with mean 0, variance 1 entries, we have that:
\begin{align*}
\E[g^T T_j(A)g] = \tr(T_j(A)).
\end{align*}
$T_j(A)g$,  and thus $g^T T_j(A)g$, can be computed using $j$ matrix-vector products with $A$. In fact, by using the Chebyshev polynomial recurrence, we can compute $g^T T_j(A)g$ for all $j = 1, \ldots, k$ using $k$ total matrix-vector products:
\begin{align*}
T_0(A) g &= g & T_1(A) g &= Ag & &\ldots & T_j(A) g &= 2 A T_{j-1}(A) g - T_{j-2}(A)g.
\end{align*}
Optimized methods can actually get away with $\lceil k/2 \rceil$ matrix-vector products \cite{Chen:2023}. 
Using a standard analysis of Hutchinson's trace estimator (see, e.g., \cite{Roosta-KhorasaniAscher:2015} or \cite{CortinovisKressner:2022}) \citet{BravermanKrishnanMusco:2022} prove the following:
\begin{lemma}[{\citep[Lemma 4.2]{BravermanKrishnanMusco:2022}}] \label{lem:bkm42}
    Let $A$ be a matrix with $\|A\|_2 \leq 1$. Let $C$ be a fixed constant, $j\in \Np$, $\alpha, \gamma \in (0,1)$, and $\ell_j = \lceil 1 +  \frac{C\log^2(1/\alpha)}{n j \gamma^2}\rceil$.
    Let $g_1,\dots,g_{\ell_j} \sim \mathrm{Uniform}(\set{-1,1}^n)$ and let $\hat m_j = \frac{1}{\ell_j n} \sum_{i=1}^{\ell_j} g_i^\top T_j(A) g_i$. Then, with probability $1-\alpha$, 
    $
    \abs{\hat m_j-m_j} \leq \sqrt{j} \gamma.
    $
\end{lemma}
We combine this lemma with \Cref{thm:master_thm} to prove the following more precise version of \Cref{cor:sde}: 
\begin{theorem}
\label{cor:sde_detailed}
There is an algorithm that, given $\e \in (0,1)$, symmetric $A\in \R^{n\times n}$ with spectral density $p$, and upper bound $S \geq \|A\|_2$, uses $\min\left\{n, O\left(\frac{1}{\e} \cdot \left (1+\frac{\log^2(1/\e) \log^2(1/(\e \delta))}{n\e} \right )\right)\right\}$ matrix-vector products with $A$ and $\tilde{O}(n/\epsilon + 1/\e^{3})$ additional time to output a distribution $q$ such that, with probability at least $1-\delta$, $W_1(p,q) \leq \e S$.
\end{theorem}
\begin{proof}
First note that, if $\epsilon \leq 1/n$, the above result can be obtained by simply recovering $A$ by multiplying by all $n \leq 1/\epsilon$ standard basis vectors. We can then compute a full eigendecomposition to extract $A$'s spectral density, which takes $o(n^3)$ time. So we focus on the regime when $\epsilon > 1/n$. 

Without loss of generality, we may assume from here forward that $\|A\|_2 \leq 1$ and our goal is to prove that $W_1(p,q) \leq \epsilon$. In particular, we can scale $A$ by $1/S$, compute an approximate spectral density $q$ with error $\e$, then rescale by $S$ to achieve error $\e S$. As mentioned in \Cref{sec:intro_applications}, an $S$ satisfying $\|A\|_2 \leq S \leq 2\|A\|_2$ can be computed using $O(\log n)$ matrix-multiplications with $A$ via the power method \cite{KuczynskiWozniakowski:1992}. Given such an $S$, \Cref{cor:sde_detailed} implies an error bound of $2\epsilon\|A\|_2$. In some settings of interest for the SDE problem, for example when $A$ is the normalized adjacency matrix of a graph \cite{Cohen-SteinerKongSohler:2018,DongBensonBindel:2019,JinKarmarkarMusco:2024}, $\|A\|_2$ is known apriori, so we can simply set $S = \|A\|_2$.

Choose $k = \hat{c}/\e$ for a sufficiently large constant $\hat{c}$ and apply \Cref{lem:bkm42} for all $j = 1, \ldots, k$ with $\gamma = \frac{1}{k\sqrt{1+\log k}}$, and $\alpha = \delta/k$. By a union bound, we obtain estimates $\hat{m}_1, \ldots, \hat{m}_k$ satisfying, for all $j$,
\begin{align}
\label{eq:pointwise_condition}
|\hat{m}_j - m_j| \leq \sqrt{j}\gamma = \sqrt{j}\cdot \frac{1}{k\sqrt{1+\log k}}.
\end{align}
Applying \Cref{thm:master_thm} (specifically, \eqref{eq:master_require}) and \Cref{corr:recovery}, we conclude that, using these moments, \Cref{alg:weighted_regression} can recover a distribution $q$ satisfying:
\begin{align*}
W_1(p,q) \leq  \frac{2c'}{k}.
\end{align*}
I.e., we have $W_1(p,q) \leq \e$ as long as $\hat{c} \geq 2c'$. This proves the accuracy bound. We are left to analyze the complexity of the method. We first bound the total number of matrix-vector multiplications with $A$, which we denote by $T$. Since $\ell_j \leq \ell_{j-1}$ for all $j$, computing the necessary matrix-vector product to approximate $m_j$ only costs $\ell_{j-1}$ additional products on top of those used to approximate $m_{j-1}$. So, recalling that $\ell_j = \lceil 1 +  \frac{C\log^2(1/\alpha)}{n j \gamma^2}\rceil$, we have:
\begin{align*} 
T = \left(1 + \frac{C\log^2(k/\delta)}{n\gamma^2}\right) + \left(1 + \frac{C\log^2(k/\delta)}{2n\gamma^2}\right) +\dots + \left(1 + \frac{C\log^2(k/\delta)}{kn\gamma^2}\right).
\end{align*}
Using the fact that $1 + 1/2 + \ldots + 1/k \leq 1+\log(k)$ we  can upper bound $T$ by:
\begin{align*}
% \label{eq:total_mv}
T = \bigO\left(k + \frac{\log^2(k/\delta)\log(k)}{n\gamma^2}\right) = \bigO\left(k + \frac{k^2 \log^2(k/\delta)\log^2(k)}{n}\right),
\end{align*}
which gives the desired matrix-vector product bound since $k = O(1/\e)$.

In terms of computational complexity, \Cref{corr:recovery} immediately yields a bound of $\poly(1/\e)$ time to solve the quadratic program in \Cref{alg:weighted_regression}. However, this runtime can actually be improved to $\tbigo{1/\e^3}$ by taking advantage of the fact that $\hat{m}_1,\ldots, \hat{m}_k$ obey the stronger bound of \eqref{eq:master_require} instead of just \eqref{eq:master_require_general}. This allows us to solve a linear program instead of a quadratic program. In particular, let $\cC$ be a grid of Chebyshev nodes, as used in \Cref{alg:weighted_regression}. I.e., 
$\cC = \set{x_1,\dots,x_{g}}$ where $x_i = \cos\paren{\frac{2i-1}{2g} \pi}$. Let $q^{\lp}_{1}, \dots, q^{\lp}_{g}$ be any solution to the following linear program with variables $z_1,\dots,z_g$:
% \begin{equation}\label{eq:lp}
%     \begin{array}{ll@{}ll}
%         \text{minimize} &  0 \\
%         \text{subject to} &\displaystyle \sum_{i=1}^{g}  z_i = 1 \\
%         & \qquad\displaystyle  z_i \geq 0, &\,\, \forall i \in \set{1,\dots,g}\\
%         & \displaystyle {\sum_{i=1}^{g} T_j(x_i)  z_i} \leq  \hat m_j +\left(\sqrt{j}\gamma +\frac{j \sqrt{2 \pi}}{g}\right), &\,\,\forall j \in \{1, \ldots, k\}\\
%         & \displaystyle {\sum_{i=1}^{g} T_j(x_i)  z_i} \geq \hat m_j - \left(\sqrt{j}\gamma +\frac{j \sqrt{2 \pi}}{g}\right),  &\,\,\forall j \in \{1, \ldots, k\}.
%         \end{array}
% \end{equation}
\begin{align}\label{eq:lp}
        \text{minimize }  0 \hspace{1em} 
        \text{subject to} \hspace{1em}  \sum_{i=1}^{g}  z_i &= 1 &\,\\
         z_i &\geq 0, &\,\, \forall i \in \set{1,\dots,g}\nonumber\\
        {\sum_{i=1}^{g} T_j(x_i)  z_i} &\leq  \hat m_j +\left(\sqrt{j}\gamma +\frac{j \sqrt{2 \pi}}{g}\right), &\,\,\forall j \in \{1, \ldots, k\}\nonumber\\
         {\sum_{i=1}^{g} T_j(x_i)  z_i} &\geq \hat m_j - \left(\sqrt{j}\gamma +\frac{j \sqrt{2 \pi}}{g}\right),  &\,\,\forall j \in \{1, \ldots, k\}.\nonumber
\end{align}
We first verify that the linear program has a solution. To do so, 
note that, by \Cref{eq:p_tilde_p} in \Cref{app:corr_proof}, there exists a distribution $\tilde{p}$ supported on $\cC = \set{x_1,\dots,x_{g}}$, such that $\abs{m_j(p)-m_j(\tilde p)} \leq \frac{j \sqrt{2 \pi}}{g}$.
By \eqref{eq:pointwise_condition} and triangle inequality, it follows that $\tilde p$ is a valid solution to the linear program.

Next, let $q^{\lp} = \sum_{i=1}^{g} q^{\lp}_i \delta(x - x_i)$ be the distribution formed by any solution to the linear program. We have that, for any $j$,
\begin{align*}
\abs{m_j - \langle{q^{\lp}, T_j}\rangle} \leq \abs{\langle{q^{\lp}, T_j}\rangle - \hat m_j} + \abs{\hat m_j - m_j} \leq 2\sqrt{j}\gamma +\frac{j \sqrt{2 \pi}}{g}.
\end{align*}
Setting $g = k^{1.5}\sqrt{1 + \log(k)}$ and plugging into \Cref{thm:master_thm}, we conclude that
$W_1(p,q^{\lp}) \leq O(1/k).$

 The linear program in \Cref{eq:lp} has $g = \tilde{O}(k^{1.5})$ variables, boundary constraints for each variable, and $2k + 1$ other constraints. It follows  that it can be solved in $\tilde{O}(gk\cdot \sqrt{k}) = \tilde{O}(k^{3})$ time \cite{LeeSidford:2014,ls15}, which equals $\tilde{O}(1/\e^{3})$ time since we chose $k = O(1/\epsilon)$.
\end{proof}

\section{Learning Populations of Parameters}\label{sec:population_of_param}
In this section, we present the final application of our results to the ``population of parameters problem'' introduced as \Cref{prob:popprob} in \Cref{sec:intro_applications}. Unlike our prior two applications to differentially private synthetic data and spectral density estimation, we obtain an improvement on the prior work by applying the global Chebyshev coefficient decay bound from \Cref{claim:sum_c_j_sq} directly, instead of applying the full moment matching bound from \Cref{thm:master_thm}. We recall the problem statement below:

\popprob*

\cite{vinayak19a} shows that the maximum likelihood estimator (MLE) of $p$ can be formulated as: 
\begin{equation}\label{eq:mle}
    \hat{p}_{\text{mle}} \in \argmax_{Q \in \cD} \sum_{i=1}^N \log \int_0^1 \binom{t}{X_i} y^{X_i} (1-y)^{t-X_i} \rd Q(y) \mcom
\end{equation}
where $\cD$ denotes the set of all distributions on $[0,1]$. They prove that, in the \emph{small sample regime}, when $t = O(\log N)$, the MLE obtains error $W_1(p, \hat{p}_{\text{mle}}) \leq O(1/t)$. This improves on the naive estimator that simply returns a uniform distribution based on empirical estimates of $p_1, \ldots, p_N$, which gives Wasserstein error $O(1/\sqrt{t} + 1/\sqrt{N})$. Moreover, they prove that it is also possible to beat the naive estimator in the \emph{medium sample regime}:

\thmvinayak*

We improve this result in the medium sample regime to hold for a wider range of $t$, showing:
\begin{theorem}[Improvement in the Medium Sample Regime] \label{thm:our_improvement}
    There exists an $\e >0$, such that, for $t \in \Brac{\Omega(\log N), \bigo{N^{1/4-\e}}}$, with probability at least $99/100$, 
    \begin{equation}
    \label{eq:our_pop_bound}
        W_1(p, \hat{p}_{\text{mle}}) \leq O \paren{\frac{1}{\sqrt{t \log N}}}   \mper
    \end{equation}
\end{theorem} 
As will be discussed in \Cref{sec:conj}, under a natural conjecture from \cite{vinayak19a}, our approach can actually be used to extend the range for which \eqref{eq:our_pop_bound} holds all the way to $t = O(N^{1-\epsilon})$ for any fixed constant $\epsilon$, which is essentially optimal.

\paragraph{Notation.} We begin by introducing notation used throughout this section. Unlike prior applications, \Cref{prob:popprob} involves distributions over $[0,1]$ instead of $[-1,1]$. For this reason, we use \emph{shifted} Chebyshev polynomials, which we denote by $\tilde T_0(x), \tilde{T}_1(x), \ldots$, where the degree $m$ {shifted} Chebyshev polynomial, $\tilde T_m$, is defined as $\tilde T_m(x) = T_m(2x-1)$. Note that the shifted Chebyshev polynomials are orthogonal on the range $[0,1]$ under weight function $w(2x-1)$, where $w(x) = \frac{1}{\sqrt{1 - x^2}}$ is as defined in \Cref{fact:cheb_prop}. Also note that Jackson's theorem (\Cref{fact:jackson}) and our global Chebyshev coefficient decay bound (\Cref{claim:sum_c_j_sq}) continue to hold up to small changes in constant factors when working with shifted Chebyshev polynomial expansions of Lipschitz functions on $[0,1]$.

\subsection{\texorpdfstring{Proof of \Cref{thm:our_improvement}}{Proof of Theorem \ref{thm:our_improvement}}}
The approach from \cite{vinayak19a} centers on rewriting $\hat{p}_{\text{mle}}$ in terms of the \emph{fingerprint} of the observed coin tosses, which can be shown to be a sufficient statistic for the estimation problem. Recall that the observations are $\set{X_i}_{i=1}^N$, where $X_i \sim \text{Binomial}(t,p_i)$. For $s \in \set{0,1,\dots,t}$, let $n_s$ denote the number of coins that evaluate to $1$ on $s$ of the $t$ tosses, i.e. $n_s = |\{i: X_i = s\}|$. Let $\hobs{s}$ denote the fraction of coins that evaluate to $1$ on $s$ tosses, i.e., $\hobs{s} = n_s/N$. The {fingerprint} is defined as $\bm{\hobs{}} \seteq (\hobs{0},\dots,\hobs{t})$. Similarly, for any distribution $Q$, let $\E_Q[h_j]$ denote the expected fraction of coins that evaluate to $1$ on $j$ out of $t$ tosses when $X_1,\ldots, X_N$ are drawn from some distribution $Q$.

% Expressing the MLE (\cref{eq:mle}) in terms of the fingerprint, we have that 
% \begin{align*}
%     \hat{p}_{\text{mle}} &\in \argmax_{Q \in \cD} \sum_{i=1}^N \log \int_0^1 \binom{t}{X_i} y^{X_i} (1-y)^{t-X_i} \rd Q(y) \\
%     & = \argmax_{Q \in \cD} \sum_{s=0}^t n_s \log \underbrace{\int_0^1 \binom{t}{s} y^s (1-y)^{t-s} q(y) \rd y}_{=:\E_Q[h_s]} \\
%     & = \argmax_{Q \in \cD} \sum_{s=0}^t \hobs{s} \log \E_Q[h_s] \\
%     & = \argmin_{Q \in \cD} \operatorname{KL}\paren{\bm{h^{\text{obs}}}, \E_Q[\bm{h}]} \mcom
% \end{align*}
% where $E_Q[h_s]$ is the expected fraction of the population that sees $s$ heads out of $t$ tosses under the distribution $Q$,  $\E_Q[\bm{h}]$ denotes the expected fingerprint vector when the biases are drawn from $Q$, and $\operatorname{KL}(\bm{h}^{\text{obs}},\bm{h})$ is the Kullback-Leibler divergence between the observed fingerprint vector and the expected fingerprint vector.

% Recall the dual of the Wasserstein-$1$ distance (\cref{eq:our_w1}), for two distributions $p,q$ over $[0,1]$, we have that 
% \begin{align*}
%         W_1(p,q) = \sup_{1\text{-Lipschitz, smooth } f} \int_{0}^1 f(x) (p(x)-q(x)) \rd x = \sup_{1\text{-Lipschitz, smooth } f} (\E_p[f]-\E_q[f]) \mper
% \end{align*}

\citet{vinayak19a} prove a result relating the Wasserstein error of $\hat{p}_{mle}$ to how closely the expected fingerprints under $\hat{p}_{mle}$ match the observed fingerprints. Like our \Cref{thm:master_thm} on matching moments, this result leverages the dual definition of Wasserstein distance involving Lipschitz functions (\Cref{fact:w1_dual}) and proceeds by replacing $f$ with a polynomial approximation, $\hat{f}$. Just as our proof depends on the Chebyshev coefficients of $\hat{f}$, their result depends on the coefficients of $\hat{f}$ when written in a \emph{Bernstein polynomial basis}. In particular, let $B_j^t(x) = \binom{t}{j} x^j (1-x)^{t-j}$ denote the $j^\text{th}$ Bernstein polynomial of degree $t$. \citet{vinayak19a} works with a degree $t$ approximation $\hat{f}$ of the form $\hat{f} = \sum_{j=0}^t b_j B_j^t(x)$. They prove the following:

% that any $1$-Lipschitz function $f$ on $[0,1]$ can be approximated using Bernstein polynomials as $\hat{f}(x) \seteq \sum_{j=0}^{t} b_j \binom{t}{j} x^j (1-x)^{t-j}$. Using this one can see that
% \begin{align*}
%     \E_p[f]-\E_q[f] & = (\E_p[f-\hat f]-\E_q[f-\hat f]) + \E_p[\hat f]-\E_q[\hat f] \\
%     & \leq 2 \norm{f-\hat f}_\infty + \sum_{j=0}^t b_j (\E_p[h_j]-\E_q[h_j]) \mcom
% \end{align*}
% where the first equality follows from linearity of expectation, and the first term in the inequality follows from the fact that $P, Q$ are distributions and the second term in the inequality follows by the definition of $\hat{f}$ and $\E_Q[h_j]$. This also shows the usefulness of the Bernstein polynomials in approximating the Lipschitz function $f$, instead of directly using the Chebyshev polynomials like we do. Therefore, we have that
\notsotiny
\begin{align} \label{eq:with_terms} 
    &W_1(p, \hat{p}_{\text{mle}}) \leq \sup_{\underset{\text{smooth}~f}{1\text{-Lipschitz,}} } \left[ \inf_{\hat{f} = \sum_{j=0}^t b_j B_j^t(x)}\left[ 2 \underbrace{\norm{f-\hat f}_{\infty}}_{\text{(a)}} + \underbrace{\sum_{j=0}^t b_j (\E_{p}[h_j]-\hobs{j})}_{\text{(b)}} + \underbrace{\sum_{j=0}^t b_j (\hobs{j}-\E_{\hat{p}_{\text{mle}}}[h_j])}_{\text{(c)}}  \right]\right].
\end{align}
\normalsize
Above and in the remainder of this section, $\norm{f-\hat f}_{\infty}$ denotes $\max_{x\in [0,1]} |f(x) - \hat{f}(x)|$ (instead of our usual definition involving $x\in [-1,1].$)
\cite{vinayak19a} bounds the terms (b) and (c) as follows:

\begin{lemma}[{\citep[Lemmas 4.1 and 4.2]{vinayak19a}}] \label{lem:termb_c}
Term (b): With probability $1-\delta$, 
\[ \abs{\sum_{j=0}^t b_j (\E_{p}[h_j]-\hobs{j})} \leq \bigo{ \max_j \abs{b_j} \sqrt{\frac{\log 1/\delta}{N}} } \mper\]
Term (c): For $3 \leq t \leq \sqrt{C_0 N}+2$, where $C_0 >0$ is constant,  with probability $1-\delta$, 
\[ \abs{\sum_{j=0}^t b_j (\hobs{j}-\E_{\hat{p}_{\text{mle}}}[h_j])} \leq \max_j \abs{b_j} \sqrt{2 \ln 2} \sqrt{ \frac{t}{2N} \log \frac{4N}{t} + \frac{1}{N} \log \frac{3e}{\delta}} \mper  \]
\end{lemma}
It remains to bound (a), i.e., $\norm{f-\hat f}_{\infty}$, as well as  $\max_j \abs{b_j}$ which appears in both bounds above. 

% We recall $\hat f$
% \[ \hat f(x) = \sum_{j=0}^t b_j \binom{t}{j} x^j (1-x)^{t-j} := \sum_{j=0}^t b_j B_j^t(x) \mcom \]
% where $B_j^t(x) = \binom{t}{j} x^j (1-x)^{t-j}$ are the $j$-th Bernstein polynomial of degree $t$. 
Doing so requires proving that there exist good uniform polynomial approximations to $f$ that have bounded coefficients $b_0, \ldots, b_t$ in the Bernstein polynomial basis. Towards that end, \citet{vinayak19a} prove the following key result: 
% One of the main challenges in \citet{vinayak19a} was to get a good uniform approximation to $f$ by $\hat f$, while keeping the coefficients $b_j$ small. \citet{vinayak19a} show in their Proposition 4.2 the following, which helps them get the desired bound on $W_1(p, \hat{p}_\text{mle})$:
\begin{proposition}[{\citep[Proposition 4.2]{vinayak19a}}]
    Any $1$-Lipschitz function on $[0,1]$ can be approximated by a degree $t$ polynomial $\hat f(x) = \sum_{j=0}^t b_j B_j^t(x)$,  such that, for any $k < t$,
    \begin{align*} 
    \norm{f-\hat f}_{\infty} &\leq \bigo{\frac{1}{k}} & &\text{and} & \max_j \abs{b_j} &\leq \sqrt{k} (t+1) e^{k^2/t}. 
    \end{align*}
\end{proposition}
Proposition 4.2 is proven by using Jackson's theorem (\Cref{fact:jackson}) to approximate $f$ by a degree $k$ polynomial $f_k$. Recall that $f_k$ is written as a linear combination of Chebyshev polynomials. \cite{vinayak19a} then obtain $\hat{f}$ by expressing these Chebyshev polynomials as linear combinations of Bernstein polynomials of degree $t$. Naturally, by using our \Cref{claim:sum_c_j_sq} to give a better bound on the Chebyshev coefficients of $f_k$, we can improve their bound on the Bernstein polynomial coefficients, $b_0,\ldots, b_t$, of $\hat{f}$. Concretely, we show the following:
\begin{proposition}[Improvement to {\citep[Proposition 4.2]{vinayak19a}}] \label{prop:bernstein_bound}
    Any $1$-Lipschitz function on $[0,1]$ can be approximated by a degree $t$ polynomial $\hat f(x) = \sum_{j=0}^t b_j B_j^t(x)$,  such that, for any $k < t$,
    \begin{align*} 
    \norm{f-\hat f}_{\infty} &\leq \bigo{\frac{1}{k}} & &\text{and} & \max_j \abs{b_j} &\leq (t+1) e^{k^2/t}. 
    \end{align*}
\end{proposition}
\begin{proof}
    Let $f_k = \sum_{m=0}^k a_m \tilde T_m(x)$ be the damped truncated Chebyshev series of $f$ as defined in \Cref{fact:jackson} (appropriately shifted and scaled to involve the shifted Chebyshev polynomials over $[0,1]$) \footnote{We use $a_0, \ldots, a_k$ to denote the damped coefficients to avoid confusion with the coefficients $b_0, \ldots, b_t$ above}. Recall that, for all $i$, $a_i \leq 1$.
    From \Cref{fact:jackson}, we have that $\norm{f-f_k} \leq \bigo{1/k}$. Any Chebyshev polynomials $\tilde T_m$ can be expressed as a linear combination of Bernstein polynomials of degree $m$:
    \begin{align*}
        \tilde T_m(x) &= \sum_{i=0}^{m} (-1)^{m-i} \frac{\binom{2m}{2i}}{\binom{m}{i}} B_i^m(x) & \paren{\text{{\citep[eq. 21]{vinayak19a}}}} \mper
    \end{align*}
    Moreover, following \cite{vinayak19a}, we can express any degree $m$ Bernstein polynomial as an appropriate sum of Bernstein polynomial of higher degree $t$:
    \begin{align*}
        B_i^m(x) &= \sum_{j=i}^{i+t-m} \frac{\binom{m}{i} \binom{t-m}{j-i}}{\binom{t}{j}} B_j^t(x) & \paren{\text{{\citep[eq. 22]{vinayak19a}}}} \mper
    \end{align*}
    Combining the two equations above, we have that, for $m <t$, 
    \begin{align*}
        \tilde T_m(x) &= \sum_{i=0}^{m} (-1)^{m-i} \frac{\binom{2m}{2i}}{\binom{m}{i}} \sum_{j=i}^{i+t-m} \frac{\binom{m}{i} \binom{t-m}{j-i}}{\binom{t}{j}} B_j^t(x) 
        =: \sum_{j=0}^t C(t,m,j) B_j^t(x) & \paren{\text{{\citep[eq. 23]{vinayak19a}}}}.
    \end{align*}
    Lemma 4.4 of \citet{vinayak19a} then gives us that 
    \begin{align} \label{eq:ctmj_bound}
        \abs{C(t,m,j)} & \leq (t+1) e^{m^2/t}  \mper
    \end{align}
    Next, 
    we choose $\hat{f}$ to be:
    \begin{align*}
            \hat{f} &=\sum_{j=0}^t b_j B_j^t(x) :=\sum_{m=0}^k a_m \paren{ \sum_{j=0}^t C(t,m,j) B_j^t(x)} = \sum_{m=0}^k a_m \tilde T_m(x) = f_k.
    \end{align*}
    Above, $b_j = \sum_{m=0}^k a_m C(t,m,j)$. Using the fact that $\abs{C(t,0,j)} \leq (t+1)$ alongside our global Chebyshev coefficient decay bound from \Cref{claim:sum_c_j_sq} we can bound each coefficients $b_j$ as follows:
    \begin{align}
        \abs{b_j} & \leq \abs{a_0 C(t,0,j)}  + \abs{\sum_{m=1}^k a_m C(t,m,j)} \nonumber \\
        & =  \abs{a_0 C(t,0,j)} + \abs{\sum_{m=1}^k m a_m \frac{C(t,m,j)}{m}} \nonumber \\
        & \leq (t+1) + \paren{\sum_{m=1}^k m^2 a_m^2}^{1/2} \cdot \paren{ \sum_{m=1}^k \frac{C(t,m,j)^2}{m^2} }^{1/2} \label{eq:cz-pop}  \\
        & \leq (t+1) + \paren{\sum_{m=1}^k m^2 a_m^2}^{1/2} \cdot \paren{ (t+1)^2 \sum_{m=1}^k \frac{e^{2m^2/t}}{m^2} }^{1/2} \label{eq:ctmj_bd_pop}  \\
        & \leq (t+1) + \paren{\sum_{m=1}^k m^2 a_m^2}^{1/2} \cdot  \paren{ (t+1)^2e^{2k^2/t} \sum_{m=1}^k \frac{1}{m^2} }^{1/2} \nonumber \\
        & \leq (t+1) +  \paren{\sum_{m=1}^k m^2 a_m^2}^{1/2} \cdot (t+1) \paren{e^{k^2/t} \sqrt{\frac{\pi^2}{6}}} \label{eq:1m2sum} \\
        & \leq (t+1) +  C_1' \cdot (t+1) \paren{e^{k^2/t} \sqrt{\frac{\pi^2}{6}}} \leq C_1 (t+1) e^{k^2/t} \label{eq:decay_shifted}\mcom
    \end{align}    
    where $C_1,C_1'>0$ are some absolute constants. Since $f$ is a $1$-Lipschitz function on $[0,1]$, we can let $\abs{f(x)} \leq 1/2$, as we can shift $f$ can such that its range is bounded between $[-1/2,1/2]$. It can be checked that this implies that $\abs{a_0} \leq 1$. \Cref{eq:cz-pop}  follows by combining the fact that $\abs{a_0}\leq 1$ with the bound on $\abs{C(t,0,j)}$, and Cauchy Schwarz inequality.
    %The last inequality follows from the fact that $\sum_{m=1}^k m^2 a_m^2 \leq 1$ and the bound on $C(t,m,j)$ from \cref{eq:ctmj_bound}. 
    \Cref{eq:ctmj_bd_pop} follows from the bound on $\abs{C(t,m,j)}$ in \Cref{eq:ctmj_bound}. \Cref{eq:1m2sum} follows from the fact that $\sum_{m=1}^{\infty} 1/m^2 \leq \pi^2/6$. Let $\sum_{m=0}^{\infty}c_m \tilde T_m(x)$ be the shifted Chebyshev series of $f$. Then, we know from \Cref{fact:jackson} that $\abs{a_m} \leq \abs{c_m}$, and from the global Chebyshev coefficient decay \Cref{claim:sum_c_j_sq} that $\sum_{m=1}^{\infty} m^2 c_m^2 \leq C_1'$, for some constant $C_1'$. This proves the first inequality of \Cref{eq:decay_shifted}.    

    \Cref{eq:decay_shifted} gives us the bound on the coefficients $\abs{b_j}$. Combined with the fact that $\norm{f-\hat f_k}_{\infty} \leq \bigo{1/k}$, the proposition follows. 
\end{proof}

We are now ready to prove our main theorem from this section. 

\begin{proof}[Proof of \Cref{thm:our_improvement}]
 For a $1$-Lipschitz function $f$, let $\hat f(x) = \sum_{j=0}^t b_j B_j^t (x)$ denote a degree $t$ Bernstein polynomial approximation to $f$. We use \Cref{eq:with_terms} to bound $W_1(p, \hat{p}_{\text{mle}})$. Specifically, by \Cref{prop:bernstein_bound}, there is a choice of $\hat f(x)$ which ensures that the (a) term can be bounded by  $2\norm{f-\hat f}_{\infty} \leq \bigo{1/k}$. Moreover, we will have that $\max_j \abs{b_j} \leq (t+1) e^{k^2/t}$, for $k <t$. We will set $k = \sqrt{t \log (N^c)}$ for a small constant $c>0$ to be chosen later. Note that since we require $k < t$ for \Cref{prop:bernstein_bound} to hold, doing so requires $t = \Omega(\log n)$. With this choice of $k$,  we have that $\max_j \abs{b_j} \leq (t+1) N^c$. We can then plug this coefficient bound into  \Cref{lem:termb_c} to show that, with probability $99/100$, and $3 \leq t \leq \sqrt{C_0 N}+2$, 
\begin{align*}
    W_1(p, \hat{p}_{\text{mle}}) & \leq \bigo{\frac{1}{k}} + \bigo{ \max_j \abs{b_j} \sqrt{\frac{1}{N}} } + \bigo{ \max_j \abs{b_j} \sqrt{ \frac{t}{2N} \log \frac{4N}{t} + \frac{1}{N} } } \\
    & \leq \bigo{\frac{1}{\sqrt{t \log N}}} + \bigo{ (t+1) N^c \sqrt{\frac{t}{N} \log N} } \mper\\ 
    & = \bigo{\frac{1}{\sqrt{t \log N}}} + \bigo{  \sqrt{\frac{t^3}{N^{1-2c}} \log N} } \mper
\end{align*}
For any target constant $\epsilon$, we can choose our constant $c$ so that $N^{\epsilon} = N^{c}\log N$. We can then check that, as long as 
 $t = \bigo{N^{1/4 - \epsilon}}$, $O\left(\sqrt{\frac{t^3}{N^{1-2c}} \log N}\right) = O\left(\sqrt{\frac{1}{t \log N}}\right)$, which proves the theorem. 
\end{proof}

\subsection{Conjectured Improvement}
\label{sec:conj}

\citet{vinayak19a} conjecture that the range of $t$ for which \Cref{thm:vinayak} holds can be improved. In particular, they conjecture that the bound on the coefficients in the proof of \Cref{prop:bernstein_bound} can be improved to $|C(t,m,j)| \leq e^{m^2/t}$ for $j=0,\ldots, t$. Moreover, they conjecture that the bound on (c) in \Cref{eq:with_terms} can be improved to $\bigo{\max_j \abs{b_j} \sqrt{\log(1/\delta)/N}}$. If these conjectures hold, the range of $t$ can be improved to $t \in \Brac{\Omega(\log N), \bigo{N^{2/3-\e}}}$. If we additionally include our improvement from \Cref{prop:bernstein_bound}, we would obtain a further improvement to $t \in \Brac{\Omega(\log N), \bigo{N^{1-\e}}}$.

We note that improving the upper limit on $t$ to $\bigo{N^{1-\e}}$ is essentially the best that we can hope for. In particular, there exist distributions that are $1/\sqrt{N}$ far away in $W_1$ distance that would need $N$ independent coins to distinguish between them, even if $t = \infty$. Consider two distributions with $q_1$ and $q_2$ such that $q_1$ has probability mass of $1/2 +1/\sqrt{N}$ on $0$ and $1/2 - 1/\sqrt{N}$ on $1$, and $q_2$ has probability mass of $1/2 -1/\sqrt{N}$ on $0$ and $1/2 + 1/\sqrt{N}$ on $1$. It is easy to see that $W_1(q_1, q_2) = 2/\sqrt{N}$. The coins drawn from $q_1$ or $q_2$ have biases of either $0$ or $1$. So, in this case, a single coin toss does not provide any less information than infinite coin tosses. By standard information-theoretic arguments \citep{karpKlienberg07}, $\Omega(1/N)$ independent samples are required to distinguish between $q_1$ and $q_2$ with probability greater than $1/2$. Accordingly, when $t = \Omega(N)$, we can no longer achieve error better than the $O(1/\sqrt{t} + 1/\sqrt{N})$ bound given by the naive estimator.

\section*{Acknowledgements}
We thank Gregory Valiant for suggesting us the work on populations of parameters. We thank Raphael Meyer for suggesting the lower bound on the number of matrix-vector multiplications required for spectral density estimation. We thank Tyler Chen for close proofreading and Gautam Kamath for helpful pointers to the literature. This work was partially supported by NSF Grants 2046235 and 2045590. 

\begingroup
\sloppy
\printbibliography
\endgroup

\appendix

\section{\texorpdfstring{Multivariate Generalization of \Cref{thm:master_thm}}{Multivariate Generalization of Theorem 1}} \label{sec:multivariate_master}
In this section, we generalize our \Cref{thm:master_thm} to $d$-dimensions. To prove this, we look at the Chebyshev series of multivariate functions. The Wasserstein-$1$ distance and its dual is analogously defined in $d$-dimensions. 

\begin{definition}[Wasserstein-$1$ Distance, Euclidean Metric]\label{def:w1_high}
   Let $p$ and $q$ be two distributions on $[-1,1]^d$. Let $Z(p,q)$ be the set of all couplings between $p$ and $q$, i.e., the set of distributions on $[-1,1]^d \times [-1,1]^d$ whose marginals equal $p$ and $q$. The Wasserstein-$1$ distance between $p$ and $q$ is:
   \[ W_1(p,q) = \inf_{z \in Z(p,q)} \Brac{ \E_{(x,y) \sim z} \norm{x-y}_2} \mcom \]
   where $\norm{x-y}_2$ denotes the Euclidean distance. 
\end{definition}
Like $W_1$ in $1$-dimension, the Wasserstein distance in $d$-dimension also measures the total cost (in terms of distance per unit mass) required to ``transport'' the distribution $p$ to $q$. Its dual form is as follows. 

\begin{fact}[Kantorovich-Rubinstein Duality in $d$-Dimensions] \label{fact:w1_dual_high} 
Let $p,q$ be as in \Cref{def:w1_high}. Then: 
\[ W_1(p,q) = \sup_{1\text{-Lipschitz, smooth }f} \int_{[-1,1]^d} f(x) \cdot (p(x)-q(x)) \rd x \mcom \]
where $f:[-1,1]^d \to \R$ is a smooth, $1$-Lipschitz function under the Euclidean metric. 
\end{fact}

\subsection{Multivariate Chebyshev Series}

 We use the fact that if $f:[-1,1]^d \to \R$ is smooth, it has a uniformly and absolutely convergent multivariate Chebyshev series \citep{Mason:1980}
\[ f(x) = \sum_{K \in \N^d} C_K T_K(x) \mcom \]
where for $x= (x_1,\dots,x_d)\in [-1,1]^d, K = (k_1,\dots,k_d) \in \N^d$, $T_K(x) = \prod_{i=1}^d T_{k_i}(x_i)$, and $C_K$ is the $K$-th Chebyshev coefficients of $f$, and $T_{k_i}(x_i)$ is the $k_i$-th Chebyshev polynomial of the first kind. First, we will note a few facts about the multivariate Chebyshev polynomials, which are easily derived using properties of the univariate Chebyshev polynomials. 

\begin{definition}[Chebyshev Polynomials in $d$ Dimensions]
    Let $x = (x_1,\dots,x_d) \in \R^d$, and let $K = (k_1,\dots,k_d) \in \N^d$. The $K$-th Chebyshev polynomial of the first kind is denoted by $T_K(x)$, and is defined as:
    \[ T_K(x) \seteq \prod_{i=1}^{d} T_{k_i}(x_i) \mcom \]
    where $T_{k_i}(x_i)$ is the $k_i$-th Chebyshev polynomial of the first kind in one dimension. Let $\wt(x)$ denote the \emph{weight function} defined as 
    \[ \wt(x) \seteq \prod_{i=1}^d \frac{1}{\sqrt{1-x_i^2}} \mper \]
\end{definition}

\begin{definition}[Inner Product in $d$-Dimensions]
    The inner product of two functions $f,g: [-1,1]^d \to \R$ is defined as:
    \[ \inprod{f,g} \seteq \int_{[-1,1]^d} f(x) g(x) \rd x \mper \]
\end{definition}

\begin{fact}[Orthogonality Property of Chebyshev Polynomials in $d$-Dimensions]
    Let $K_1,K_2 \in \N^d$. Let $\nnz{K}$ denote the number of non-zero entries in $K \in \N^d$. The higher dimensional Chebyshev polynomials satisfy the following orthogonality property:
    \[ \inprod{T_{K_1}, W \cdot T_{K_2}} = \int_{[-1,1]^d} T_{K_1}(x) T_{K_2}(x) \wt(x) \rd x = 
    \begin{cases}
         0 & \text{if } K_1 \neq K_2 \\ 
         \frac{\pi^d}{2^{\nnz{K_1}}} & \text{if } K_1 = K_2  
     \end{cases} \mper \]
\end{fact}

\begin{definition}[Normalized $d$-Dimensional Chebyshev Polynomials] \label{def:norm_multi_cheb}
    The normalized $d$-dimensional Chebyshev polynomial $T_K^d(x)$, for $K \in \N^d$ is defined as:
    \[ \T_K(x) \seteq \frac{T_K(x)}{\sqrt{T_K, W \cdot T_K}} = \sqrt{\frac{2^{\nnz{K}}}{\pi^d}} T_K(x)  \mper \]
\end{definition}

With the notations and definitions in place, we can now state the multivariate Jackson's theorem. The following theorem shows that the damped, truncated Chebyshev series of a smooth function is a good uniform multivariate polynomial approximation to the function.

\begin{restatable}[Multivariate Jackson's Theorem]{theorem}{jacksonhighdim} \label{thm:jackson_high_dim}
    Let $f: [-1,1]^d \to \R$ be an $\ell$-Lipschitz smooth function, and for $K\in \Np^d$, let $c_K = \inprod{f, W \cdot \T_K}$. Then the polynomial $\tilde f(x) = \sum_{K \in \set{0,\dots,2m-2}^d} \tilde c_K T_K(x)$ satisfies that 
    \[ \norm{\tilde f - f}_{\infty} \leq \frac{9 \ell d}{m}, \text{ and } \abs{\tilde c_K} \leq \abs{c_K} , \text{ for } K\in \N^d \mper\]
\end{restatable}

We now prove the theorem in \Cref{sec:proof_jackson_high_dim}. With the high dimensional Jackson's theorem in place, we now prove the multivariate global Chebyshev coefficient decay lemma.

\begin{lemma}[Multivariate Global Chebyshev Coefficient Decay] \label{lem:cheb_coeff_high}
    Let $f:[-1,1]^d \to \R$ be a smooth, $\ell$-Lipschitz function. For $K \in \N^d$, let $c_K = \inprod{f, W \cdot \T_K}$. Then, we have that
    \[ \sum_{K \in \N^d}\norm{K}^2_2 c_K^2 \leq d \ell^2 \frac{\pi^d}{2} \mper \]
\end{lemma}

\begin{proof}
    Let $f:[-1,1]^d\to \R$ be a smooth, $\ell$-Lipschitz function, with Chebyshev series 
    \[ f(x) = \sum_{K \in \N^d} \coeff_K T_K(x)  \mper \]
    Let $K=(k_1,\dots,k_d) \in \N^d$. Since $f$ is $\ell$-Lipschitz, it follows that $\norm{\nabla f}_2^2 \leq \ell^2$. Consequently, for $i \in [d]$, we have:
    \[ \paren{\frac{\partial f}{\partial x_i}}^2 (x') \leq \ell^2 \text{ at any } x' \in [-1,1]^d \mper \]
    Therefore, we get that for $x = (x_1,\dots,x_d) \in [-1,1]^d$,

    \begin{equation}\label{eq:lhs_rhs}
        \sum_{i=1}^d \int_{[-1,1]^d} \paren{\frac{\partial f}{\partial x_i}}^2 \cdot \frac{\sqrt{1-x_i^2}}{\prod_{j \neq i \in [d]} \sqrt{1-x_j^2}} \rd x \leq d \ell^2 \frac{\pi^d}{2} \mper
    \end{equation}

    The upper bound follows from the fact that $0 \leq \paren{\frac{\partial f}{\partial x_i}}^2 \leq \ell^2$, $\int_{-1}^{1} \sqrt{1-x_i^2} = \pi/2$, and that $\int_{-1}^{1} 1/\sqrt{1-x_i^2} = \pi$.
    We multiply $\paren{\frac{\partial f}{\partial x_i}}^2$ by $\frac{\sqrt{1-x_i^2}}{\prod_{j \neq i \in [d]} \sqrt{1-x_j^2}}$ and integrate from $[-1,1]^d$ to exploit the orthogonality property of the Chebyshev polynomials of the first and second kind.

    We now evaluate the LHS of \Cref{eq:lhs_rhs}. Computing the gradient, we have from \Cref{fact:tj_uj} that for $i \in [d]$: 
    \[ \frac{\partial f}{\partial x_i} = \sum_{k_i=0}^{\infty} \sum_{(k_1,\dots,k_{i-1},k_{i+1}\dots,k_d)\in \N^d} k_i C_{(k_1,\dots,k_i,\dots,k_d)} \prod_{j \neq i \in [d]} T_{k_j}(x_j)  U_{k_i-1}(x_i) \mper \]
    We consider the square of the above expression. The orthogonality property of Chebyshev polynomials ensures that only the squared terms contribute non-zero values to the integral in \Cref{eq:lhs_rhs}:
    \begin{align}
         & \sum_{k_i=0}^{\infty} \sum_{(k_1,\dots,k_{i-1},k_{i+1}\dots,k_d)\in \N^{d-1}}  k_i^2 C^2_{(k_1,\dots,k_i,\dots,k_d)} \prod_{j \neq i \in [d]} (T_{k_j}(x_j))^2  (U_{k_i-1}(x_i))^2 \label{eq:partial_f2} \mper
    \end{align}
     Using the above equations, we evaluate the LHS of \Cref{eq:lhs_rhs} and get that
    \[  \sum_{i=1}^d \int_{[-1,1]^d} \paren{\frac{\partial f}{\partial x_i}}^2 \cdot \frac{\sqrt{1-x_i^2}}{\prod_{j \neq i \in [d]} \sqrt{1-x_j^2}} \rd x  = \sum_{K=(k_1,\dots,k_d)\in \N^{d}} \frac{\pi^d}{2^{\nnz{K}}} \norm{K}^2 \coeff_{K}^2~,  \]
    by using the orthogonality property of Chebyshev polynomials of the first and second kind and by inspecting the \Cref{eq:partial_f2}.
    Therefore, combining above with \Cref{eq:lhs_rhs}, we get that
    \[ \sum_{K=(k_1,\dots,k_d)\in \N^{d}} \frac{\pi^d}{2^{\nnz{K}}} \norm{K}^2 \coeff_{K}^2 \leq d \ell^2 \frac{\pi^d}{2} \]
    Note that these coefficients are not normalized. To get the normalized Chebyshev coefficients, we use the fact that $\T_K = \sqrt{\frac{2^{\nnz{K}}}{\pi^d}} T_K$. We let $f(x) = \sum_{K \in N^d} c_K \T_K(x)$, which yields
    \[ \sum_{K \in \N^d}\norm{K}^2 c_K^2 \leq d \ell^2 \frac{\pi^d}{2} \mper \qedhere \]
\end{proof}

\subsection{\texorpdfstring{Proof of Multivariate Generalization of \Cref{thm:master_thm}}{Proof of Multivariate Generalization of Theorem 1}}
With the multivariate Jackson's theorem and the global Chebyshev coefficient decay lemma in place, we can now prove the multivariate version of our main theorem.

\begin{theorem}\label{thm:master_thm_high}
    Let $p,q$ be distributions supported on $[-1,1]^d$. For any $K \in \N^d$, if the distributions' normalized Chebyshev moments satisfy
\begin{align}
\label{eq:master_require_general_high}
\sum_{K \in \set{0,\dots,m}^d \setminus \bm 0} \frac{1}{\norm{K}_2^2}\left(\E_{x\sim p}\T_K(x) - \E_{x\sim q}\T_K(x)\right)^2 \leq \Gamma^2, 
\end{align}
where $\bm 0 = (0,\dots,0) \in \R^d$, then, for an absolute constant $c$,
\begin{align}
\label{eq:master_ensure_high}
    W_1(p,q) \leq \frac{cd}{m} + \sqrt{\frac{d\pi^d}{2}}\Gamma.
\end{align}

\end{theorem}
\begin{proof}
    By \Cref{fact:w1_dual_high}, to bound $W_1(p,q)$, it suffices to bound $\langle{f, p-q\rangle}$ for any $1$-Lipschitz, smooth $f$. Let $f_m$ be the approximation to any such $f$ guaranteed by \Cref{thm:jackson_high_dim}. We have:
\begin{align}
\label{eq:ip_split_high}
    \inprod{f, p-q} = \inprod{f_m, p-q} + \inprod{f-f_m, p-q} &\leq \inprod{f_m, p-q} + \|f - f_m\|_\infty\|p-q\|_1 \nonumber\\ &\leq \inprod{f_m, p-q} + \frac{36 d}{m}.
\end{align}
In the last step, we use that $\|f - f_m\|_\infty\leq 18  d/m$ by \Cref{thm:jackson_high_dim}, and that $\|p-q\|_1 \leq \|p\|_1 + \|q\|_1 = 2$.
So, to bound $\inprod{f, p-q}$, we turn our attention to bounding $\inprod{f_m, p-q}$. 

For technical reasons, we will assume from here on that $p$ and $q$ are supported on the interval $[-1+\delta,1-\delta]^d$ for arbitrarily small $\delta \rightarrow 0$. This is to avoid an issue with the Chebyshev weight function $W(x) = \prod_{i=1}^d {1}/{\sqrt{1-x_i^2}}$, for $x = (x_1,\dots,x_d)$ going to infinity at $x_i =-1,1$. The assumption is without loss of generality since we can rescale the support of $p$ and $q$ by a $(1-\delta)$ factor, and the distributions' moments and Wasserstein distance change by an arbitrarily small factor as $\delta \rightarrow 0$.
% \begin{align*}
% f_k = \sum_{j=0}^k c_j \T_j,
% \end{align*}
% where each $c_j$ satisfies $|c_j| \leq \abs{\inprod{f \cdot w, \T_j}}$.

We proceed by writing the Chebyshev series of the function $(p-q)/W$:
\begin{align}
\label{eq:pq_cheby_series_high}
    \frac{p-q}{W} = \sum_{K \in \N^d}^\infty \inprod{\frac{p-q}{W} \cdot W, \T_K } \T_K = \sum_{K \in \N^d} \langle p-q,\T_K\rangle \cdot  \T_K = \sum_{K \neq \bm 0, K \in \N^d}^\infty \langle p-q,\T_K\rangle\cdot  \T_K,
\end{align}
where $\bm 0 = (0,\dots,0) \in \R^d$.
In the last step we use that both $p$ and $q$ are distributions so $\inprod{ p-q, \T_{\bm 0}} = 0$.

Next, recall from \Cref{thm:jackson_high_dim} that $f_m = \sum_{K \in \set{0,\dots,m}^d} \tilde c_K \T_K$, where each $\tilde c_K$ satisfies $|\tilde c_K| \leq |c_K|$ for $c_K \defeq \langle f \cdot W, \T_K\rangle$.
Using \eqref{eq:pq_cheby_series_high}, the fact that $\langle { \T_K\cdot w, \T_{K'}}\rangle  = 0$ whenever $K\neq K'$, and that $\langle { \T_K\cdot W, \T_K}\rangle = 1$ for all $K$, we have:
\begin{align*}
\inprod{ f_m, p-q} = \inprod{ f_m\cdot W, \frac{p-q}{W}} 
& = \inprod{ \sum_{K \in \set{0,\dots,m}^d} \tilde c_K \T_K \cdot W, \sum_{K \neq \bm 0, K \in \N^d} \langle p-q,\T_j\rangle\T_j} \\
& = \sum_{K \in \set{0,\dots,m}^d \setminus \bm 0} \tilde c_K \cdot \langle p-q,\T_K\rangle. 
\end{align*}
Via Cauchy-Schwarz inequality and our high-dimensional global decay bound from \Cref{lem:cheb_coeff_high}, we then have: 
\begin{align} 
\inprod{ f_m, p-q}  & = \sum_{K \in \set{0,\dots,m}^d \setminus \bm 0} \norm{K}_2 \tilde c_K \cdot \frac{\langle p-q,\T_K\rangle}{\norm{K}_2} \nonumber \\
&\leq \left(\sum_{K \in \set{0,\dots,m}^d \setminus \bm 0} \norm{K}_2^2 \tilde c_K^2 \right)^{1/2}\cdot \left(\sum_{K \in \set{0,\dots,m}^d \setminus \bm 0} \frac{1}{\norm{K}_2^2} \langle p-q,\T_K\rangle^2\right)^{1/2} \nonumber\\
&\leq \left(\sum_{K \in \set{0,\dots,m}^d \setminus \bm 0} \norm{K}_2^2 c_K^2 \right)^{1/2}\cdot \left( \sum_{K \in \set{0,\dots,m}^d \setminus \bm 0} \frac{1}{\norm{K}_2^2} \langle p-q,\T_j\rangle^2\right)^{1/2} \nonumber\\
&\leq \sqrt{\frac{d \pi^d}{2}}\left( \sum_{K \in \set{0,\dots,m}^d \setminus \bm 0} \frac{1}{\norm{K}_2^2} \langle p-q,\T_K\rangle^2\right)^{1/2}. \label{eq:cauchy_schwarz_high} 
\end{align}
We can apply the assumption of the theorem, \eqref{eq:master_require_general_high}, to upper bound \eqref{eq:cauchy_schwarz_high} by $\Gamma$. 

Plugging this bound into \Cref{eq:ip_split_high}, we conclude the main bound of \Cref{thm:master_thm_high}: 
\begin{equation*}
W_1(p,q) = \sup_{1\text{-Lipschitz, smooth } f} \langle f, p-q\rangle \leq \sqrt{\frac{d \pi^d}{2}}\Gamma + \frac{36 d}{m}. \qedhere
\end{equation*}
\end{proof}

\begin{remark}[Efficient Recovery in $d$ Dimensions]
    We note that given sufficiently accurate Chebyshev moments, we can back out a distribution close to the original distribution in Wasserstein-$1$ distance. The \Cref{alg:weighted_regression} immediately generalizes to the $d$-dimensional setting; see \Cref{sec:eff_recover} for the details in $1$ dimension. We leave the details to the reader.
\end{remark}

We now give a constructive proof of the multivariate Jackson's theorem.

\subsection{Proof of Multivariate Jackson's Theorem}\label{sec:proof_jackson_high_dim}
We extend the $1$-dimensional constructive proof of Jackson's theorem in \citet{BravermanKrishnanMusco:2022} to $d$ dimensions. To prove the multivariate Jackson's theorem, we will use Fourier analysis. We first define the Fourier series of a function in $d$ dimensions. We start with a few standard preliminary definitions found in any standard text on Fourier analysis, such as \citet{stein2011fourier}.

Let $f: \R^d \to \R$ be an $\ell$-Lipschitz function, i.e., $\abs{f(x)-f(y)} \leq \ell \norm{x-y}_2$ $\forall x,y \in \R^d$. We say that $f \in L^2(\brac{-\pi,\pi}^d)$ if $\int_{\brac{-\pi,\pi}^d} \abs{f(x)}^2 \rd x < \infty$. 

\begin{definition}[Periodic Function]
    A function $f: \R^d \to \R$ is $2 \pi$ periodic if $f(x) = f(x + 2\pi K)$ for all $x \in \R^d$ and $K \in \Z^d$. Formally, this is known as coordinate-wise periodic, but we will refer to it as periodic for simplicity.
\end{definition}

\begin{definition}[Even Function] \label{def:even_high_dimn}
    Let $x = (x_1,\dots,x_d) \in \R^d$. The function $f: \R^d \to \R$ is even if $f(x_1,\dots,x_d) = f(\abs{x_1},\dots,\abs{x_d})$ for all $x \in \R^d$.
\end{definition}

\begin{definition}[Fourier Series]
    Let $f \in L^2(\brac{-\pi,\pi}^d)$ be a $2\pi$ periodic function. The function $f$ can be written via a Fourier series as:
    \begin{align*}
        f(x) = \sum_{K \in \Z^d} \hat{f}(K) e^{\im \inprod{k,x}}, \text{ where } \hat{f}(K) = \frac{1}{(2\pi)^d} \int_{\brac{-\pi,\pi}^d} f(x) e^{-\im \inprod{K,x}}  \rd x \mcom
    \end{align*}
    and $\im = \sqrt{-1}$. For $K \in \Z^d$, $\hat{f}(K)$ is called the Fourier coefficient of $f$. 
\end{definition}

\begin{claim}[Convolution Theorem]\label{clm:convolution}
    Let $f,g \in L^2([-\pi,\pi]^d)$ be $2\pi$-periodic functions with Fourier coefficients $\set{\hat{f}(K)}_{K \in \Z^d}$ and $\set{\hat{g}(K)}_{K \in \Z^d}$ respectively. Let $h$ be their convolution: 
    \[ h(x) \seteq [f * g](x) = \int_{[-\pi,\pi]^d} f(u) g(x-u) \rd u \mper \]
    The Fourier coefficients of $h$, $\set{\hat{h}(K)}_{K \in \Z^d}$, are given by:
    \[ \hat{h}(K) = (2\pi)^d \cdot  \hat{f}(K) \hat{g}(K) \mper \]
\end{claim}

We now build a multivariate version of the Jackson kernel, a key ingredient in the proof of the multivariate Jackson's theorem. \citet{BravermanKrishnanMusco:2022} define the Jackson kernel in one dimension, which we generalize to $d$-dimensions by just multiplying the one-dimensional Jackson kernel in each dimension.

\begin{definition}[Jackson Kernel]\label{def:Jackson}
    For $x_i \in \R$, $m \in \Np$, let $b_1: \R \to \R$ be the following function:
    \[ b_1(x_i) \seteq \paren{\frac{\sin(mx_i/2)}{\sin(x_i/2)}}^4 = 
    \sum_{k_1 = -2m+2}^{2m-2} \hat{b_1}(k_1) e^{\im k_1 x_i} \mcom \]
    where the Fourier coefficients $\hat{b_1}(-2m+2),\dots,\hat{b_1}(0),\dots, \hat{b_1}(2m-2)$ equals to:
    \begin{equation}\label{eq:1d_coeff}
        \hat{b_1}(-k_1) = \hat{b_1}(k_1) = \sum_{t = -m}^{m-k_1} (m-\abs{t}) \cdot (m - \abs{t + k_1}) \quad \text{for } k_1 = 0,\dots,2m-2 \mper
    \end{equation}
    Note that $ \hat{b_1}(0) \geq \dots \geq \hat{b_1}(2m-2)$. Let $\bm{0} = (0,\dots,0) \in \R^d$, $x = (x_1,\dots,x_d) \in \R^d$, and $b$ be the following trigonometric polynomial:
    \[ b(x_1,\dots,x_d) \seteq \prod_{i=1}^{d} b_1(x_i) = \prod_{i=1}^{d} \paren{\frac{\sin(mx_i/2)}{\sin(x_i/2)}}^4 = \sum_{K \in \set{ -2m+2,\dots,2m-2 }^d} \hat{b}(K) e^{\im \inprod{K,x}}  \mper \]
    From \Cref{eq:1d_coeff}, we have that $\hat{b}(\bm{0}) \geq \hat{b}(K)$, for $K \neq \bm{0}$.
\end{definition}

We also need the following fact from \citet{BravermanKrishnanMusco:2022} about Jackson's kernel.

\begin{fact}[{\citep[Theorem C.5]{BravermanKrishnanMusco:2022}}] \label{lem:thmc5}
For $x_i \in \R$, $m \in \Np$, the one-dimensional Jackson's Kernel $b_1$, defined in \Cref{def:Jackson}, satisfies the following
\[ \frac{\int_{0}^{\pi} x_i b_1(x_i) \rd x_i }{\int_{0}^{\pi} b_1(x_i) \rd x_i} \leq \frac{8.06}{m} \mper \]
    
\end{fact}

We are not ready to prove that a truncated and ``damped'' Fourier series of $f$ is a \emph{good} uniform approximation to $f$. 

\begin{theorem}\label{thm:high_jacksons}
    Let $f:\R^d \to \R$ be a $\ell$-Lipschitz continuous, $2\pi$-periodic function. For $m \in \Np$, let $b: \R^d \to \R$ be the Kernel from \Cref{def:Jackson}. The function $\tilde{f}(x) = \frac{1}{ \hat{b}(\bm 0) (2\pi)^d } \int_{[\pi,\pi]^d} b(u) f(x-u) \rd u$ satisfies: 
    \[ \norm{\tilde f- f}_\infty \leq \frac{9 \ell d}{m} \mper \]
    Moreover, the Fourier coefficients of $\tilde f$, $\set{\hat{\tilde f}(K)}_{K \in \Z^d}$, are given by:
    \[ \hat{\tilde f}(K) = \frac{\hat b(K)}{\hat b(\bm 0)} \hat{f}(K)  \mcom\]
    where $\hat b(\bm 0) \geq \hat b(K)$ for all $K \in \Z^d$, and for $K \not \in \set{-2m+2,\dots,2m-2}^d$, we have that $\hat{\tilde f}(K) = 0$. 
\end{theorem}

\begin{proof}
    Let $u = (u_1,\dots,u_d) \in \R^d$. Note that $\hat{b}(\bm 0) = \frac{1}{(2\pi)^d}\int_{[-\pi,\pi]^d} b(x) \rd x$, whence we get that $\frac{1}{\hat{b}(\bm 0) (2\pi)^d }\int_{[-\pi,\pi]^d} b(u) \rd u = 1$. Therefore, by the definition of $\tilde f$, we get that: 
    \[ \abs{\tilde f(x) - f(x)} \leq \int_{[-\pi,\pi]^d} \frac{1}{ \hat{b}(\bm 0) (2\pi)^d } b(u) \cdot \abs{ f(x) - f(x-u) } \rd u \mper\]
    Since $f$ is $\ell$-Lipschitz, we have that $\abs{ f(x) - f(x-u) } \leq \ell \norm{u}_2$. Therefore, we get that:
    \begin{align*}
        &\max_x \abs{\tilde f(x) - f(x)} \\
        & \quad \leq \int_{[-\pi,\pi]^d} \frac{1}{ \hat{b}(\bm 0) (2\pi)^d } b(u) \cdot \ell \norm{u}_2 \rd u \\
        & \quad \leq \int_{[-\pi,\pi]^d} \frac{1}{ \hat{b}(\bm 0)(2\pi)^d} b(u) \cdot \ell \norm{u}_1 \rd u & \paren{\because \norm{\cdot}_2 \leq \norm{\cdot}_1} \\
        & \quad = \frac{\ell}{ \hat{b}(\bm 0)(2\pi)^d}  \cdot  \sum_{i=1}^{d} \paren{\int_{-\pi}^{\pi} \abs{u_i}  b_1(u_i) \rd u_i} \cdot \paren{\prod_{j \in [d], j \neq i} \int_{-\pi}^{\pi} b_1(u_j) \rd u_j} & \paren{\because \norm{u}_1 = \sum_{i=1}^d \abs{u_i}} \\
        & \quad = \ell \sum_{i=1}^{d} \frac{\int_{-\pi}^{\pi} \abs{u_i}  b_1(u_i) \rd u_i}{\int_{-\pi}^{\pi} b_1(u_i) \rd u_i} \cdot \prod_{j \in [d], j \neq i} \frac{\int_{-\pi}^{\pi} b_1(u_j) \rd u_j}{\int_{-\pi}^{\pi} b_1(u_j) \rd u_j} &\paren{\text{By \Cref{def:Jackson}}} \\
        & \quad = \ell \sum_{i=1}^{d} \frac{\int_{0}^{\pi} u_i  b_1(u_i) \rd u_i}{\int_{0}^{\pi} b_1(u_i) \rd u_i} \leq (8.06)d \frac{\ell}{m} \mcom 
    \end{align*}
    where the last inequality follows from \Cref{lem:thmc5}.
    We now reason about the Fourier coefficients of $\tilde f$. Note that for $K \not\in \set{-2m+2,\dots,2m-2}^d$, we have that $\hat{\tilde f}(K) = 0$. For $K \in \set{-2m+2,\dots,2m-2}^d$, we have by the convolution theorem (\Cref{clm:convolution}) that:
    \[ \hat{\tilde{f}}(K) = \frac{\hat b(K)}{\hat b(\bm 0)} \cdot \hat f(K) \mper \]
    Using the fact from \Cref{def:Jackson} that $\hat b(\bm 0) \geq \hat b(K)$,  for $K \neq \bm 0$, and the fact that $\hat b(K) = 0$ for $K \not\in \set{-2m+2,\dots,2m-2}^d$, we get the desired result.
\end{proof}

We now prove the multivariate Jackson's theorem for a smooth, $\ell$-Lipschitz function $f:[-1,1] \to \R$. To do so, we construct a mapping to a periodic function  $h$  with period  $2\pi$  and then apply the previous theorem.

\jacksonhighdim*

\begin{proof}[Proof of \Cref{thm:jackson_high_dim}]
    Let $(\cos \theta_1,\dots,\cos \theta_d) \in [-1,1]^d$. We will use the identity that for $K = (k_1,\dots,k_d) \in \N^d$, 
    \[ T_K(\cos \theta_1,\dots,\cos \theta_d) = \prod_{i=1}^d \cos(k_i \theta_i) \mper \]  
    Consider the Lipschitz continuous function $f: [-1,1]^d \to \R$ with Chebyshev expansion coefficients $c_{K}$ for $K \in \N^d$, where $c_K = \inprod{r, W\cdot \T_K}$. We transform $f$ into a periodic function as follows: For $\Theta = (\theta_1,\dots,\theta_d) \in [-\pi,0]^d$, let $g(\Theta) = f(\cos \theta_1,\dots,\cos \theta_d)$ and let $h(\Theta) = g(-\abs{\theta_1},\dots, -\abs{\theta_d})$ for $\Theta \in [-\pi,\pi]^d$. The function $h: [\pi,\pi]^d \to \R$ is a periodic and even function (\Cref{def:even_high_dimn}). Since the function is even, one can check that its Fourier series can be written as follows: For  $\Theta = (\theta_1,\dots,\theta_d) \in [-\pi,\pi]^d$,
    \[ h(\Theta) = \sum_{K = (k_1,\dots,k_d) \in \N^d} \alpha_K \prod_{i=1}^{d} \cos(k_i \theta_i) \mcom \]
    where:
    \[ \alpha_K = \frac{2^{\nnz{K}}}{(2\pi)^d} \int_{[-\pi,\pi]^d} h(\Theta) \paren{\prod_{i=1}^d \cos(k_i \theta_i)} \rd \Theta = \frac{2^{\nnz{K}}}{\pi^d} \int_{[-\pi,0]^d} g(\Theta) {\prod_{i=1}^d \cos(k_i \theta_i)}  \rd \Theta \mper  \]
    Let $x = (x_1,\dots,x_d) \in [-1,1]^d$. Using that fact that for $i \in [d]$,  $\frac{\rd}{\rd x_i} \cos^{-1}(x_i) = \frac{1}{\sqrt{1-x_i^2}}$, we get 
    \[ \int_{[-\pi,0]^d} g(\Theta) {\prod_{i=1}^d \cos(k_i \theta_i)}  \rd \Theta = \int_{[-1,1]^d} f(x) T_{K}(x) W(x) \rd x \mper  \]
    Using the above equations, we conclude that for $K \in \N^d$, the Fourier coefficients of $h$ are just a scaling depending on $K$ of the Chebyshev coefficients of $f$, and we get by \Cref{def:norm_multi_cheb} that: 
    \begin{equation}\label{eq:c_k_alpha_k_relation}
        \sqrt{\frac{2^{\nnz{K}}}{\pi^d}} c_K = \alpha_K \mper
    \end{equation}
    
    We observe that the mapping from  $f$  to  $h$  preserves the $\ell$-Lipschitz property. The function $h: [\pi,\pi]^d \to \R$ is periodic and an even function (\Cref{def:even_high_dimn}), and is $\ell$-Lipschitz. Let $\tilde h$ be the function obtained by applying Jackson's theorem (\Cref{thm:high_jacksons}) to $h$. We know that $\tilde h$ is an even function since $h$ is even and Jackson's Kernel $b$, which $h$ is convolved with, is also even. Recall that the Fourier series coefficients of $\tilde h$, denoted by $\hat{\tilde h}(K)$, are $0$ for $K \not\in \set{0,\dots,2m-2}^d$. Finally, let $\tilde f: [-1,1]^d \to \R$ be defined as 
    \[ \tilde f(\cos \theta_1,\dots,\cos\theta_d) \seteq \tilde h(\theta_1,\dots,\theta_d) \mper \]  
    By \Cref{eq:c_k_alpha_k_relation}, we get that the Chebyshev coefficients of $\tilde f$ are exactly $\frac{\hat{b}(K)}{\hat{b}(\bm 0)}c_K$. Note from \Cref{thm:high_jacksons} that $\abs{\frac{\hat{b}(K)}{\hat{b}(\bm 0)}} \leq 1$, therefore, we get the $\tilde f$ is the damped Chebyshev truncated series of $f$. 

    Moreover, we have that $\norm{\tilde f - f}_\infty = \norm{\tilde h - h}_\infty \leq \frac{9 \ell d}{m}$, where the inequality follows from \Cref{thm:high_jacksons}. This completes the proof.
\end{proof} 

\section{Differentially Private Synthetic Data Generation in Higher Dimensions} \label{sec:dp_high}

\begin{algorithm}[tb]\caption{$d$-Dimension Private Chebyshev Moment Matching} \label{alg:dp_algorithm_high}
    \begin{algorithmic}[1]
        
            \Require Dataset $\x_1, \ldots, \x_n \in [-1,1]^d$, privacy parameters $\epsilon, \delta > 0$. 
            \Ensure A probability distribution $q$ approximating the distribution, $p \seteq \text{Unif}\set{\x_1, \ldots, \x_n}$. 
            \State Let $\mathcal{G} = \Set{-1,-1 + \frac{1}{\lceil (\e n)^{1/d}\rceil},-1 + \frac{2}{\lceil(\e n)^{1/d}\rceil}, \ldots, 1}^d$. Let $r \seteq (2\ceil{(\e n)^{1/d}}+1)$ and for $J =(j_1,\dots,j_d) \in [r]^d$ let ${g}_J = \paren{-1 + \frac{j_1-1}{\ceil{(\e n)^{1/d}}},\dots,-1+\frac{j_d-1}{\ceil{(\e n)^{1/d}}}}$ denote the $J^\text{th}$ element of $\mathcal{G}$.
            \State For $i = 1, \ldots, n$, let $\tilde{\x}_i = \argmin_{\bm{y}\in \mathcal{G}} |\x_i -\bm{y}|$. I.e., round $\x_i$ to the nearest point in the grid $\cG$.
            \State Set $\sigma^2 = \frac{4\cdot 2^{d}}{\pi^d} \normsum\ln(1.25 / \delta) / (n^2\e^2)$, where $\normsum = \sum_{K \in \set{0,\dots,m}^d \setminus \bm{0}} \frac{1}{\norm{K}_2}$. See \Cref{lem:normsum} for the bound on $\normsum$.
            \State Set $m = \ceil{2(\e n)^{1/d}}$. For $K \in \set{0,\dots,m}^d \setminus \bm{0}$, let $\hat{m}_K = \eta_K +  \frac{1}{n}\sum_{i=1}^{n} \T_K(\tilde{\x}_i)$, where $\eta_K \sim \mathcal{N}(0,\norm{K}_2\sigma^2)$. 
            \State Let $\set{{q}_J}_{J \in [r]^d}$ be the solution to the following optimization problem: 
            \begin{equation*}
                \begin{array}{ll@{}ll}
                \min_{ \set{{z}_J}_{J \in [r]^d}} &  \displaystyle\sum_{ K \in \set{0,\dots,m}^d \setminus \bm{0}} \frac{1}{\norm{K}_2^2} \paren{\hat m_K - \sum_{J \in [r]^d} {z}_J \T_K(g_J)}^2 \\
                \text{subject to} &\displaystyle \sum_{J \in [r]^d}  z_J = 1 \text{ and } z_J \geq 0, \,\, \forall J \in [r]^d.
                \end{array}
            \end{equation*}
            \State Return ${q} = \sum_{J \in [r]^d} {q}_J \delta(x-g_J)$, where $\delta$ is the Dirac delta function.
    \end{algorithmic}
    \end{algorithm}

\begin{theorem}
  Let $X = \{\x_1, \ldots, \x_n\}$ be a dataset with each $\x_j\in [-1,1]^d$, for $d \geq 2$. Let $p$ be the uniform distribution on $X$. For any $\epsilon,\delta \in (0,1)$, there is an $(\e, \delta)$-differentially private algorithm based on Chebyshev moment matching that, in $\poly\paren{n,\e,\delta,2^d}$ time, returns a distribution $q$ satisfying for a fixed constant $c_4$,
    \begin{align*}
          \E[W_1(p,q)] \leq c_4 d \paren{\frac{1+ \ln^{0.5}(1.25/\delta)}{n \e}}^{1/d} \mper 
    \end{align*}
\end{theorem}

    \begin{proof}
    We analyze both the privacy and accuracy of \Cref{alg:dp_algorithm_high}. 
    ~
    \paragraph{Privacy.}  For a dataset $X = \{\x_1, \ldots, \x_n\} \in [-1,1]^{d \times n}$, where each data-point $\x_i \in [-1,1]^d$. Let $f(X)$ be a vector-valued function mapping to the $K \in \set{0,\dots,m}^d \setminus \bm{0}$ \emph{scaled} Chebyshev moments of the uniform distribution over $X$. I.e.,
    \begin{align*}
        f(X)_K = \frac{1}{\sqrt{\norm{K}_2}} \cdot \frac{1}{n} \sum_{i=1}^{n} \T_K(\x_i) \mcom
    \end{align*}
    where $f(X)_K$ denotes the $K$-th entry of the vector $f(X)$, and $\T_K(\x)$ is the $K$-th normalized multivariate Chebyshev polynomial. 
    
    We will show that \Cref{alg:dp_algorithm_high} is $(\e,\delta)$-differentially private and that the output of the algorithm is close in Wasserstein distance to the true moments of the uniform distribution over $X$.
    By \Cref{def:norm_multi_cheb}, $\max_{\x_i \in [-1,1]^d} |\T_K(x_i)| \leq \sqrt{2^{\nnz{K}}/\pi^d}$ for $K \in \N^d$, so we have: 
    \begin{align}
    \Delta_{2,f}^2 = \underset{X,X'\in \mathcal{X}}{\max_{\text{neighboring datasets}}} \| f(X) - f(X') \|_2^2 & \leq \sum_{ K \in \set{0,\dots,m}^d \setminus \bm 0} \frac{1}{\norm{K}_2}\cdot \frac{1}{n^2} \cdot \frac{4 \cdot 2^{\nnz{K}}}{\pi^d} \nonumber \\
    & \leq \frac{4\cdot 2^d}{\pi^d n^2} \sum_{ K \in \set{0,\dots,m}^d \setminus \bm 0} \frac{1}{\norm{K}_2} \nonumber \\
    & = \frac{4\cdot 2^{d}}{\pi^d n^2} \normsum \mcom \label{eq:sensitivity_high}
    \end{align}
    where $\normsum =  \sum_{ K \in \set{0,\dots,m}^d \setminus \bm 0} \frac{1}{\norm{K}_2}$.
    For two neighboring datasets $X,X'$, let $\tilde X$ and $\tilde X'$ be the rounded datasets computed in line 2 of \Cref{alg:dp_algorithm_high} -- i.e., $\tilde{X} = \{\tilde{\x}_1, \ldots, \tilde{\x}_{n}\}$. Observe that $\tilde X$ and $\tilde X'$ are also neighboring. Thus, it follows from \Cref{def:gaussian_mechanism} and the sensitivity bound of \Cref{eq:sensitivity_high} that $\tilde{m}_K = f(\tilde{X})_K + \eta_K$ is $(\e,\delta)$-differentially private for $\eta_K \sim \mathcal{N}(0,\sigma^2)$ as long as $\sigma^2 = \frac{4\cdot 2^{d}}{\pi^d}\normsum \ln(1.25 / \delta) / (n^2\e^2)$. Finally, observe that $\hat{m}_K$ computed by \Cref{alg:dp_algorithm_high} is exactly equal to $\sqrt{\norm{K}_2}$ times the $K^\text{th}$ entry of such an $\tilde{m}$. So $\set{\hat{m}_K}_{K \in \set{0,\dots,m^d}\setminus \bm{0}}$ are $(\e,\delta)$-differentially private. Since the remainder of \Cref{alg:dp_algorithm} simply post-processes $\set{\hat{m}_K}_{K \in \set{0,\dots,m^d}\setminus \bm{0}}$ without returning to the original data $X$, the output of the algorithm is also $(\e,\delta)$-differentially private, as desired.
    ~
    \paragraph{Accuracy.} The \Cref{alg:dp_algorithm_high} begins by rounding the dataset $X$ so that every coordinate of every data point is a multiple of $1/\ceil{(\e n)^{1/d}}$. Let $\tilde{p}$ be the uniform distribution over the rounded dataset $\tilde{X}$. Then, it is not hard to see from the transportation definition of the Wasserstein-$1$ distance that:
    \begin{equation}\label{eq:rounding_error_high}
        W_1(p,\tilde{p}) \leq \frac{d}{2\ceil{(\e n)^{1/d}}} \mper
    \end{equation}
    In particular, we can transport $p$ to $\tilde{p}$ by moving every unit of $1/n$ probability mass a distance of at most $1/2\ceil{(\e n)^{1/d}}$, along each of the $d$ coordinates. Given \cref{eq:rounding_error_high}, it will suffice to show that \Cref{alg:dp_algorithm_high} returns a distribution $q$ that is close in Wasserstein distance to $\tilde{p}$. We will apply triangle inequality to bound $W_1(p,q)$.

    To show that \Cref{alg:dp_algorithm} returns a distribution $q$ that is close to $\tilde{p}$ in Wasserstein distance, we begin by bounding the moment estimation error: 
    \begin{align*}
        E \defeq \sum_{K \in \set{0,\dots,m}^d \setminus \bm{0}} \frac{1}{\norm{K}_2^2} \left(\hat{m}_K(\tilde p) - \langle \tilde{p}, \T_K \rangle\right)^2, 
    \end{align*}
    where $m$ is as chosen in \Cref{alg:dp_algorithm_high} and $\langle \tilde{p}, T_K \rangle = \frac{1}{n}\sum_{i=1}^{n} \T_K(\tilde{\x}_i)$. Let $\sigma^2$ and $\set{\eta_K}_{K \in \set{0,\dots,m}^d \setminus \bm{0}}$ be as in \Cref{alg:dp_algorithm_high}. Applying linearity of expectation, we have that:
    \begin{align}
    \label{eq:basic_exp_bound_high}
    \E[E] = \E\left[\sum_{K \in \set{0,\dots,m}^d \setminus \bm{0}} \frac{1}{\norm{K}_2^2} \eta_K^2\right] &= \sum_{K \in \set{0,\dots,m}^d \setminus \bm{0}} \frac{1}{\norm{K}_2^2} \E\left[\eta_K^2\right] \nonumber \\
    & = \sum_{K \in \set{0,\dots,m}^d \setminus \bm{0}} \frac{1}{\norm{K}_2^2}\cdot \norm{K}_2 \sigma^2 \nonumber \\ 
    & \leq\sum_{K \in \set{0,\dots,m}^d \setminus \bm{0}}  \frac{1}{\norm{K}_2} \sigma^2 = \sigma^2 \normsum \mcom 
    \end{align}
    where we recall that $\normsum = \sum_{K \in \set{0,\dots,m}^d \setminus \bm{0}}  \frac{1}{\norm{K}_2}$.
    
    Now, let $q$ be as in \Cref{alg:dp_algorithm_high}. Using a triangle inequality argument as in \Cref{sec:eff_recover}, we have:
    \begin{align*}
        \Gamma^2 & = \sum_{K \in \set{0,\dots,m}^d \setminus \bm{0}} \frac{1}{\norm{K}_2^2} \left(\langle {q}, \T_K \rangle - \langle \tilde{p}, \T_K \rangle\right)^2 \\
        & \leq \sum_{K \in \set{0,\dots,m}^d \setminus \bm{0}} \frac{1}{\norm{K}_2^2} \left(\langle {q}, \T_K \rangle - \hat{m}_j\right)^2 + \sum_{K \in \set{0,\dots,m}^d \setminus \bm{0}} \frac{1}{\norm{K}_2^2} \left(\langle \tilde{p}, \T_K \rangle - \hat{m}_j\right)^2 \leq 2E.
    \end{align*}
    Above we use that $\tilde{p}$ is a feasible solution to the optimization problem solved in \Cref{alg:dp_algorithm_high} and, since $q$ is the optimum,  $\sum_{K \in \set{0,\dots,m}^d \setminus \bm{0}} \left(\langle {q}, \T_K \rangle - \hat{m}_j\right)^2 \leq \sum_{K \in \set{0,\dots,m}^d \setminus \bm{0}} \frac{1}{\norm{K}_2^2} \left(\langle \tilde{p}, \T_K \rangle - \hat{m}_j\right)^2$. 
    It follows that $\E[\Gamma^2] \leq 2\E[E]$, and, via Jensen's inequality, that $\E[\Gamma] \leq \sqrt{2\E[E]}$. Plugging into \Cref{thm:master_thm_high}, we have for constant $c$: 
    \begin{align}
    \label{eq:final_triangle_high}
    \E[W_1(\tilde{p},q)] &\leq \sqrt{\frac{d \pi^d}{2}}\E[\Gamma] + \frac{cd}{m} \nonumber \\
     & \leq \sqrt{\frac{d \pi^d}{2}} \sqrt{2 \normsum \sigma^2} + \frac{cd}{m} \mper
    \end{align}
    From the bound on $\sigma^2$ computed above, and from the upper bound on $\normsum$ in \Cref{lem:normsum}, we get that 
    \begin{align*}
        \sqrt{\frac{d \pi^d}{2}} \sqrt{2 \normsum \sigma^2} & \leq \sqrt{\frac{d \pi^d}{2}} \sqrt{\frac{8 \ln(1.25/\delta) \cdot 2^d}{\pi^d}} \cdot \frac{S}{n\e} \\
        & \leq \sqrt{\frac{d \pi^d}{2}} \sqrt{\frac{8 \ln(1.25/\delta) \cdot 2^d}{\pi^d}} \cdot \frac{4 (\pi e)^{d/2}}{2^d} \cdot \frac{m^{d-1}}{d n \e} \\ & = 8 \sqrt{\ln(1.25/\delta)} (\pi e /2)^{d/2} \cdot \frac{m^{d-1}}{\sqrt{d} n \e} \mper
    \end{align*}
    Therefore, for some absolute constant $c_2$, setting $m = c_2 \paren{\frac{ dn \e  }{\ln^{0.5}(1.25/\delta)}}^{\frac{1}{d}}$, we get from \Cref{eq:final_triangle_high}, 
    \begin{align*}
        \E[W_1(\tilde{p},q)] & \leq c_3 \cdot d \cdot \paren{\frac{\ln^{0.5}(1.25/\delta)}{ n \e}}^{1/d} \mcom
    \end{align*}
    for an absolute constant $c_3$. By triangle inequality $W_1(p,q) \leq W_1(p,\tilde p)+ W_1(\tilde p, q)$ and using the bound on  $W_1(p,\tilde p)$ in \Cref{eq:rounding_error_high}, we get that for an absolute constant $c_4$,
    \[  \E[W_1(p,q)] \leq c_4 d \paren{\frac{1+ \ln^{0.5}(1.25/\delta)}{n \e}}^{1/d} \mper  \]
    ~
    \paragraph{Runtime. } The number of points in the grid in \Cref{alg:dp_algorithm_high} is upper bounded by $\abs{\cG} = (1+ 2\ceil{\e n}^{1/d})^d = \bigo{ 2^d \ceil{n\e}^{1/d}}$. The number of Chebyshev moments we calculate is less than $(m+1)^d = \bigo{2^d \ceil{(dn\e)}^{1/d}}$. Since the optimization problem runs in polynomial time in its variables and constraints, we get that the running time of the algorithm is bounded by $\poly\paren{n,\e,\delta,2^d}$.
\end{proof}

\begin{remark}[Comparison to \citet{boedihardjo2022private}]
    We remark that \citet{boedihardjo2022private} use the $\ell_\infty$ metric instead of the $\ell_2$ metric for the Wasserstein-$1$ distance, and they achieve an error of $\E[W_1(p,q)] \leq \bigo{\frac{\log^{1.5}(\e n)}{\e n}}^{1/d}$. Since the Wasserstein-$1$ distance in the $\ell_\infty$ metric is bounded by $1$, their bound is non-vacuous for $d = \bigo{\log n}$. For $d = \bigo{\log n}$, our bound matches their bound to $\log(n)$-factors.  
\end{remark}

Finally, we show how to upper bound $\normsum$. 
\begin{lemma}\label{lem:normsum}
    Let $m \in \Np$ and $d \geq 2$. Then, we have that 
    \[ \normsum \seteq \sum_{K \in \set{0,\dots,m}^d \setminus \bm 0} \frac{1}{\norm{K}_2} \leq \frac{4 (\pi e)^{d/2}}{2^d} \cdot \frac{m^{d-1}}{d} \mper \]
\end{lemma}
\begin{proof}
    Note that the function $\frac{1}{\norm{K}_2}$ is decreasing in $\norm{K}_2$. Moreover, for $K \in \set{0,\dots,m}^d \setminus \bm 0$, $\norm{K}_2 \geq 1$. Thus, we can upper bound the sum by the integral as follows: 
    \begin{equation*}
        \sum_{K \in \set{0,\dots,m}^d \setminus \bm 0} \frac{1}{\norm{K}_2} \leq  d + \int_{K \in [0, m]^d, \norm{K}_2 \geq 1} \frac{1}{\norm{K}_2} \rd K \mper
    \end{equation*}
    Using the fact that for $K \in \set{0,\dots,m}^d \setminus \bm 0$, $\norm{K}_2 \leq m \sqrt{d}$, and that the $m\sqrt{d}$ ball contains the hypercube $[0,m]^d$, we get that 
    \[ \int_{K \in [0, m]^d, \norm{K}_2 \geq 1} \frac{1}{\norm{K}_2} \rd K \leq \frac{1}{2^d} \int_{ 1\leq \norm{K}_2 \leq m\sqrt{d}} \frac{1}{\norm{K}_2} \rd K \mcom \]
    where the factor of $1/2^d$ comes due to symmetry and the fact that the set $K \in [0,m]^d$, i.e., the vectors in $K$ contains non-negative entries, which is only $1/2^d$ fraction of vectors in the set $\set{K : 1 \leq \norm{K}_2 \leq m \sqrt{d}}$. To compute the integral, we can transition to polar coordinates. Let $\norm{K}_2 = r$, then we get that  
    \[ \int_{ 1\leq \norm{K}_2 \leq m\sqrt{d}} \frac{1}{\norm{K}_2} \rd K = \int_{r=1}^{m \sqrt{d}} \int_{\Omega} \frac{1}{r} \cdot r^{d-1} \rd \Omega \rd r \mcom  \]
    where $\Omega$ is the \emph{angular domain} in spherical coordinates. Separating the terms in the integral gives,
    \[ \int_{r=1}^{m \sqrt{d}} \int_{\Omega} \frac{1}{r} \cdot r^{d-1} \rd \Omega \rd r  = \int_{r=1}^{m \sqrt{d}}  r^{d-2} \rd r  \int_{\Omega} \rd \Omega \mper\]
    Using the fact that $\int_{\Omega} \rd \Omega$ is the surface area of the unit sphere in $d$-dimensions, we can use the equation for the surface area of the unit sphere, see e.g. \cite{boyd2004convex}, to evaluate the integral as
    \[ \int_{r=1}^{m \sqrt{d}}  r^{d-2} \rd r \cdot  \int_{\Omega} \rd \Omega = \frac{(m\sqrt{d})^{d-1} - 1}{d-1} \cdot \frac{2 \pi^{d/2}}{\Gamma(d/2)} \mcom \]
    where $\Gamma$ is the gamma function. Combining all the terms, we get that 
    \[ \sum_{K \in \set{0,\dots,m}^d \setminus \bm 0} \frac{1}{\norm{K}_2} \leq d+ \frac{1}{2^d} \cdot  \frac{(m\sqrt{d})^{d-1} - 1}{d-1} \cdot \frac{2 \pi^{d/2}}{\Gamma(d/2)} \mper  \]
    By Stirling's approximation $\Gamma(d/2) \geq \sqrt{2} (d/2)^{d/2-1/2} e^{-d/2}$, and observing that the second term dominates in the RHS above, we have 
    \[ \sum_{K \in \set{0,\dots,m}^d \setminus \bm 0} \frac{1}{\norm{K}_2} \leq \frac{4 (\pi e)^{d/2}}{2^d} \cdot \frac{m^{d-1}}{d} \mper \qedhere \]
\end{proof}

\section{Accuracy of Generic Moment Matching Algorithm}
\label{app:corr_proof}
In this section, we give the full proof of \Cref{corr:recovery}, which establishes the accuracy of the generic Chebyshev moment regression algorithm (\Cref{alg:weighted_regression}). We require the following basic property about the Chebyshev nodes:
\begin{lemma}[Chebyshev Node Approximation] \label{lem:cheb_nodes}
    Let $\cC = \set{x_1,\dots,x_{g}}$ be the degree $g$ Chebyshev nodes. I.e., $x_i = \cos\paren{\frac{2i-1}{2g} \pi}$. Let $r_{\cC}: [-1,1] \to \cC$ be
    a function that maps a point $x \in [-1,1]$ to the point $y \in \cC$ that minimizes $\abs{\cos^{-1}(x) - \cos^{-1}(y)}$, breaking ties arbitrarily. For any $x\in [-1,1]$, $\abs{\cos^{-1}(x) - \cos^{-1}(r_{\cC}(x))} \leq \frac{\pi}{2g}$.
\end{lemma}    
\begin{proof}
    For any two consecutive points $x_i, x_{i+1}$ in the $\cC$, 
    \begin{align*}
    \abs{\cos^{-1}(x_i) - \cos^{-1}(x_{i+1})} = \frac{\pi}{g}.
    \end{align*}
    Since $\cos^{-1}(x)$ is non-increasing, for any $x \in [x_{i+1}, x_i]$, $\cos^{-1}(x)\in [\cos^{-1}(x_i),\cos^{-1}(x_{i+1})]$. So, $\cos^{-1}(x)$ has distance at most $\frac{\pi}{2g}$ from either $\cos^{-1}(x_i)$ or $\cos^{-1}(x_{i+1})$. Additionally, we can check that $\abs{\cos^{-1}(x) - \cos^{-1}(x_1)} \leq \frac{\pi}{2g}$ for any $x < x_1$ and  $\abs{\cos^{-1}(x) - \cos^{-1}(x_g)} \leq \frac{\pi}{2g}$ for any $x> x_g$.
\end{proof}
With \Cref{lem:cheb_nodes} in place, we are ready to prove  \Cref{corr:recovery}.
\begin{proof}[Proof of \Cref{corr:recovery}]
Let $\cC$ and $r_{\cC}:[-1,1] \to \cC$ be as in \Cref{lem:cheb_nodes}. For $i \in \set{1,\dots,g}$, let $Y_i$ be the set of points in $[-1,1]$ that are closest to $x_i \in \cC$, i.e., $Y_i = \set{x \in [-1,1]: r_{\cC}(x) = x_i}$. Let $\tilde p$ be a distribution supported on the set $\cC$ with mass $\int_{Y_i} p(x) \, dx$ on $x_i \in \cC$. For all $j\in 1, \ldots, k$ we have:
\begin{align}
  \abs{\langle p, \T_j \rangle - \langle \tilde p, \T_j\rangle } 
   & = \abs{\sum_{i=1}^{g} \int_{Y_i} \T_j(x) p(x) \, dx- \paren{\int_{Y_i} p(x) \, dx} \T_j(x_i)} \nonumber \\
   & = \abs{\sum_{i=1}^{g} \paren{\int_{Y_i} p(x) \, dx} \T_j(y_i) - \paren{\int_{Y_i} p(x) \, dx} \T_j(x_i)} \qquad \paren{\text{for some $y_i \in Y_i$}} \nonumber \\ 
   & \leq \sum_{i=1}^{g} \paren{\int_{Y_i} p(x) \, dx}  \abs{\T_j(y_i) - \T_j(x_i)} \nonumber\\
  & = \sum_{i=1}^{g} \paren{\int_{Y_i} p(x) \, dx} \cdot \sqrt{\frac{2}{\pi}} \cdot \abs {\cos(j \cos^{-1}(y_i)) -  \cos(j \cos^{-1}(x_i))} \nonumber\\
  & \leq \sum_{i=1}^{g}\paren{\int_{Y_i} p(x) \, dx} \cdot \sqrt{\frac{2}{\pi}} \cdot \frac{j\pi}{2g}  = \frac{j \sqrt{\pi/2}}{g} \label{eq:p_tilde_p} 
\end{align}
The second equality follows from the intermediate value theorem. The first inequality follows by triangle inequality. The third equality follows by the trigonometric definition of the (normalized) Chebyshev polynomials. The second inequality follows from \Cref{lem:cheb_nodes} and the fact that the derivative of $\cos(j x)$ is bounded by $j$. The bound in \eqref{eq:p_tilde_p} then yields:
\begin{align}
\left(\sum_{j=1}^k \frac{1}{j^2} \paren{\langle p, \T_j \rangle - \langle \tilde p, \T_j \rangle}^2\right)^{1/2} \leq \frac{\sqrt{\pi k/2}}{g} \label{eq:mhat_tilde_p}.
\end{align}
Observe also that, since $\tilde{p}$ is supported on $\cC$, it is a valid solution to the optimization problem solved by \Cref{alg:weighted_regression}. Accordingly, we have that:
\begin{align}
\label{eq:compare_to_opt}
    \left(\sum_{j=1}^k \frac{1}{j^2} \paren{\hat{m}_j - \langle {q}, \T_j \rangle}^2\right)^{1/2} \leq \left(\sum_{j=1}^k \frac{1}{j^2} \paren{\hat{m}_j - \langle \tilde{p}, \T_j \rangle}^2\right)^{1/2}
\end{align}
Applying triangle inequality, followed by \eqref{eq:compare_to_opt}, triangle inequality again, and finally \eqref{eq:mhat_tilde_p}, we have:
\begin{align*}
\left(\sum_{j=1}^k \frac{1}{j^2} \paren{\langle p, \T_j \rangle - \langle {q}, \T_j \rangle}^2\right)^{1/2} \hspace{-.7em} &\leq \left(\sum_{j=1}^k \frac{1}{j^2} \paren{\langle p, \T_j \rangle - \hat{m}_j}^2\right)^{1/2} + \left(\sum_{j=1}^k \frac{1}{j^2} \paren{\hat{m}_j - \langle {q}, \T_j \rangle}^2\right)^{1/2}\\
& \leq \left(\sum_{j=1}^k \frac{1}{j^2} \paren{\langle p, \T_j \rangle - \hat{m}_j}^2\right)^{1/2} + \left(\sum_{j=1}^k \frac{1}{j^2} \paren{\hat{m}_j - \langle \tilde{p}, \T_j \rangle}^2\right)^{1/2}\\
& \leq 2\left(\sum_{j=1}^k \frac{1}{j^2} \paren{\langle p, \T_j \rangle - \hat{m}_j}^2\right)^{1/2} + \left(\sum_{j=1}^k \frac{1}{j^2} \paren{\langle {p}, \T_j \rangle - \langle \tilde{p}, \T_j \rangle}^2\right)^{1/2} \\
&\leq 2\Gamma + \frac{\sqrt{2\pi k}}{g}.
\end{align*}
Setting $g  = \lceil k^{1.5} \rceil$, we can apply \Cref{thm:master_thm} to conclude that, for a fixed constant $c'$,
\begin{align*}
W_1(p,q) \leq \frac{c}{k} + 2\Gamma + \frac{\sqrt{\pi/2}}{k} \leq c' \cdot \left(\frac{1}{k} + \Gamma\right). &\qedhere
\end{align*}
\end{proof}

\section{High Probability Bound for Private Synthetic Data}\label{app:concentration}
In this section, we prove the high probability bound on Wasserstein distance stated in \Cref{thm:synth_data}, which follows from a standard concentration bound for sub-exponential random variables \cite{wainwright_2019}. We recall that a random variable $X$ is subexponential with parameters $(\nu, \alpha)$ if:
\begin{align*}
\E[e^{\lambda(X-\E[X])}]&\leq e^{\nu^2\lambda^2/2} & &\text{for all} & |\lambda| &\leq \frac{1}{\alpha}.
\end{align*}
We require the following well-known fact that a chi-square random variable with one degree of freedom is subexponential:
\begin{fact}[Sub-Exponential Parameters {\cite[Example 2.8]{wainwright_2019}}] \label{lem:sub_exp_of_gau_sq}
    Let $\eta \sim \cN(0,\sigma^2)$. Then, $\eta^2$ is sub-exponential random variable with parameters $(2\sigma^2, 4\sigma^2)$.
\end{fact}
% \begin{proof} \begin{align*}\E\Brac{\exp\paren{\lambda(\eta^2 - \E[\eta^2])}} & = \frac{1}{\sqrt{2 \pi \sigma^2}} \cdot \int_{-\infty}^{\infty} \exp\paren{\lambda z^2 - \lambda \sigma^2} \cdot \exp\paren{\frac{-z^2}{2\sigma^2}} \rd z \\
%         & = \frac{\exp\paren{-\lambda \sigma^2}}{\sqrt{2 \pi \sigma^2}} \cdot \int_{-\infty}^{\infty} \exp\paren{ -z^2 \paren{\frac{1}{2\sigma^2} - \lambda} }\rd z \\
%         & = \exp\paren{-\lambda \sigma^2} \cdot \frac{\sqrt{\frac{2 \pi \sigma^2}{1-2\lambda\sigma^2} }}{\sqrt{2 \pi \sigma^2}}, & {\text{for } \lambda \leq \frac{1}{2 \sigma^2}}\\
%         &  = \frac{\exp\paren{-\lambda \sigma^2}}{\sqrt{1-2\lambda\sigma^2}} \leq \exp\paren{4 \sigma^4\lambda^2/2}, & \text{for } \abs{\lambda} \leq \frac{1}{4 \sigma^2} \mper & \qedhere
%     \end{align*}
% \end{proof}
We also require the following concentration inequality for a sum of sub-exponential random variable:
\begin{fact}[{\citep[Equation 2.18]{wainwright_2019}}] \label{fact:sub_exp_conc}
Consider independent random variables $\gamma_1, \ldots, \gamma_k$, where, $\forall j\in 1, \ldots,k$, $\gamma_j$ is sub-exponential with parameters $(\nu_j,\alpha_j)$. Let $\nu_{*} = \sqrt{\sum_{j=1}^k \nu_j^2}$ and $\alpha_{*} = \max \set{\alpha_1,\dots,\alpha_k}$. Then we have:
\begin{equation*}
    \Pr{ \sum_{j=1}^{k} \paren{\gamma_j - \E[\gamma_j]} \geq t } \leq 
    \begin{cases}
        \exp\paren{ \frac{-t^2}{2 \nu_{*}^2} } & \text{ for }0 \leq t \leq \frac{\nu_{*}^2}{\alpha_{*}} \mcom \\
        \exp\paren{ \frac{-t}{2 \alpha_{*}} } & \text{ for }t > \frac{\nu_{*}^2}{\alpha_{*}} \mper
    \end{cases}
\end{equation*}
\end{fact}

\begin{proof}[Proof of high-probability bound of \Cref{thm:synth_data}]
    Recalling the proof of the expectation bound of \Cref{thm:synth_data} from \Cref{sec:synth_data}, it suffices to bound $E =\sum_{j=1}^{k} \frac{1}{j^2} \left(\hat{m}_j(p) - \langle \tilde{p}, T_j \rangle\right)^2$ with high probability. 
    % To do so, we first recall \Cref{eq:basic_exp_bound}:
    % \begin{align} \tag{\ref{eq:basic_exp_bound}}
    % \E[E] = \E\left[\sum_{j=1}^{k} \frac{1}{j^2} \eta_j^2\right] = \sum_{j=1}^{k} \frac{1}{j^2} \E\left[\eta_j^2\right] = \sum_{j=1}^{k} \frac{1}{j^2}\cdot j\sigma^2 \leq (1+\log k)\sigma^2, 
    % \end{align}
    % where $\sigma^2$ is as defined in \Cref{alg:dp_algorithm}.
    Let $\gamma_j =\eta_j^2/j^2$, where $\eta_j \sim \mathcal{N}(0,j\sigma^2)$ is as in \Cref{alg:dp_algorithm}. Then recall that $E = \sum_{j=1}^{k} \gamma_j$. 
    
    From \Cref{lem:sub_exp_of_gau_sq}, $\gamma_j$ is a sub-exponential random variable with parameter $\paren{2\sigma^2/j, 4\sigma^2/j}$. We can then apply \Cref{fact:sub_exp_conc}, for which we have $\nu_* = \sqrt{\sum_{j=1}^{{k}}{4 \sigma^4}/{j^2}} \leq {2 \pi \sigma^2}/{\sqrt{6}}$ and $\alpha_{*} = 4 \sigma^2$. For any failure probability $\beta\in (0,1/2)$, setting $t = 8 \log(1/\beta) \sigma^2$, we conclude that:
    \begin{align*}
    % \label{eq:whp_exp_bound}
    \Pr{ E - \E[E]  \geq 8 \log\paren{{1}/{\beta}} \sigma^2} \leq  \beta \mper
    \end{align*}
    Recalling from \Cref{eq:basic_exp_bound} that $\E[E] \leq (1+\log k) \sigma^2$, we conclude that $E \leq 8 \log\paren{{1}/{\beta}} \sigma^2 + (1+\log k) \sigma^2$ with probability at least $1-\beta$.
    
    The rest of the details follow as before. In particular, as in \Cref{eq:final_triangle}, we can bound:
    \begin{align*}
    W_1(p,q) \leq \sqrt{2}\Gamma + \frac{36}{k} + \frac{1}{2\ceil{\e n}},
    \end{align*}
    where $\Gamma \leq \sqrt{2E}$.
    Plugging in $k = \ceil{2\e n}$ (as chosen in \Cref{alg:dp_algorithm}) and recalling that $\sigma^2 = {\frac{16}{\pi}(1 + \log {k})\ln(1.25 / \delta)}/{(\e^2 n^2)}$, we conclude that with probability $\geq 1-\beta$, for a fixed constant $c$,
    \begin{align*}
        W_1(p,q) \leq c\left(\frac{\sqrt{\log (\e n) + \log(1/\beta)} \sqrt{\log (\e n)\log({1/\delta})}}{\e n}\right)\mper &\qedhere
    \end{align*}
\end{proof}

%\printbibliography

\section{Spectral Density Estimation Lower Bound}\label{app:sde_lower_bound}

In this section, we provide a lower bound on the number of matrix-vector multiplications required for spectral density estimation, showing that our upper bound in  \Cref{cor:sde} is optimal up to logarithmic factors. We first need the following theorem from \cite{Meyer:2024}, which shows that estimating the trace of a positive semi-definite matrix $A$ to within a multiplicative error of $(1 \pm \e)$ requires $\Omega(1/\e)$ matrix-vector multiplications with $A$. 

\begin{theorem}[Restated {\citep[Theorem 17]{Meyer:2024}}, see also {\citep[Theorem 17]{SwartworthWoodruff:2023}}] \label{thm:trace_lb}
    %For all $\delta>0$ and $\e = \bigo{1/\sqrt{\log(1/\delta)}}$, a
    Any algorithm that is given matrix-vector multiplication access to a positive semi-definite (PSD) input matrix $A \in \R^{n \times n}$ with $\norm{A}_2 \leq 1$, $n/4 \leq \tr(A) \leq n$ and succeeds with probability at least $2/3$ in outputting an estimate $\tilde t$ such that $\abs{\tilde t - \tr(A)} \leq \e \cdot \tr(A)$ requires $\Omega\paren{\frac{1}{\e}}$ matrix-vector multiplications with $A$.
\end{theorem}
As a corollary of this result, we obtain the following lower bound, which shows that \Cref{cor:sde} is tight up to $\log(1/\epsilon)$ factors:
\begin{corollary}
    Any algorithm that is given matrix-vector multiplication access to a symmetric matrix $A \in \R^{n \times n}$ with spectral density $p$ and $\norm{A}_2 \leq 1$ requires $\Omega\paren{\frac{1}{\e}}$ matrix-vector multiplications with $A$ to output a distribution $q$ such that $W_1(p,q) \leq \epsilon$ with probability at least $2/3$.
\end{corollary}
\begin{proof}
The proof is via a direct reduction. 
Consider a PSD matrix $A$ with $\norm{A}_2 \leq 1$, $n/4 \leq \tr(A) \leq n$, and spectral density $p$. Suppose we have a spectral density estimate $q$ of $p$ such that $W_1(p,q) \leq \e/4$. We claim that $\tilde{t} = n\cdot \int_{-1}^1 xq(x)\,dx$ yields a relative error approximate to $A$'s trace, implying that computing such a $q$ requires $\Omega({1}/{\e})$ matrix-vector products by \Cref{thm:trace_lb}.

In particular, applying Kantorovich-Rubinstein duality (\Cref{fact:w1_dual}) with the $1$-Lipschitz functions $f(x) = x$ and $f(x) = -x$, we have that:
\begin{align}
\label{eq:trace_before_scale}
\int_{-1}^1 xp(x)\,dx - \int_{-1}^1 xq(x)\,dx &\leq \e/4 & &\text{and} & \int_{-1}^1 xq(x)\,dx - \int_{-1}^1 xp(x)\,dx &\leq \e/4.
\end{align}
We have that $\int_{-1}^1 xp(x)\,dx = \frac{1}{n}\tr(A)$. So \eqref{eq:trace_before_scale} implies that $\tilde{t} = n\cdot \int_{-1}^1 xq(x)\,dx$ satisfies:
\begin{align*}
|\tilde{t} - \tr(A)| \leq n\cdot \e/4 \leq \e \cdot \tr(A). &\qedhere
\end{align*}
\end{proof}

\end{document}

%% file: macros.tex
\usepackage[  textheight=9in, textwidth=6.5in, top=1in, headheight=14pt, headsep=25pt, footskip=30pt]{geometry}
\usepackage[T1]{fontenc}
\usepackage[final]{microtype}
\usepackage[english]{babel}
\usepackage{amsmath,amssymb,amsfonts,mathtools,amsthm}
\usepackage{thmtools}
\usepackage{thm-restate}
\usepackage{bbm}
\usepackage{color,graphicx}
\usepackage{verbatim}
\usepackage{algorithm}
\usepackage{algpseudocode}
\usepackage{setspace}
\usepackage{caption}
\usepackage{subcaption}
\usepackage{fancyhdr}
\usepackage[colorinlistoftodos]{todonotes}
\usepackage{rotating}
\usepackage{booktabs}
\usepackage{enumitem, comment}
\usepackage{soul} 
\usepackage{dsfont}
\usepackage{bookmark}
\usepackage{csquotes}
\usepackage{fnpct} 
\usepackage{float}
\usepackage{parskip}
\usepackage{cleveref}
\usepackage{lmodern}
\usepackage{xspace}   
\usepackage{ifdraft}
\usepackage{nicefrac}
\usepackage{mdframed}
\newmdtheoremenv{theomd}{Theorem}

\newtheorem{theorem}{Theorem}
\newtheorem*{theorem*}{Theorem}

\newtheorem{claim}[theorem]{Claim}
\newtheorem*{claim*}{Claim}

\newtheorem{proposition}[theorem]{Proposition}
\newtheorem*{proposition*}{Proposition}
\newtheorem{lemma}[theorem]{Lemma}
\newtheorem*{lemma*}{Lemma}
\newtheorem{corollary}[theorem]{Corollary}

\newtheorem*{conjecture*}{Conjecture}

\newtheorem{fact}[theorem]{Fact}

\newtheorem*{hypothesis*}{Hypothesis}
\theoremstyle{definition}
\newtheorem{definition}[theorem]{Definition}
\newmdtheoremenv{defmd}[theorem]{Definition}

\newtheorem{remark}[theorem]{Remark}

\usepackage{prettyref}
\newcommand{\savehyperref}[2]{\texorpdfstring{\hyperref[#1]{#2}}{#2}}

\newrefformat{eq}{\savehyperref{#1}{\textup{(\ref*{#1})}}}
\newrefformat{lem}{\savehyperref{#1}{Lemma~\ref*{#1}}}
\newrefformat{def}{\savehyperref{#1}{Definition~\ref*{#1}}}
\newrefformat{thm}{\savehyperref{#1}{Theorem~\ref*{#1}}}
\newrefformat{cor}{\savehyperref{#1}{Corollary~\ref*{#1}}}
\newrefformat{cha}{\savehyperref{#1}{Chapter~\ref*{#1}}}
\newrefformat{sec}{\savehyperref{#1}{Section~\ref*{#1}}}
\newrefformat{app}{\savehyperref{#1}{Appendix~\ref*{#1}}}
\newrefformat{tab}{\savehyperref{#1}{Table~\ref*{#1}}}
\newrefformat{fig}{\savehyperref{#1}{Figure~\ref*{#1}}}
\newrefformat{hyp}{\savehyperref{#1}{Hypothesis~\ref*{#1}}}
\newrefformat{alg}{\savehyperref{#1}{Algorithm~\ref*{#1}}}
\newrefformat{sdp}{\savehyperref{#1}{SDP~\ref*{#1}}}
\newrefformat{rmk}{\savehyperref{#1}{Remark~\ref*{#1}}}
\newrefformat{item}{\savehyperref{#1}{Item~\ref*{#1}}}
\newrefformat{step}{\savehyperref{#1}{step~\ref*{#1}}}
\newrefformat{conj}{\savehyperref{#1}{Conjecture~\ref*{#1}}}
\newrefformat{fact}{\savehyperref{#1}{Fact~\ref*{#1}}}
\newrefformat{prop}{\savehyperref{#1}{Proposition~\ref*{#1}}}
\newrefformat{claim}{\savehyperref{#1}{Claim~\ref*{#1}}}
\newrefformat{relax}{\savehyperref{#1}{Relaxation~\ref*{#1}}}
\newrefformat{red}{\savehyperref{#1}{Reduction~\ref*{#1}}}
\newrefformat{part}{\savehyperref{#1}{Part~\ref*{#1}}}
\newrefformat{prob}{\savehyperref{#1}{Problem~\ref*{#1}}}
\newrefformat{ass}{\savehyperref{#1}{Assumption~\ref*{#1}}}
\newrefformat{ques}{\savehyperref{#1}{Question~\ref*{#1}}}

\newrefformat{ex}{\savehyperref{#1}{Example~\ref*{#1}}}

\usepackage{bm}
\newcommand{\im}{\mathrm{i}\mkern1mu}

\hypersetup{colorlinks,linkcolor=blue,filecolor = blue, citecolor = violet, urlcolor  = magenta}

\usepackage[backend=biber, style=trad-alpha,
sorting=nyt, natbib=true, backref=true, eprint=false, url=false, doi = false, isbn = false, maxbibnames = 99, maxcitenames = 2, maxalphanames = 5]{biblatex}

\DeclareSourcemap{
  \maps[datatype=bibtex, overwrite]{
    \map{
      \step[fieldset=editor, null]
    }
  }
}

\DeclareSourcemap{
  \maps[datatype=bibtex]{
    \map[overwrite]{
      \step[fieldsource=eprint, final]
      \step[fieldset=doi, null]
    }
  }
}

\DeclareSourcemap{
  \maps[datatype=bibtex]{
    \map[overwrite]{
      \step[fieldsource=doi, final]
      \step[fieldset=url, null]
    }
  }
}

\DeclareSourcemap{
  \maps[datatype=bibtex]{
    \map[overwrite]{
      \step[fieldsource=eprint, final]
      \step[fieldset=url, null]
    }
  }
}

\DeclareSourcemap{
  \maps[datatype=bibtex, overwrite]{
    \map{
      \step[fieldset=address, null]
      \step[fieldset=location, null]
      \step[fieldset=editor, null]
    }
  }
}

\newcommand{\mper}{\,.}
\newcommand{\mcom}{\,,}

\newcommand{\paren}[1]{\left(#1 \right )}

\newcommand{\brac}[1]{[#1 ]}
\newcommand{\Brac}[1]{\left[#1\right]}

\newcommand{\set}[1]{\left\{#1\right\}}
\newcommand{\Set}[1]{\left\{#1\right\}}

\newcommand{\abs}[1]{\left\lvert#1\right\rvert}

\newcommand{\ceil}[1]{\lceil #1 \rceil}

\newcommand{\norm}[1]{\left\lVert#1\right\rVert}

\newcommand{\nnz}[1]{\mathrm{nnz}(#1)}
\newcommand{\tr}{\mathrm{tr}}

\newcommand{\defeq}{\stackrel{\textup{def}}{=}}
\newcommand{\seteq}{\mathrel{\mathop:}=}

\newcommand{\inprod}[1]{\left\langle #1\right\rangle}

\newcommand{\Z}{{\mathbb Z}}

\newcommand{\N}{{\mathbb Z}_{\geq 0}}
\newcommand{\Np}{{\mathbb Z}_{> 0}}
\newcommand{\R}{\mathbb R}

\newcommand{\Esymb}{\mathbb{E}}
\newcommand{\Psymb}{\mathbb{P}}

\DeclareMathOperator*{\E}{\Esymb}

\DeclareMathOperator*{\ProbOp}{\Psymb}

\newcommand{\Prob}{\mathbb{P}}

\renewcommand{\Pr}[1]{\ProbOp\Brac{#1}}

\newcommand{\e}{\epsilon}

\let\e\epsilon

\newcommand{\cC}{\mathcal C}
\newcommand{\cD}{\mathcal D}

\newcommand{\cG}{\mathcal G}

\newcommand{\cN}{\mathcal N}

\newcommand{\bigO}{O}
\newcommand{\bigo}[1]{\bigO\!\left(#1\right)}
\newcommand{\tbigO}{\tilde{O}}
\newcommand{\tbigo}[1]{\tbigO\!\left(#1\right)}

\DeclareMathOperator*{\argmin}{argmin} 
\DeclareMathOperator*{\argmax}{argmax} 
\newcommand{\poly}{{\sf poly}}

\newcommand{\rd}{{\sf d}}

\usepackage{dsfont}
\usepackage{bbm}

\newcommand{\T}{\Bar{T}}
\newcommand{\for}{\text{ for }}

\newcommand{\lp}{\mathrm{LP}}

\newcommand{\hobs}[1]{h_{#1}^{\text{obs}}}
\newcommand{\wt}{\mathsf{W}}
\newcommand{\coeff}{\mathsf{C}}
\newcommand{\x}{\mathbf{x}}
\newcommand{\normsum}{\mathsf{S}}

%% file: bib_macros.tex
  \usepackage{nth}
  \usepackage{intcalc}

  \newcommand{\cAAAI}[1]{AAAI\ Conference\ on\ Artificial (AAAI)}